\documentclass[twocolumn]{autart}    

\usepackage{amsfonts}
\usepackage{amsmath}
\usepackage{cases}
\usepackage{makecell}
\usepackage{booktabs}
\usepackage{caption}
\usepackage{url}
\usepackage{bm}
\usepackage{amsfonts,mathrsfs,bbm,amsmath,stmaryrd,amssymb,pifont}%
\usepackage{amscd,graphicx,array,dsfont,texdraw,tikz}%
\usepackage{bbm,mathrsfs}%
\usepackage{epstopdf}
\usepackage{algorithm}
\usepackage{algorithmicx}
\usepackage{algpseudocode}
\usepackage{afterpage}

\allowdisplaybreaks[4]
\newtheorem{theorem}{Theorem}
\newtheorem{lemma}{Lemma}
\newtheorem{proposition}{Proposition}
\newtheorem{remark}{Remark}
\newtheorem{assumption}{Assumption}

\newtheorem{corollary}{Corollary}

\DeclareMathOperator*{\bigtimes}{\text{\large $\times$}}

\definecolor{blue}{rgb}{0,0,0}

\usepackage[round]{natbib}

\begin{document}

\begin{frontmatter}
\title{Distributed Zeroth-Order Policy Gradient for Networked Multi-agent Reinforcement Learning from Human Feedback} 

\thanks[footnoteinfo]{This work was supported in part by the National Key Research and Development Program of China under Grant 2022ZD0120002 and in part by the National Natural Science Foundation of China under Grant 62233004. }
\thanks[Corresponding author]{Corresponding author.}

\author[Singapore]{Pengcheng Dai}\ead{Jldaipc@163.com},
\author[China1,Corresponding author]{He Wang}\ead{wanghe91@seu.edu.cn},
\author[USA]{Dongming Wang}\ead{wdong025@ucr.edu},
\author[China2]{Jian Qin}\ead{qinjian@scu.edu.cn},
\author[China1,China3,Corresponding author]{Wenwu Yu}\ead{wwyu@seu.edu.cn},  

\address[Singapore]{Singapore University of Technology and Design, Singapore 487372, Singapore}  
\address[China1]{Southeast University, Nanjing 211189, China}
\address[USA]{University of California, Riverside, CA 92521, USA}
\address[China2]{Sichuan University, Chengdu, 610065, China}
\address[China3]{Purple Mountain Laboratories, Nanjing 211102, China}


\begin{abstract}                          
We study {\color{blue}a networked multi-agent reinforcement learning (NMARL) problem with human feedback in an infinite-horizon setting, where agents interact over an underlying network with localized state dependencies and} aim to collaboratively maximize the average discounted return.
{\color{blue}Existing approaches with preference
feedback are primarily developed for single-agent settings and rely on centralized training, which limits their scalability and applicability to large-scale networked multi-agent systems.}
To address this, we introduce a novel human feedback mechanism based on spatiotemporally truncated trajectories, defined as $H$-horizon trajectory pairs aggregated over each agent's $\kappa$-hop neighborhood.
Building on this, we develop a distributed zeroth-order policy gradient algorithm, where each agent estimates its local policy gradient using human preference feedback generated from both the current joint policy and a perturbed {\color{blue}joint policy} drawn from zero-mean Gaussian {\color{blue}distribution}.
Specifically, the algorithm is {\color{blue}fully} distributed, as the feedback received by each agent depends solely on the state-action information within its $\kappa$-hop neighborhood and does not require explicit reward signals or centralized control.
We further rigorously establish that the proposed algorithm converges to an $\epsilon$-stationary point with polynomial sample complexity.
Finally, simulation results in a stochastic GridWorld environment and {\color{blue}a predator-prey environment} further demonstrate that the effectiveness {\color{blue}and scalability} of the proposed algorithm in achieving collaborative optimization based solely on human preference feedback.
\end{abstract}

\end{frontmatter}

\section{Introduction}
Reinforcement learning (RL) has achieved remarkable success in a wide range of sequential decision-making problems, including smart grids~\citep{DaiTII2020,LiTCYB2019,DaiTCYB2022}, intelligent transportation~\citep{ChuTITS2020,WangTCYB2021,DaiTII2024}, wireless communications~\citep{TanTWC2021,MengTWC2020}, and recommendation systems~\citep{AfsarSurvey2022,LinSurvey2022}, etc.
Traditional RL methods~\citep{Sutton1998,Minhnature2015,Sha2022} rely on a well-defined reward signal at each time step, which is typically assumed to be known and consistent with human intentions.
\par
\begin{table*}[!ht]
\centering
\caption{{\color{blue}Comparison with existing algorithms under networked information structures.}}\label{tab:comparison}
\scriptsize
{\color{blue}\begin{tabular}{cccc}
\toprule
References & Our Algorithm~\ref{Distributedzeroth-orderpolicygradientalgorithm} &
\makecell{Scalable algorithms \\ \citep{QuCLDC2020,QuNIPS2020}}
&
\makecell{Decentralized/distributed algorithms \\ \citep{Zhangkaiqing2018,DaiTAC2025}} \\
\midrule
State-action pairs&  $\kappa$-hop neighbors  & $\kappa$-hop neighbors & Global \\
True rewards & None & Local & Local \\
$Q$-values/parameters & None & $\kappa$-hop neighbors & 1-hop neighbors \\
Human preferences & $\surd$ & $\bigtimes$ & $\bigtimes$\\
\bottomrule
\end{tabular}}
\end{table*}
However, {\color{blue}in many real-world control and decision-making applications}, the design
of a high-quality reward function is highly non-trivial and often impractical.
This fundamental challenge has motivated the development of {\color{blue}preference-based reinforcement learning~(PbRL)~\citep{Christiano2017,Wirth2017}},
where the reward signals are implicitly derived from {\color{blue}human preferences expressed over trajectory comparisons, rather than being explicitly specified.}
{\color{blue}A common instantiation of this paradigm can be described as a three-stage training pipeline.}
(i) Policy pre-training: the agent pre-trains a policy using standard RL techniques or behavioral cloning;
(ii) Reward inference: the agent collects pairwise trajectory comparisons annotated by human evaluators and trains a reward model {\color{blue}via} maximum likelihood estimation under a preference model, such as the Bradley-Terry model~\citep{Bengs2021};
{\color{blue}and} (iii) Policy optimization: the agent optimizes its policy {\color{blue}using the learned reward model via RL algorithms} such as PPO~\citep{Schulman2017}.
{\color{blue}Despite the empirical success of this pipeline},
the reward inference stage {\color{blue}remains a bottleneck: it is sample-inefficient}, susceptible to model missspecification, {\color{blue}and difficult to evaluate without ground-truth reward signals.}
To address these {\color{blue}limitations,
some preference-based policy optimization methods}~\citep{Rafailov2023,Zhangarxiv2024} have been proposed {\color{blue}to bypass} the reward modeling stage {\color{blue}and directly optimize the policy from preference
feedback}.
{\color{blue}In particular, DPO algorithm~\citep{Rafailov2023} showed that,} under the Bradley-Terry preference model and a KL-regularized RL objective, the reward function can be reparameterized in terms of the policy and a reference policy, yielding a tractable loss function directly from preference data.
{\color{blue}However, its theoretical guarantees are established under a contextual bandit assumption and do not directly extend to sequential decision-making problems.}
{\color{blue}For general RL problems}, including stochastic state transitions or infinite state-action space, a zeroth-order policy gradient approximation {\color{blue}method with theoretical guarantee was proposed} in~\citep{Zhangarxiv2024}.
{\color{blue}This approach} enables policy optimization directly
from human preference feedback, without requiring explicit reward signals or reward model inference.
\par
{\color{blue}While the preference-based policy optimization methods have achieved significant success in single-agent settings, many real-world tasks inherently require coordination among multiple agents operating under local observations and communication constraints, which are naturally modeled as networked
multi-agent reinforcement learning (NMARL)~\citep{QuCLDC2020,QuNIPS2020}.}
{\color{blue}In multi-robot navigation~\citep{Huang2024}, for instance, robots perceive only local regions and communicate solely with neighbors, yet must jointly navigate to their respective goals while avoiding inter-agent collisions.
Although individual objectives such as goal reaching and collision avoidance admit straightforward definitions, the overall coordination quality that reflects trade-offs among travel efficiency, path smoothness, and inter-agent cooperation, is inherently subjective and difficult to capture through a single scalar reward.
Similarly, in multi-intersection traffic signal
control~\citep{Zhao2026},
each intersection controller observes only local traffic conditions and communicates with neighboring intersections, yet the system must jointly optimize traffic flow across the entire network.
The overall traffic coordination quality which involves balancing throughput, fairness across intersections, and passenger comfort, is once again difficult to quantify using a single scalar reward. As a result, human preference feedback serves as a more natural supervisory signal in such networked settings.}
{\color{blue}However, the underlying networked information structure introduces additional challenges that are absent in single-agent setting, including how to properly define human preference feedback in a networked setting and how to ensure cooperative learning among agents under communication constraints.}
\par
{\color{blue}Recent works on NMARL have developed scalable actor-critic algorithms~\citep{QuCLDC2020,QuNIPS2020} and decentralized/distributed RL algorithms~\citep{Zhangkaiqing2018,DaiTAC2025} that exploit the localized interaction structure of networked systems, these methods universally assume access to true local reward signals and rely on $Q$-value or $Q$-parameter sharings among neighboring agents.
As summarized in Tables~\ref{tab:comparison}, our objective fundamentally deviates from this paradigm: explicit reward signals are completely eliminated and substituted with human preference feedback, and the requirement for $Q$-value or $Q$-parameter sharing is concurrently removed.}
\par
{\color{blue}This distinction introduces new algorithmic and theoretical challenges
that are not addressed by existing scalable or distributed RL frameworks.
Specifically, three fundamental challenges must be overcome.}
(i)~\textbf{Human feedback mechanism design.} In the absence of explicit reward signals, {\color{blue}a principled feedback mechanism must be designed to enable} each agent to leverage localized preference information {\color{blue}toward} the global cooperative objective, while remaining compatible with distributed execution under communication constraints;
(ii)~\textbf{Policy update under uncertain and noisy
preferences.} Since human feedback constitutes the sole learning signal {\color{blue}and} is inherently stochastic, and limited to sampled trajectory comparisons, trajectory pair generation strategies must be carefully designed to elicit informative preferences, and such uncertain signals must be incorporated into the policy update process in a stable and efficient manner;
(iii)~\textbf{Theoretical analysis in distributed setting.} Existing convergence results for {\color{blue}preference-based policy optimization} are primarily {\color{blue}established} for single-agent, finite-horizon {\color{blue}settings}~\citep{Zhangarxiv2024}, and do not directly extend to NMARL, where agents access only limited neighborhood state-action information rather than global information.
{\color{blue}A new analytical framework is thus required to rigorously characterize algorithmic convergence in such distributed settings.}
These challenges motivate this paper to investigate a distributed
{\color{blue}preference-based policy optimization} framework for general NMARL over an infinite horizon, where agents exhibit localized state-action interactions and aim to collaboratively maximize the average discounted return based solely on human preference feedback, without relying on explicit reward signals {\color{blue}or centralized coordination}.
The central question addressed is:
\emph{How can a provably efficient distributed {\color{blue}preference-based method} be
designed for general NMARL without explicit reward signals?}
\par
To address the aforementioned challenges, we introduce a novel human feedback mechanism grounded in the concept of spatiotemporally truncated trajectories.
By leveraging this, we incorporate zeroth-order optimization techniques for policy updates and develop a distributed zeroth-order policy gradient algorithm.
Additionally, we rigorously establish the convergence of the proposed algorithm to an $\epsilon$-stationary point. The main contributions of this work are summarized as follows.
\begin{enumerate}
\item \textbf{Human feedback mechanism design.} We propose a novel preference feedback mechanism based on spatiotemporally truncated trajectories,
    defined as $H$-horizon trajectory pairs constructed from each agent's $\kappa$-hop neighborhood.
    Each agent collects such spatiotemporally trajectories and {\color{blue}queries human evaluators for preference signals}
    that implicitly reflect cooperative performance within its neighborhood, while remaining {\color{blue}fully compatible with} distributed execution.
%

\item \textbf{Policy update under uncertain and noisy preferences.} {\color{blue}In the concrete policy update process, agents independently sample Gaussian perturbations to construct a perturbed joint policy.}
    Each agent {\color{blue}then collects two spatiotemporally truncated trajectories under the current and perturbed policies, and elicits human preference feedback through pairwise comparisons.}
    Each agent then uses the received signals
    to approximate its local policy gradient via a zeroth-order estimator, {\color{blue}enabling fully distributed} policy update without reliance on explicit reward signals or centralized coordination.
  \item \textbf{Theoretical analysis in distributed setting.}
  {\color{blue}We propose a distributed zeroth-order policy gradient algorithm for general NMARL problems, and provide rigorous theoretical guarantees for its convergence behavior.}
  {\color{blue}Specifically, we establish a rigorous characterization of} the first-order gradient bias and variance induced by Gaussian perturbations, spatiotemporal truncation, and stochastic preference feedback.
  {\color{blue}Building on this,} we prove $\epsilon$-stationary convergence of the proposed algorithm, which does not require the reward signal {\color{blue}and is applicable to fully distributed multi-agent settings.}
\end{enumerate}

\par
The rest of this paper is organized as follows. Section~\ref{SectionIII} introduces the NMARL model and its related knowledge.
Section~\ref{SectionIV} introduces a novel human feedback mechanism and design a distributed zeroth-order policy gradient
algorithm.
Section~\ref{SectionofConvergenceanalysis} provides the convergence analysis of the proposed algorithm.
The performances of the proposed algorithm in a stochastic GridWorld environment and {\color{blue}a predator-prey environment} are represented in Section~\ref{SectionSimulations}.
Finally, Section~\ref{SectionVConclusions} discusses the
concluding remarks and potential directions for future work.

\par
\textbf{Notations}: $\mathbb{R}$ is the set of reals and
$\mathbb{R}^{N}$ denotes the $N$-dimensional real vector set.
For any $x,y\in\mathbb{R}^{N}$, $\|x\|_{1}$ and $\|x\|$ represent the standard $\mathcal{L}_{1}$-norm and $\mathcal{L}_{2}$-norm of $x$, respectively.
$\|x-y\|_{\mathrm{TV}}$ denotes the total variation distance between $x$ and $y$.
$x^{\top}$ is the transpose of $x$.
Let $\Delta \in \mathbb{R}$, for scalars $a\in\mathbb{R}$, we define $\mathrm{trim}[a|\Delta]$ as $\min\{\max\{a, \Delta\}, 1 - \Delta\}$.
Similarly, for a vector $\bm{v} \in \mathbb{R}^N$, $\mathrm{trim}[\bm{v}|\Delta]$ denotes the vector obtained by applying the trimming operation element-wise.
$\mathbb{E}[\cdot]$ denotes the expectation taken over all involved variables.

\section{NMARL problem}\label{SectionIII}
In this section, we introduce the model of the NMARL problem and its related knowledge.
\subsection{The model of NMARL}\label{ThemodelofNMARL}
The NMARL problem is formally defined as $\big(\mathcal{G}(\mathcal{N},\mathcal{E}),$ $\{\mathcal{S}_{i}\}_{i\in\mathcal{N}},\{\mathcal{A}_{i}\}_{i\in\mathcal{N}},\{\mathcal{P}_{i}\}_{i\in\mathcal{N}},\{r_{i}\}_{i\in\mathcal{N}},\{\pi_{i}\}_{i\in\mathcal{N}},\gamma\big)$, where the detailed descriptions of each element are summarized as follows.
\par
\textbf{{\color{blue}Underlying} network among agents}: $\mathcal{G}\big(\mathcal{N},\mathcal{E}\big)$ is {\color{blue}an underlying} network among agents, where $\mathcal{N}=\{1,\cdots,N\}$ denotes the set of agents and the edge set $\mathcal{E}$ specifies the links between them.
Define $\mathcal{N}_{i}=\{j|e_{ij}\in\mathcal{E}\}$ as the neighborhood of agent $i$.
For {\color{blue}any} integer $\kappa\geq1$, we let $\mathcal{N}^{\kappa}_{i}$ denote the $\kappa$-hop neighborhood of agent $i$, i.e., the agents whose graph distance to agent $i$ is less than or equal to $\kappa$, including agent $i$ itself.
Let $\mathcal{N}^{\kappa}_{i,-j} = \mathcal{N}^{\kappa}_{i} \setminus \{j\}$ represent the set of agent $i$'s $\kappa$-hop neighbors excluding agent $j$.
Moreover,  $\mathcal{N}^{\kappa}_{-i}=\mathcal{N}\setminus\mathcal{N}^{\kappa}_{i}$ represents the set of agents other than the set $\mathcal{N}^{\kappa}_{i}$.
\par
\textbf{State and action}: $\mathcal{S}_{i}$ and $\mathcal{A}_{i}$ represent a finite state space and a finite action space of agent $i$, respectively.
$s_{i} \in \mathcal{S}_{i}$ and $a_{i} \in \mathcal{A}_{i}$ {\color{blue}denote the local state and the local action of agent $i$, respectively.}
Let $\bm{s} = (s_{1}, \cdots, s_{N}) \in \bm{\mathcal{S}} = \prod_{i=1}^{N} \mathcal{S}_{i}$ represent the global state and $\bm{a} = (a_{1}, \cdots, a_{N}) \in \bm{\mathcal{A}} = \prod_{i=1}^{N} \mathcal{A}_{i}$ denote the global action.
Similarly,  $s_{\mathcal{N}^{\kappa}_{i}}\in\mathcal{S}_{\mathcal{N}^{\kappa}_{i}}=\prod_{j\in\mathcal{N}^{\kappa}_{i}}\mathcal{S}_{j}$ and $a_{\mathcal{N}^{\kappa}_{i}}\in\mathcal{A}_{\mathcal{N}^{\kappa}_{i}}=\prod_{j\in\mathcal{N}^{\kappa}_{i}}\mathcal{A}_{j}$ represent the states and actions of agent $i$'s $\kappa$-hop neighbors, respectively.
\par
\textbf{State transition probability function}:
$\mathcal{P}_{i}(s'_{i}|s_{\mathcal{N}_{i}},a_{i})$ represents the state transition probability function for agent $i$, specifying that the subsequent state of agent $i$ is determined by both the current states of its neighbors and its own local action.
{\color{blue}Specifically, the} global state transition probability function for agents is defined as
$\bm{\mathcal{P}}(\bm{s}'|\bm{s},\bm{a})=\prod_{i=1}^{N}\mathcal{P}_{i}(s'_{i}|s_{\mathcal{N}_{i}},a_{i})$.
\par
\textbf{Reward functions}:
$r_{i}(s_{i},a_{i})$ represents the local reward function {\color{blue}of} agent $i$, which is usually handcrafted by domain experts to ensure it aligns with human interests.
\par
\textbf{Parameterized policy}: Each agent $i\in\mathcal{N}$ is associated with a class of localized policies $\pi_{i}(a_{i}|s_{i},\theta_{i})$ with $\theta_{i}\in\mathbb{R}^{d_{i}}$ being the policy parameter.
Let $\bm{\pi}_{\bm{\theta}}(\bm{a}|\bm{s}) = \prod_{i=1}^{N} \pi_{i}(a_{i}|s_{i}, \theta_{i})$ denote the joint policy of the agents, where $\bm{\theta} = (\theta^{\top}_{1}, \cdots, \theta^{\top}_{N})^{\top} \in \mathbb{R}^{d_{\mathrm{tot}}}$ represents the vector of joint policy {\color{blue}parameter} and $d_{\mathrm{tot}} = \sum_{i=1}^{N} d_{i}$ denotes the total dimensionality of the parameter space.
\par
\textbf{Discount factor}: $\gamma\in(0,1)$ is the discount factor.
\par
In the NMARL problem, the objective of agents is to find a joint policy parameter $\bm{\theta}$ to maximize the discounted average cumulative rewards, i.e.,
\begin{align}\label{theobjectivefunction}
\max_{\bm{\theta}}J(\bm{\theta})\triangleq&\mathbb{E}_{\bm{s}\sim\bm{\rho}}\Big[\frac{1}{N}\sum_{t=0}^{\infty}\sum_{i=1}^{N}\gamma^{t}r_{i,t}\Big|\bm{s}_{0}=\bm{s},\notag\\
&\bm{a}_{t}\sim\bm{\pi_{\theta}}(\cdot|\bm{s}_{t})\Big],
\end{align}
where $\bm{s}_{t}=(s_{1,t},\cdots,s_{N,t})$ and $\bm{a}_{t}=(a_{1,t},\cdots,a_{N,t})$ represent the global state and global action at time $t$, respectively.
$r_{i,t}=r_{i}(s_{i,t},a_{i,t})$ is the local reward of agent $i$ at time $t$ and $\bm{\rho}$ is the distribution of the initial state $\bm{s}_{0}$.
For the reward signal in NMARL problem, we make the following assumption.
\begin{assumption}\label{theassumptionofreward}
There exists a constant $R>0$ such that the instantaneous reward $r_{i,t}$ of agent $i\in\mathcal{N}$ at time $t\geq0$ satisfies $|r_{i,t}|\leq R$.
\end{assumption}
\par
Assumption~\ref{theassumptionofreward} gives a bounded of the instantaneous rewards of {\color{blue}agents}, which is {\color{blue}common in RL literature~\citep{QuCLDC2020,QuNIPS2020,Zhangkaiqing2018,DaiTAC2025}} and beneficial for ensuring the boundedness of the objective function and characterizing convergence errors.
\subsection{Some knowledge related to NMARL}\label{PreliminaryknowledgerelatedtoNMARL}
In the NMARL problem (\ref{theobjectivefunction}), for any joint policy $\bm{\pi_{\theta}}$, let $Q^{\bm{\pi_{\theta}}}(\bm{s},\bm{a})$ and $Q^{\bm{\pi_{\theta}}}_{i}(\bm{s},\bm{a})$ denote the global $Q$-function and the local $Q$-function of agent $i$, respectively. Specifically, {\color{blue}they} can be formally expressed as
\begin{align}
Q^{\bm{\pi_{\theta}}}(\bm{s},\bm{a})=&\mathbb{E}_{\bm{\pi_{\theta}}}\Big[\frac{1}{N}\sum_{t=0}^{\infty}\sum_{i=1}^{N}\gamma^{t}r_{i,t}\Big|\bm{s}_{0}=\bm{s},\bm{a}_{0}=\bm{a}\Big]\label{thedefinitionofglobalQfunction}
\end{align}
and
\begin{align}
Q^{\bm{\pi_{\theta}}}_{i}(\bm{s},\bm{a})=&\mathbb{E}_{\bm{\pi_{\theta}}}\Big[\sum_{t=0}^{\infty}\gamma^{t}r_{i,t}\Big|\bm{s}_{0}=\bm{s},\bm{a}_{0}=\bm{a}\Big],\label{thedefinitionoflocalQfunction}
\end{align}
{\color{blue}where} the global $Q$-function {\color{blue}$Q^{\bm{\pi_{\theta}}}(\bm{s},\bm{a})$ denotes} the discounted cumulative rewards averaged across all agents.
In contrast, the local $Q$-function $Q^{\bm{\pi_{\theta}}}_{i}(\bm{s},\bm{a})$ accounts solely for the discounted cumulative rewards associated with agent $i$.
{\color{blue}According to the definitions {\color{blue}in} (\ref{thedefinitionofglobalQfunction})-(\ref{thedefinitionoflocalQfunction}),} we directly have
\begin{align}\label{thedecomposeofQfunction}
Q^{\bm{\pi_{\theta}}}(\bm{s},\bm{a})=\frac{1}{N}\sum_{i=1}^{N}Q^{\bm{\pi_{\theta}}}_{i}(\bm{s},\bm{a}).
\end{align}
Define $d^{\bm{\pi_{\theta}}}_{\bm{\rho}}(\bm{s})$ as the discounted state visitation distribution of $\bm{s}\in\bm{\mathcal{S}}$ generated by the joint policy $\bm{\pi_{\theta}}$, {\color{blue}which satisfies}
\begin{align}\label{Thediscountedstatevisitationdistribution}
d^{\bm{\pi_{\theta}}}_{\bm{\rho}}(\bm{s})=(1-\gamma)\sum_{t=0}^{\infty}\gamma^{t}\mathrm{Pr}^{\bm{\pi_{\theta}}}(\bm{s}_{t}=\bm{s}|\bm{s}_{0}\sim\bm{\rho}),
\end{align}
where $\mathrm{Pr}^{\bm{\pi_{\theta}}}(\bm{s}_{t}=\bm{s}|\bm{s}_{0}\sim\bm{\rho})$ is the probability of occurrence of $\bm{s}_{t}=\bm{s}$ at time $t$ under {\color{blue}the} joint policy $\bm{\pi_{\theta}}$ and {\color{blue}the} initial state distribution $\bm{\rho}$.
\par
Based on $Q^{\bm{\pi_{\theta}}}(\bm{s},\bm{a})$ in (\ref{thedefinitionofglobalQfunction}) and $d^{\bm{\pi_{\theta}}}_{\bm{\rho}}(\bm{s})$ in (\ref{Thediscountedstatevisitationdistribution}),
the policy gradient of $J(\bm{\theta})$ with respect to $\theta_{i}$ is described as follows.
\begin{lemma}\label{thelemmaofthepolicygradienttheoreminNMARLCP}
In {\color{blue}the} NMARL problem (\ref{theobjectivefunction}), for any joint policy $\bm{\pi_{\theta}}$, the {\color{blue}gradient} of $J(\bm{\theta})$ with respect to $\theta_{i}$ is represented as
\begin{align}\label{thepolicygradienttheorem}
\nabla_{\theta_{i}}J(\bm{\theta})=&\frac{1}{1-\gamma}\mathbb{E}_{\bm{s}\sim d^{\bm{\pi_{\theta}}}_{\bm{\rho}},\bm{a}\sim\bm{\pi_{\theta}}}\Big[Q^{\bm{\pi_{\theta}}}(\bm{s},\bm{a})\notag\\
&{\color{blue}\times}\nabla_{\theta_{i}}\log\pi_{i}(a_{i}|s_{i},\theta_{i})\Big].
\end{align}
\end{lemma}
\par
Lemma~\ref{thelemmaofthepolicygradienttheoreminNMARLCP}
provides an expression of the policy gradient {\color{blue}$\nabla_{\theta_{i}}J(\bm{\theta})$ in the NMARL problem, which is a direct extension of~\citep{Sutton2000}}.
{\color{blue}Note that} the computation of this expression heavily depends on the global $Q$-function, {\color{blue}which itself requires an explicit reward model to derive.}
However, in various practical tasks including dialogue generation, visual aesthetic assessment, and ethical reasoning, explicit reward models are either unobtainable or difficult to define.
Instead, human {\color{blue}preference} feedback is typically provided as an alternative supervisory signal.
{\color{blue}These challenges fundamentally limit the applicability of traditional reward-based RL methods in these domains.}
Inspired {\color{blue}by} the approaches proposed {\color{blue}in}~\citep{Rafailov2023,Zhangarxiv2024}, which utilize human preference feedback directly to train agents' policies, we aim to develop a distributed algorithm that does not rely on explicit reward signals and is applicable to the NMARL problem.

\section{Distributed zeroth-order policy gradient
algorithm}\label{SectionIV}
In this section, we introduce a novel human feedback mechanism and design a distributed zeroth-order policy gradient
algorithm for {\color{blue}the} NMARL problem, which does not rely on reward signals.
\subsection{A novel human feedback mechanism}\label{HumanFeedback}
In the NMARL problem, each agent $i\in\mathcal{N}$ maintains a local trajectory of horizon length $H$, defined as $\tau_{i}=\{(s_{i,h},a_{i,h})\}_{h=0}^{H-1}$.
{\color{blue}To capture both the spatial and temporal locality of agent interactions, we introduce a notion of spatiotemporally truncated trajectories: for each agent $i$, the spatiotemporally truncated trajectory $\tau_{\mathcal{N}^{\kappa}_{i}}=\{\tau_{j}\}_{j\in\mathcal{N}^{\kappa}_{i}}$ aggregates the local trajectories of agents within its $\kappa$-hop neighborhood $\mathcal{N}^{\kappa}_{i}$ in the underlying network $\mathcal{G}$, thereby truncating the global information both spatially (limited to $\kappa$-hop neighbors) and temporally (limited to horizon $H$).}
Define the {\color{blue}cumulative} reward of each trajectory $\tau_{\mathcal{N}^{\kappa}_{i}}$ as
\begin{align}\label{therewardoftrajectory}
\hat{r}_{i}(\tau_{\mathcal{N}^{\kappa}_{i}})=\frac{1}{N}\sum_{h=0}^{H-1}\gamma^{h}\sum_{j\in\mathcal{N}^{\kappa}_{i}}r_{j}(s_{j,h},a_{j,h}),
\end{align}
which represents the aggregated assessment of the sum of rewards of agents within agent $i$'s $\kappa$-hop neighborhood on the horizon $H$.
{\color{blue}It is worth noting that} the division by $N$ in (\ref{therewardoftrajectory}) is introduced to keep consistency with the objective function (\ref{theobjectivefunction}).
\par
Assume that each agent $i\in\mathcal{N}$ has access to human feedback in the form of a one-bit preference signal $o_{i} \in \{0,1\}$, which is {\color{blue}elicited from a pairwise comparison between two spatiotemporally truncated trajectories} $\tau_{\mathcal{N}^{\kappa}_{i},0}$ and $\tau_{\mathcal{N}^{\kappa}_{i},1}$.
Specifically, the preference signal $o_{i}$ is governed by a known link function $\sigma:\mathbb{R} \rightarrow [0,1]$, which maps {\color{blue}the cumulative reward difference between} the two trajectories to the corresponding preference probability.
That is,
\begin{align}\label{linkfunction}
\mathbb{P}(\tau_{\mathcal{N}^{\kappa}_{i},1} \succ \tau_{\mathcal{N}^{\kappa}_{i},0}) = \sigma\big(\hat{r}_{i}(\tau_{\mathcal{N}^{\kappa}_{i},1}) - \hat{r}_{i}(\tau_{\mathcal{N}^{\kappa}_{i},0})\big),
\end{align}
where $\tau_{\mathcal{N}^{\kappa}_{i},1} \succ \tau_{\mathcal{N}^{\kappa}_{i},0}$ denotes the event that human feedback favors $\tau_{\mathcal{N}^{\kappa}_{i},1}$ over $\tau_{\mathcal{N}^{\kappa}_{i},0}$.
Based on (\ref{linkfunction}), the human feedback $o_{i}$ is a random sample drawn from a Bernoulli distribution with
\begin{align}\label{linkfunctionvalue}
\mathbb{P}(o_{i}=1) = \mathbb{P}(\tau_{\mathcal{N}^{\kappa}_{i},1} \succ \tau_{\mathcal{N}^{\kappa}_{i},0}).
\end{align}
\par
{\color{blue}Among the various preference models proposed in the literature,} the Bradley-Terry model~\citep{Bengs2021} {\color{blue}has emerged as the most widely adopted framework for preference-based learning. It characterizes the probability of one trajectory being preferred over another as a logistic function of the cumulative reward difference, formally expressed} as follows:
\begin{align}\notag
\mathbb{P}(\tau_{\mathcal{N}^{\kappa}_{i},1} \succ \tau_{\mathcal{N}^{\kappa}_{i},0}) =\frac{1}{1+\exp{\big(-\hat{r}_{i}(\tau_{\mathcal{N}^{\kappa}_{i},1})+\hat{r}_{i}(\tau_{\mathcal{N}^{\kappa}_{i},0})\big)}}.
\end{align}
{\color{blue}Beyond the Bradley-Terry model, a variety of alternative preference models have been studied in the literature, including} the linear model with a linear link function, the Weibull model, the Cauchy model, and the complementary log-log model, {\color{blue}each tailored to} specific application contexts~\citep{Train2009,Greene2010}.
{\color{blue}\begin{remark}
The spatiotemporally truncated trajectory mechanism differs fundamentally from single-agent PbRL~\citep{Zhangarxiv2024}, where preference
feedback is elicited from global trajectory comparisons involving the entire system.
In contrast, each agent $i$ in the proposed mechanism elicits
preference feedback from trajectory pairs
$(\tau_{\mathcal{N}^{\kappa}_{i},0}, \tau_{\mathcal{N}^{\kappa}_{i},1})$
collected from its $\kappa$-hop neighborhood $\mathcal{N}^{\kappa}_{i}$ over $H$-horizon.
In this framework, the human evaluator of agent $i$ observes only the state-action sequences within $\kappa$-hop neighborhood and expresses a preference without access to any numerical reward values.
This localized preference feedback implicitly captures the cooperative performance within the local neighborhood, thereby promoting coordinated optimization among agents in a fully distributed manner.
\end{remark}}
\subsection{Distributed algorithm designing}
To develop a distributed algorithm,
we define the joint policy parameter of agents in the $t$-th iteration as  $\bm{\theta}_{t}=(\theta^{\top}_{1,t},\cdots,\theta^{\top}_{N,t})^{\top}$ and introduce the following assumption governing information exchange among agents.
\begin{assumption}\label{theassumptionofcommunicationdistance}
In each $t$-th iteration of the learning process, each agent $i \in \mathcal{N}$ receives {\color{blue}a pair of spatiotemporally truncated sample trajectories}  $(\tau_{\mathcal{N}^{\kappa}_{i},0},\tau_{\mathcal{N}^{\kappa}_{i},1})$ {\color{blue}generated} from {\color{blue}two} distinct policies, along with the corresponding preference signal $o_{i} \in \{0, 1\}$.
\end{assumption}
\par
{\color{blue}Under Assumption~\ref{theassumptionofcommunicationdistance}, each agent collects trajectory pairs from its $\kappa$-hop neighbors via the underlying network $\mathcal{G}$ without requiring any reward signals,} based on which the policy parameters are updated through the following steps.
\par
\textbf{Step 1. Perturbed joint policy}:
In $t$-th iteration, each agent $i$ {\color{blue}independently} samples a perturbation vector $v_{i,t}\sim\mathcal{N}(\mathbf{0}_{d_{i}}, \mathbf{I}_{d_{i}})$,
{\color{blue}and we denote the concatenated global perturbation as}
$\bm{v}_{t} = (v^{\top}_{1,t}, \cdots, v^{\top}_{N,t})^{\top}\sim\mathcal{N}(\mathbf{0}_{d_{\mathrm{tot}}},\mathbf{I}_{d_{\mathrm{tot}}})$.
{\color{blue}The resulting perturbed joint policy is then given by} $\bm{\pi}_{\bm{\theta}_{t} + \mu \bm{v}_{t}}$, {\color{blue}where $\mu>0$ denotes the perturbation distance.}
\par
\textbf{Step 2. Collect samples and preferences}:
In the $t$-th iteration, agents execute {\color{blue}both} the joint policy $\bm{\pi}_{\bm{\theta}_{t}}$ and the perturbed joint policy $\bm{\pi}_{\bm{\theta}_{t}+\mu\bm{v}_{t}}$ {\color{blue}over $K$ independent trials.}
In each trial $k\in\{1,\cdots,K\}$, each agent $i$ collects {\color{blue}a pair of} spatiotemporally truncated trajectories $\tau^{t,k}_{\mathcal{N}^{\kappa}_{i},0}\sim\bm{\pi}_{\bm{\theta}_{t}}$ and $\tau^{t,k}_{\mathcal{N}^{\kappa}_{i},1}\sim\bm{\pi}_{\bm{\theta}_{t}+\mu\bm{v}_{t}}$, and submits the pair $(\tau^{t,k}_{\mathcal{N}^{\kappa}_{i},0}, \tau^{t,k}_{\mathcal{N}^{\kappa}_{i},1})$ to $M$ human evaluators, {\color{blue}obtaining} preference feedback $(o^{t,k}_{i,1}, \cdots, o^{t,k}_{i,M})$.
The preference probability $\mathbb{P}(\tau^{t,k}_{\mathcal{N}^{\kappa}_{i},1} \succ \tau^{t,k}_{\mathcal{N}^{\kappa}_{i},0})$ is {\color{blue}then} estimated by
\begin{align}\label{estimatepreferenceprobability}
\hat{p}^{t,k}_{i}=\mathrm{trim}\Big[\frac{1}{M}\sum_{m=1}^{M}o^{t,k}_{i,m}\Big|\Delta\Big],
\end{align}
where $\Delta=\min\{\sigma(-\frac{2R(1-\gamma^{H})}{1-\gamma}),1-\sigma(\frac{2R(1-\gamma^{H})}{1-\gamma})\}$ {\color{blue}and $\mathrm{trim}\big[\frac{1}{M}\sum_{m=1}^{M}o^{t,k}_{i,m}|\Delta\big]=\min\big\{\max\{\frac{1}{M}\sum_{m=1}^{M}o^{t,k}_{i,m}, \Delta\},$ $1 - \Delta\big\}$.
This trimming operation constrains the estimated preference probability to remain bounded away from 0 and 1, which stabilizes the policy update and enhances robustness against noisy or inconsistent human annotations.
}
{\color{blue}\begin{remark}
In contrast to existing studies on preference-based
learning~\citep{Zhangarxiv2024}, which impose a unit-norm constraint on policy perturbations in single-agent settings, our approach adopts Gaussian perturbations, where each agent $i$ independently draws its local perturbation $v_{i,t}$ from $\mathcal{N}(\mathbf{0}_{d_{i}}, \mathbf{I}_{d_{i}})$.
This is motivated by the fact that satisfying the unit-norm constraint in a distributed manner requires global coordination among all agents to normalize the joint perturbation vector, which is incompatible with the local communication constraints of the networked system.
By adopting Gaussian perturbations, each agent independently generates its local perturbation without any inter-agent coordination, making the procedure fully compatible with distributed execution.
\end{remark}}

\par
\textbf{Step 3. Policy parameter update}:
{\color{blue}Upon obtaining} the estimated preference probabilities
$\{\hat{p}^{t,k}_{i}\}_{k=1}^{K}$,
{\color{blue}each agent $i$ leverages} the link function $\sigma(\cdot)$ to {\color{blue}construct} a local estimate of the policy gradient as follows:
\begin{align}\label{policygradientestimation}
\hat{g}_{i,t}=\frac{1}{K\mu}\sum_{k=1}^{K}\sigma^{-1}(\hat{p}^{t,k}_{i})v_{i,t}.
\end{align}
Based on $\hat{g}_{i,t}$ in (\ref{policygradientestimation}), the policy parameter {\color{blue}$\theta_{i,t+1}$} can be updated by
\begin{align}\label{updateofpolicyparameter}
\theta_{i,t+1}=\theta_{i,t}+\alpha\hat{g}_{i,t},
\end{align}
where $\alpha>0$ is a learning rate {\color{blue}in} policy parameter {\color{blue}update}.
\par
Following the Steps 1 to 3 outlined above,
a distributed zeroth-order policy gradient algorithm
is proposed in Algorithm~\ref{Distributedzeroth-orderpolicygradientalgorithm}.
\begin{algorithm}[H]
\caption{Distributed Zeroth-Order Policy Gradient Algorithm}\label{Distributedzeroth-orderpolicygradientalgorithm}
\begin{algorithmic}[1]
\Require{Initial parameter $\bm{\theta}_{0}$, learning rate $\alpha>0$, perturbed distance $\mu>0$, horizon length $H$, human evaluators $M$, trim size $\Delta=\min\{\sigma(-\frac{2R(1-\gamma^{H})}{1-\gamma}),1-\sigma(\frac{2R(1-\gamma^{H})}{1-\gamma})\}$;}
\For{$t=0,1,\cdots,T-1$}
\State\parbox[t]{0.93\linewidth}{
Each agent $i\in\mathcal{N}$ {\color{blue}independently samples a local perturbation vector $v_{i,t}\sim \mathcal{N}(\mathbf{0}_{d_{i}}, \mathbf{I}_{d_{i}})$, forming the global perturbation $\bm{v}_{t} = (v^{\top}_{1,t}, \cdots,v^{\top}_{N,t})^{\top}$}\;}
    \For{$k=1,2,\cdots,K$}
        \State\parbox[t]{0.86\linewidth}{Agents execute the joint policy $\bm{\pi}_{\bm{\theta}_{t}}$ to generate the first joint trajectory.
        {\color{blue}Subsequently, they reset the environment and execute the perturbed joint policy $\bm{\pi}_{\bm{\theta}_{t}+\mu\bm{v}_{t}}$ to generate the second joint trajectory};}
        \State\parbox[t]{0.86\linewidth}{Each agent $i\in\mathcal{N}$ {\color{blue}collects spatiotemporally truncated trajectories}
        $\tau^{t,k}_{\mathcal{N}^{\kappa}_{i},0}$ and $\tau^{t,k}_{\mathcal{N}^{\kappa}_{i},1}$ {\color{blue}from first joint trajectory and second joint trajectory, respectively, via the underlying network $\mathcal{G}$};}
        \State\parbox[t]{0.86\linewidth}{Each agent $i\in\mathcal{N}$ queries $M$ human evaluators with $(\tau^{t,k}_{\mathcal{N}^{\kappa}_{i},0},\tau^{t,k}_{\mathcal{N}^{\kappa}_{i},1})$ and obtain feedback $(o^{t,k}_{i,1},\cdots,o^{t,k}_{i,M})$;}
        \State\parbox[t]{0.86\linewidth}{Each agent $i\in\mathcal{N}$ estimates preference probability as $\hat{p}^{t,k}_{i}$ in (\ref{estimatepreferenceprobability});}
    \EndFor
    \State\parbox[t]{0.93\linewidth}{Each agent $i\in\mathcal{N}$ estimates its local policy gradient through (\ref{policygradientestimation});}
    \State\parbox[t]{0.93\linewidth}{Each agent $i\in\mathcal{N}$ updates its local policy parameter by (\ref{updateofpolicyparameter});}
\EndFor
\State\textbf{Output:} Final policy parameter $\bm{\theta}_{T}$
\end{algorithmic}
\end{algorithm}

\par
Algorithm~\ref{Distributedzeroth-orderpolicygradientalgorithm} is a distributed human preference-based algorithm specifically designed for the NMARL problem.
{\color{blue}It is worth
noting that the Algorithm~\ref{Distributedzeroth-orderpolicygradientalgorithm} is communication-efficient by design: perturbation sampling (see Line~2) and policy parameter update (see Lines~9-10) are performed entirely independently by each agent without any inter-agent coordination, and inter-agent communication occurs only during the trajectory collection stage (see Line~5), where each agent receives the state-action information from its $\kappa$-hop neighbors via local message passing over the underlying network $\mathcal{G}$.}
This design significantly reduces the communication overhead for agents and enhances the scalability of the algorithm in large-scale systems.
{\color{blue}\begin{remark}
The online human preference feedback in the proposed Algorithm~\ref{Distributedzeroth-orderpolicygradientalgorithm} is a standard abstraction in the PbRL literature~\citep{Christiano2017,Wirth2017,Zhangarxiv2024}.
We emphasize that ``human preference feedback'' should
not be exclusively interpreted as requiring a real-time human annotation for every trajectory pair.
Rather, it encompasses general preference signals encoding human intentions, including feedback from a pre-trained preference model, thereby enabling broad applicability to networked multi-agent settings where direct reward specification is impractical.
\end{remark}}
\par
{\color{blue}To further contextualize the computational efficiency of the proposed algorithm, we compare Algorithm~\ref{Distributedzeroth-orderpolicygradientalgorithm} with the preference-based policy gradient method (i.e., PG-RLHF) in~\citep{Duyihan2024}.
Table~\ref{tab:samplecomplexitycomparison} compares the sample
complexity and number of preference queries between the proposed algorithm and PG-RLHF, where $T$ and $K$ denote the number of policy update iterations and pairwise trajectory pairs per iteration, respectively, and $M$ denotes the number of preference feedback queries per trajectory pair.
For PG-RLHF, $K_{1}$, $K_{2}$, and $T_{1}$ denote the number of samples for covariance matrix update stage, the number of samples in the stochastic gradient descent stage, and the number of iterations in the inner-loop of natural policy gradient update stage, respectively.
As shown in Table~\ref{tab:samplecomplexitycomparison}, the proposed Algorithm~\ref{Distributedzeroth-orderpolicygradientalgorithm} achieves a more compact sample complexity $\mathcal{O}(TK)$ than
$\mathcal{O}(TK_{1}+TK+TT_{1}K_{2})$ in PG-RLHF, as the latter requires multiple additional estimation stages.
Nevertheless, the higher query complexity $\mathcal{O}(TKM)$ is a natural and acceptable trade-off: by eliminating the multiple estimation stages required in PG-RLHF, the proposed algorithm exchanges additional preference queries for a simpler algorithmic structure.}
\begin{table}[!ht]
\centering
\caption{{\color{blue}Comparison of sample complexity and the number of queries between Algorithm~\ref{Distributedzeroth-orderpolicygradientalgorithm} and PG-RLHF in \citep{Duyihan2024}.}}\label{tab:samplecomplexitycomparison}
\scriptsize
{\color{blue}\begin{tabular}{ccc}
\toprule
References & Algorithm~\ref{Distributedzeroth-orderpolicygradientalgorithm} &
 PG-RLHF~\citep{Duyihan2024}\\
\midrule
Samples &  $\mathcal{O}(TK)$  & $\mathcal{O}(TK_{1}+TK+TT_{1}K_{2})$ \\
Queries & $\mathcal{O}(TKM)$ & $\mathcal{O}(TK)$ \\
\bottomrule
\end{tabular}}
\end{table}

{\color{blue}\begin{remark}
Note that the distributed zeroth-order policy gradient estimation is adopted for two main reasons.
First, as discussed above and illustrated in Table~\ref{tab:samplecomplexitycomparison},
it avoids reward model learning and leads to a more compact sample complexity than first-order alternatives.
Second, it is naturally compatible with trajectory-level human preference feedback. Since human evaluators compare sample trajectories rather than providing step-level rewards,
first-order methods that rely on per-step reward signals or $Q$-value estimates are not directly applicable.
In contrast, zeroth-order estimation approximates the policy gradient via trajectory-level perturbations, which aligns well with pairwise preference feedback.
\end{remark}}

\section{Convergence analysis of Algorithm~\ref{Distributedzeroth-orderpolicygradientalgorithm}}\label{SectionofConvergenceanalysis}
{\color{blue}While the proposed Algorithm~\ref{Distributedzeroth-orderpolicygradientalgorithm} can be viewed at a high level as extending single-agent PbRL~\citep{Zhangarxiv2024} to a networked multi-agent
setting~\citep{QuCLDC2020,QuNIPS2020}, the analysis reveals three non-trivial technical obstacles that distinguish it from a direct combination of existing results.
\textbf{(i) Feedback mechanism design}: unlike single-agent PbRL where references are elicited from global trajectory comparisons, the proposed spatiotemporally truncated trajectory mechanism must simultaneously be localized, informative for cooperative performance, and consistent with the global objective, which introduces a co-design challenge in the theoretical analysis that is absent in single-agent settings.
\textbf{(ii) Gradient estimation under Gaussian perturbations}: unit-norm perturbation used in single-agent PbRL~\citep{Zhangarxiv2024} requires global coordination incompatible with distributed execution.
The adoption of Gaussian perturbations introduces fundamentally new bias and variance sources from the interplay of spatiotemporal truncation and stochastic preference feedback, making the theoretical characterization of the gradient estimator significantly more involved than in the single-agent case.
\textbf{(iii) Convergence analysis}: unlike existing NMARL methods~\citep{QuCLDC2020,QuNIPS2020} that rely on local reward signals and $Q$-value sharing, establishing convergence here requires simultaneously bounding the zeroth-order estimation error, the spatiotemporal truncation bias, and the stochastic preference feedback noise within a unified analytical framework,
which constitutes the primary technical difficulty of this work.}
\par
{\color{blue}In the sequel, we develop a unified convergence analysis of Algorithm~\ref{Distributedzeroth-orderpolicygradientalgorithm} that simultaneously addresses all three aforementioned challenges.}
Before analyzing the convergence of Algorithm~\ref{Distributedzeroth-orderpolicygradientalgorithm}, we impose the following assumption on the link function, as adopted in~\citep{Zhangarxiv2024}.
\begin{assumption}\label{theassumptionofsigma}
The link function $\sigma(\cdot)$ in the preference model is defined over the interval $[-\frac{2R(1-\gamma^{H})}{1-\gamma},\frac{2R(1-\gamma^{H})}{1-\gamma}]$  bounded within $[0,1]$, and is strictly
monotonically increasing with $\sigma(0)=\frac{1}{2}$.
The inverse link function $\sigma^{-1}(\cdot)$ is $L_{\sigma}$-Lipschitz continuous on $[\Delta,1-\Delta]$.
\end{assumption}
\par
Assumption~\ref{theassumptionofsigma} serves as the basis for inferring the reward difference in agent's $\kappa$-hop neighbors from human preference probabilities, and it can be effectively fulfilled by the Bradley-Terry model.
Moreover, under Assumption~\ref{theassumptionofsigma}, the error incurred by preference estimation generated by link function $\sigma(\cdot)$ in Algorithm~\ref{Distributedzeroth-orderpolicygradientalgorithm} is characterized in the following lemma.
\begin{lemma}\label{thelemmaofpreferenceestimation}
Suppose Assumption~\ref{theassumptionofsigma} holds.
For any $t$-th iteration, $k$-th trial, and agent $i\in\mathcal{N}$.
The trajectory
pairs $(\tau^{t,k}_{\mathcal{N}^{\kappa}_{i},0},\tau^{t,k}_{\mathcal{N}^{\kappa}_{i},1})$ that is queried from $M\geq2$ human evaluators satisfy
\begin{align}\label{thelemmaofpreferenceestimation(i)}
&\mathbb{E}\big[\big|\sigma^{-1}(\hat{p}^{t,k}_{i})-\big(\hat{r}_{i}(\tau^{t,k}_{\mathcal{N}^{\kappa}_{i},1})-\hat{r}_{i}(\tau^{t,k}_{\mathcal{N}^{\kappa}_{i},0})\big)\big|\big]\notag\\
\leq&L_{\sigma}\sqrt{\frac{2\log{M}}{M}}+\frac{4R(1-\gamma^{H})}{(1-\gamma)M^{2}}
\end{align}
and
\begin{align}\label{thelemmaofpreferenceestimation(iii)}
&\mathbb{E}\big[\big|\sigma^{-1}(\hat{p}^{t,k}_{i})-\big(\hat{r}_{i}(\tau^{t,k}_{\mathcal{N}^{\kappa}_{i},1})-\hat{r}_{i}(\tau^{t,k}_{\mathcal{N}^{\kappa}_{i},0})\big)\big|^{4}\big]\notag\\
\leq&\frac{4L^{4}_{\sigma}(\log{M})^{2}}{M^{2}}+\frac{256R^{4}(1-\gamma^{H})^{4}}{(1-\gamma)^{4}M^{2}}.
\end{align}
\end{lemma}
\par
Lemma~\ref{thelemmaofpreferenceestimation} quantifies the error arising from preference estimation, and a detailed proof is provided in {\color{blue}Appendix~\ref{ProofofLemmathelemmaofpreferenceestimation}}.
It shows that the reward difference estimation from human preference in (\ref{estimatepreferenceprobability}) is accurate as long as the number of human evaluators $M$ is large.

\subsection{Auxiliary definitions and results}
In this subsection, we introduce several auxiliary definitions, along with the relevant assumptions and theoretical results, which serve as a foundation for the subsequent proofs.
\par
Similar to the objective function $J(\bm{\theta})$ defined in (\ref{theobjectivefunction}), we define $\widehat{J}_{i}(\bm{\theta})$ as the spatiotemporally truncated objective function associated with agent $i$, which can be expressed as
\begin{align}\label{thetruncatedobjectivefunction}
\widehat{J}_{i}(\bm{\theta})=&\mathbb{E}_{\bm{s}\sim\bm{\rho}}\Big[\frac{1}{N}\sum_{t=0}^{H-1}\sum_{j\in\mathcal{N}^{\kappa}_{i}}\gamma^{t}r_{j,t}\Big|\bm{s}_{0}=\bm{s},\notag\\
&\bm{a}_{t}\sim\bm{\pi_{\theta}}(\cdot|\bm{s}_{t})\Big],
\end{align}
It should be noted that $\widehat{J}_{i}(\bm{\theta})$ depends solely on the accumulated rewards obtained from agent $i$'s $\kappa$-hop neighbors over $H$-horizon.
\par
For any joint policy $\bm{\pi}_{\bm{\theta}}$, we define $\nabla_{\bm{\theta}}J(\bm{\theta})=\big(\nabla_{\theta_{1}}J(\bm{\theta})^{\top},$ $\cdots,\nabla_{\theta_{N}}J(\bm{\theta})^{\top}\big)^{\top}$ and  $\nabla_{\bm{\theta}}\widehat{J}_{i}(\bm{\theta})=\big(\nabla_{\theta_{1}}\widehat{J}_{i}(\bm{\theta})^{\top},\cdots,$
$\nabla_{\theta_{N}}\widehat{J}_{i}(\bm{\theta})^{\top}\big)^{\top}$, and make the following assumptions, which are a set of standard assumptions commonly adopted in~\citep{Zhangarxiv2024,QuCLDC2020,QuNIPS2020,LinNIPS2021,Zhou2023UAI,Zhangrunyu2022NIPS,Ying2023NIPS}
\begin{assumption}\label{theassumptionofdistributionofs}
For any joint policy parameter $\bm{\theta}$, $d^{\bm{\pi_{\theta}}}_{\bm{\rho}}(\bm{s})$ and $\xi^{\bm{\pi_{\theta}}}_{\bm{\rho}}(\bm{s},\bm{a})$ satisfy  $\inf_{\bm{\theta}}\min_{\bm{s}\in\bm{\mathcal{S}}}d^{\bm{\pi_{\theta}}}_{\bm{\rho}}(\bm{s})>0$ and $\inf_{\bm{\theta}}\min_{(\bm{s},\bm{a})\in\bm{\mathcal{S}}\times\bm{\mathcal{A}}}\xi^{\bm{\pi_{\theta}}}_{\bm{\rho}}(\bm{s},\bm{a})>0$, where $\xi^{\bm{\pi_{\theta}}}_{\bm{\rho}}(\bm{s},\bm{a})$ is the discounted state-action visitation distribution of $(\bm{s},\bm{a})\in\bm{\mathcal{S}}\times\bm{\mathcal{A}}$ and  satisfies
\begin{align}\label{thestationarydistributionofsa}
\xi^{\bm{\pi_{\theta}}}_{\bm{\rho}}(\bm{s},\bm{a})=d^{\bm{\pi_{\theta}}}_{\bm{\rho}}(\bm{s})\bm{\pi_{\theta}}(\bm{a}|\bm{s}).
\end{align}
\end{assumption}
\par
Assumption~\ref{theassumptionofdistributionofs} {\color{blue}is a standard condition adopted in~\citep{Zhou2023UAI,Zhangrunyu2022NIPS}, which} requires that the parameterized policy explores all the states and state-action pairs with some positive probability.
{\color{blue}This assumption is commonly satisfied in finite state-action settings under sufficient exploration, for instance when a softmax policy is used and the underlying Markov decision process is ergodic.}
\begin{assumption}\label{theassumptionofpolicy}
For any joint policy $\bm{\pi_{\theta}}$ and agent $i\in\mathcal{N}$, $\nabla_{\theta_{i}}\log\pi_{i}(a_{i}|s_{i},\theta_{i})$ exists and satisfies $\|\nabla_{\theta_{i}}\log\pi_{i}(a_{i}|s_{i},\theta_{i})\|\leq B$ with $B>0$ for any  $(s_{i},a_{i})\in\mathcal{S}_{i}\times\mathcal{A}_{i}$.
\end{assumption}
\par
Assumption~\ref{theassumptionofpolicy} is a {\color{blue}common} requirement for convergence analysis in the field of MARL~{\color{blue}\citep{QuCLDC2020,QuNIPS2020,LinNIPS2021,Zhangrunyu2022NIPS,Ying2023NIPS}}.
\begin{assumption}\label{theassumptionofgradientofobjectivefunction}
For any policy parameter $\bm{\theta}$, $\nabla_{\bm{\theta}}J(\bm{\theta})$ and $\nabla_{\bm{\theta}}\widehat{J}_{i}(\bm{\theta})$
are both $L$-Lipschitz continuous in $\bm{\theta}$.
\end{assumption}
\par
Assumption~\ref{theassumptionofgradientofobjectivefunction} is standard in the convergence analysis of NMARL problems.
{\color{blue}The $L$-Lipschitz continuity of $\nabla_{\bm{\theta}}J(\bm{\theta})$ has been widely adopted in prior work~\citep{QuCLDC2020,QuNIPS2020,Ying2023NIPS}, and its validity under the softmax policy parameterization is well-established in~\citep{Zhou2023UAI,Zhangrunyu2022NIPS}.}
\begin{remark}
{\color{blue}The $L$-Lipschitz continuity of $\nabla_{\bm{\theta}}\widehat{J}_{i}(\bm{\theta})$ is consistent with Assumption~3 in~\citep{Zhangarxiv2024} established for single-agent finite-horizon objective.
Building on this, we further verify that $\nabla_{\bm{\theta}}\widehat{J}_{i}(\bm{\theta})$ satisfies an analogous condition in the multi-agent setting, with a detailed justification provided in Appendix~\ref{DiscussionoftherationalityAssumptiontheassumptionofgradientofobjectivefunction}.}
\end{remark}
\par
Under Assumptions~\ref{theassumptionofreward},~\ref{theassumptionofdistributionofs}, and~\ref{theassumptionofpolicy}, the spatiotemporally truncated objective function $\widehat{J}_{i}(\bm{\theta})$ in (\ref{thetruncatedobjectivefunction}) satisfies the following property.
\begin{theorem}\label{lemmaoftruncatederror}
Suppose Assumptions~\ref{theassumptionofreward},~\ref{theassumptionofdistributionofs}, and~\ref{theassumptionofpolicy} hold.
For any agent $i\in\mathcal{N}$ and any policy parameters $\bm{\theta}\in\mathbb{R}^{d_{\mathrm{tot}}}$, we have
\begin{align}\label{theresultoftrunctederror}
\|\nabla_{\theta_{i}}\widehat{J}_{i}(\bm{\theta})-\nabla_{\theta_{i}}J(\bm{\theta})\|\leq\frac{BR\big((H+1)\gamma^{H}+2\gamma^{\kappa+1}\big)}{(1-\gamma)^{2}}.
\end{align}
\end{theorem}
\par
Theorem~\ref{lemmaoftruncatederror} establishes an upper bound on the discrepancy between the exact policy gradient $\nabla_{\theta_{i}}J(\bm{\theta})$ as defined in~(\ref{thepolicygradienttheorem}) and the policy gradient $\nabla_{\theta_{i}}\widehat{J}_{i}(\bm{\theta})$ of the spatiotemporally truncated objective function $\widehat{J}_{i}(\bm{\theta})$ in (\ref{thetruncatedobjectivefunction}).
This bound decreases exponentially with respect to the spatiotemporal truncation parameters $H$ and $\kappa$, and the detailed proof is provided in {\color{blue}Appendix~\ref{ProofofLemmalemmaoftruncatederror}}.
\par
Recalling the objective function $\widehat{J}_{i}(\bm{\theta})$ in (\ref{thetruncatedobjectivefunction}), for {\color{blue}any} perturbed distance $\mu$, we define a perturbed objective function as
\begin{align}\label{perturbedvaluefunction}
\widehat{J}^{\mu}_{i}(\bm{\theta})=\mathbb{E}_{\bm{v}}[\widehat{J}_{i}(\bm{\theta}+\mu\bm{v})],
\end{align}
where $\bm{v}$ is sampled from Gaussian distribution $\mathcal{N}(\mathbf{0}_{d_{\mathrm{tot}}},\mathbf{I}_{d_{\mathrm{tot}}})$.
By the definition of $\widehat{J}^{\mu}_{i}(\bm{\theta})$ in (\ref{perturbedvaluefunction}), we can have the following {\color{blue}lemma}.
\begin{lemma}\label{thelemmaofperturbedobjectivefunction}
Suppose Assumption~\ref{theassumptionofgradientofobjectivefunction} holds.
For the perturbed objective function $\widehat{J}^{\mu}_{i}(\bm{\theta})$ defined in (\ref{perturbedvaluefunction}), we have the following results:
\par
(i) $\widehat{J}^{\mu}_{i}(\bm{\theta})$ is $L$-smooth and satisfies
\begin{align}\label{thegradientofperturbedobjectivefunction}
\nabla_{\bm{\theta}}\widehat{J}^{\mu}_{i}(\bm{\theta})=\mathbb{E}_{\bm{v}}\Big[\frac{1}{\mu}\big(\widehat{J}_{i}(\bm{\theta}+\mu\bm{v})-\widehat{J}_{i}(\bm{\theta})\big)\bm{v}\Big];
\end{align}
\par
(ii) For any $\bm{\theta}\in\mathbb{R}^{d_{\mathrm{tot}}}$, the difference in function values satisfies $|\widehat{J}^{\mu}_{i}(\bm{\theta})-\widehat{J}_{i}(\bm{\theta})|\leq\frac{L\mu^{2}d_{\mathrm{tot}}}{2}$;
\par
(iii) For any $\bm{\theta}\in\mathbb{R}^{d_{\mathrm{tot}}}$, the difference in policy gradient satisfies $\|\nabla_{\bm{\theta}}\widehat{J}^{\mu}_{i}(\bm{\theta})-\nabla_{\bm{\theta}}\widehat{J}_{i}(\bm{\theta})\|\leq L\mu\sqrt{d_{\mathrm{tot}}}$;
\par
(iv) For any $\bm{\theta}\in\mathbb{R}^{d_{\mathrm{tot}}}$, the policy gradient noise satisfies
\begin{align}
&\mathbb{E}_{\bm{v}}\Big[\Big\|\frac{1}{\mu}\big(\widehat{J}_{i}(\bm{\theta}+\mu\bm{v})-\widehat{J}_{i}(\bm{\theta})\big)\bm{v}\Big\|^{2}\Big]\notag\\
\leq&2d_{\mathrm{tot}}(d_{\mathrm{tot}}\!+\!2)\|\nabla_{\bm{\theta}}\widehat{J}_{i}(\bm{\theta})\|^{2}\!+\!\frac{\mu^{2}L^{2}d_{\mathrm{tot}}(d_{\mathrm{tot}}\!+\!2)(d_{\mathrm{tot}}\!+\!4)}{2}.\label{thegradientofperturbedobjectivefunction(iv)}
\end{align}
\end{lemma}
\par
Lemma~\ref{thelemmaofperturbedobjectivefunction} presents several important properties of the perturbed objective function $\widehat{J}^{\mu}_{i}(\bm{\theta})$, with its proof detailed in {\color{blue}Appendix~\ref{ProofofLemmathelemmaofperturbedobjectivefunction}}.
These properties will lay the groundwork for the analysis of the $\epsilon$-stationary convergence of Algorithm~\ref{Distributedzeroth-orderpolicygradientalgorithm}.

\subsection{First-order gradient bias and gradient
noise}
In this subsection, we examine the first-order bias and gradient noise introduced by the estimated policy gradients $\{\hat{g}_{i,t}\}_{i\in\mathcal{N}}$ in~(\ref{policygradientestimation}).
\par
Let $\bm{\hat{g}}_{t}=\big(\hat{g}^{\top}_{1,t},\cdots,\hat{g}^{\top}_{N,t}\big)^{\top}$ and $\mathcal{F}_{t}$ denote the $\sigma$-algebra representing the filtration up to iteration $t$.
The first-order gradient bias and gradient noise caused by $\bm{\hat{g}}_{t}$ in the $t$-th iteration are respectively defined as
\begin{align}\label{theequationofthefirst-ordergradientbias}
\mathrm{Bias}_{t}=\big\langle\nabla_{\bm{\theta}}J(\bm{\theta}_{t}),\mathbb{E}\big[\nabla_{\bm{\theta}}J(\bm{\theta}_{t})-\bm{\hat{g}}_{t}\big|\mathcal{F}_{t}\big]\big\rangle
\end{align}
and
\begin{align}\label{theequationofthegradientnoise}
\mathrm{Var}_{t}=\mathbb{E}[\|\bm{\hat{g}}_{t}\|^{2}|\mathcal{F}_{t}].
\end{align}
Here, $\mathrm{Bias}_{t}$ characterizes the directional bias of the estimator relative to the true gradient $\nabla_{\bm{\theta}}J(\bm{\theta}_{t})$, {\color{blue}measuring the projection of} the expected estimation error onto the true gradient direction.
$\mathrm{Var}_{t}$ denotes the expected squared norm of $\bm{\hat{g}}_{t}$ conditioned on $\mathcal{F}_t$, quantifying the estimator's intrinsic variance.
Based on the definitions of $\mathrm{Bias}_{t}$ in (\ref{theequationofthefirst-ordergradientbias}) and $\mathrm{Var}_{t}$ in (\ref{theequationofthegradientnoise}), we are able to derive the following theorem.
\begin{theorem}\label{thetheoremofpolicygradient}
Suppose Assumptions~\ref{theassumptionofreward}-\ref{theassumptionofgradientofobjectivefunction} hold.
For joint policy $\bm{\pi}_{\bm{\theta}_{t}}$ in the $t$-th iteration of Algorithm~\ref{Distributedzeroth-orderpolicygradientalgorithm}, by selecting $M\geq\max\{e,2(\frac{R^{2}(1-\gamma^{H})^{2}}{(1-\gamma)^{2}L^{2}_{\sigma}})^{\frac{1}{3}},\exp\big(\frac{8R^{2}(1-\gamma^{H})^{2}}{(1-\gamma)^{2}L^{2}_{\sigma}}\big)\}$,
we obtain:
\par
(i)
The first-order gradient bias $\mathrm{Bias}_{t}$ satisfies
\begin{align}\label{theequationofcorollary}
\mathrm{Bias}_{t}
\leq&\big\|\nabla_{\bm{\theta}}J(\bm{\theta}_{t})\big\|\Bigg(\underbrace{\frac{BR\sqrt{N}\big((H+1)\gamma^{H}+2\gamma^{\kappa+1}\big)}{(1-\gamma)^{2}}}_{\mathrm{Bias}^{\mathrm{(i)}}_{t}}\notag\\
&+\!\underbrace{L\mu\sqrt{d_{\mathrm{tot}}N}}_{\mathrm{Bias}^{\mathrm{(ii)}}_{t}}\!+\!\underbrace{\frac{2L_{\sigma}\sqrt{d_{\mathrm{tot}}N}}{\mu}\sqrt{\frac{2\log{M}}{M}}}_{\mathrm{Bias}^{\mathrm{(iii)}}_{t}}\Bigg).
\end{align}
\par
(ii)
The gradient noise $\mathrm{Var}_{t}$ satisfies
\begin{align}
\mathrm{Var}_{t}
\leq&\underbrace{12Nd_{\mathrm{tot}}(d_{\mathrm{tot}}+2)\|\nabla_{\bm{\theta}}J(\bm{\theta}_{t})\|^{2}}_{\mathrm{Var}^{\mathrm{(i)}}_{t}}\notag\\
&+\underbrace{\frac{12B^{2}NR^{2}\big((H\!+\!1)\gamma^{H}\!+\!2\gamma^{\kappa\!+\!1}\big)^{2}d_{\mathrm{tot}}(d_{\mathrm{tot}}\!+\!2)}{(1\!-\!\gamma)^{4}}}_{\mathrm{Var}^{\mathrm{(ii)}}_{t}}\notag\\
&+\underbrace{\frac{3\mu^{2}L^{2}Nd_{\mathrm{tot}}(d_{\mathrm{tot}}+2)(d_{\mathrm{tot}}+4)}{2}}_{\mathrm{Var}^{\mathrm{(iii)}}_{t}}\notag\\
&+\underbrace{\frac{6\sum_{i=1}^{N}\sqrt{2d_{i}(d_{i}+2)}L^{2}_{\sigma}\log{M}}{\mu^{2}M}}_{\mathrm{Var}^{\mathrm{(iv)}}_{t}}\notag\\
&+\underbrace{\frac{132\sum_{i=1}^{N}\sqrt{d_{i}(d_{i}+2)}R^{2}}{\mu^{2}(1-\gamma)^{2}K}}_{\mathrm{Var}^{\mathrm{(v)}}_{t}}.\label{theequationofcorollarygradientnoise}
\end{align}
\end{theorem}
\par
The proof of Theorem~\ref{thetheoremofpolicygradient} is provided in {\color{blue}Appendix~\ref{ProofofTheoremthetheoremofpolicygradient}}.
Part (i) of Theorem~\ref{thetheoremofpolicygradient} establishes an upper bound {\color{blue}on} the first-order gradient bias $\mathrm{Bias}_{t}$ {\color{blue}of the estimated policy gradient derived from human feedback}, where $\mathrm{Bias}^{\mathrm{(i)}}_{t}$ {\color{blue}captures the error due to} spatiotemporal truncation parameters $\kappa$ and $H$, $\mathrm{Bias}^{\mathrm{(ii)}}_{t}$ reflects the policy gradient {\color{blue}discrepancy} induced by the perturbed distance $\mu$, and $\mathrm{Bias}^{\mathrm{(iii)}}_{t}$ represents the bias in the estimated gradient $\hat{g}_{i,t}$ in~(\ref{policygradientestimation}), which depends on {\color{blue}both} $\mu$ and the $M$ preference signals collected from human evaluators.
\par
Part~(ii) characterizes the gradient noise $\mathrm{Var}_{t}$, {\color{blue}where $\mathrm{Var}^{\mathrm{(i)}}_{t}$ represents the multiplicative noise inherent in zeroth-order gradient estimation.
$\mathrm{Var}^{\mathrm{(ii)}}_{t}$ captures the error induced by the truncation parameters $\kappa$ and $H$, while $\mathrm{Var}^{\mathrm{(iii)}}_{t}$-$\mathrm{Var}^{\mathrm{(v)}}_{t}$ are all influenced by the perturbed distance $\mu$.
In particular, $\mathrm{Var}^{\mathrm{(iv)}}_{t}$ additionally depends on the number of human evaluators $M$, reflecting the preference estimation error, and $\mathrm{Var}^{\mathrm{(v)}}_{t}$ further depends on the number of trials $K$. }
\subsection{$\epsilon$-stationary convergence}
In this subsection, we establish the theoretical guarantees of Algorithm~\ref{Distributedzeroth-orderpolicygradientalgorithm} under Assumptions~\ref{theassumptionofreward}-\ref{theassumptionofgradientofobjectivefunction}. Specifically, we analyze its convergence rate and characterize the sample complexity required to attain an $\epsilon$-stationary policy.
\begin{theorem}\label{thetheoremofconvergence}
Suppose Assumptions~\ref{theassumptionofreward}-\ref{theassumptionofgradientofobjectivefunction} hold.
For joint policy $\{\bm{\pi}_{\bm{\theta}_{t}}\}_{t=0}^{T-1}$ generated by Algorithm~\ref{Distributedzeroth-orderpolicygradientalgorithm},
{\color{blue}suppose the number of human evaluators satisfies}
$M\geq\max\big\{e,2(\frac{R^{2}(1-\gamma^{H})^{2}}{(1-\gamma)^{2}L^{2}_{\sigma}})^{\frac{1}{3}},\exp\big(\frac{8R^{2}(1-\gamma^{H})^{2}}{(1-\gamma)^{2}L^{2}_{\sigma}}\big)\big\}$, {\color{blue}and the perturbation distance $\mu$ and the learning rate $\alpha$ are selected as 
$\mu^{2}=\max\big\{\frac{R}{(1-\gamma)L\sqrt{KNd_{\mathrm{tot}}}},$
$\frac{3L_{\sigma}}{L}\sqrt{\frac{\log{M}}{M}}\big\}$
and $\alpha=\frac{1}{12LNd_{\mathrm{tot}}(d_{\mathrm{tot}}+2)}$.} Then the convergence rate of Algorithm~\ref{Distributedzeroth-orderpolicygradientalgorithm} satisfies:
\begin{align}
&\frac{1}{T}\sum_{t=0}^{T-1}\mathbb{E}[\|\nabla_{\bm{\theta}}J(\bm{\theta}_{t})\|^{2}]\notag\\
=&\mathcal{O}\Big(\frac{LNRd_{\mathrm{tot}}(d_{\mathrm{tot}}+2)}{(1-\gamma)T}\notag\\
&+\frac{B^{2}R^{2}N\big((H+1)\gamma^{H}+2\gamma^{\kappa+1}\big)^{2}}{(1-\gamma)^{4}}\notag\\
&+\max\Big\{\frac{LR\sqrt{Nd_{\mathrm{tot}}}}{(1-\gamma)\sqrt{K}},LNd_{\mathrm{tot}}L_{\sigma}\sqrt{\frac{\log{M}}{M}}\Big\}\Big).\label{theresultequationoftheorem2}
\end{align}
\end{theorem}
\par
Theorem~\ref{thetheoremofconvergence} establishes the approximate stationary convergence of Algorithm~\ref{Distributedzeroth-orderpolicygradientalgorithm}, {\color{blue}with} the proof provided in {\color{blue}Appendix~\ref{ProofofTheoremthetheoremofconvergence}}.
To interpret the result, {\color{blue}the term} $\frac{LNRd_{\mathrm{tot}}(d_{\mathrm{tot}}+2)}{(1-\gamma)T}$ reflects the zeroth-order gradient ascent rate, $\frac{B^{2}R^{2}N\big((H+1)\gamma^{H}+2\gamma^{\kappa+1}\big)^{2}}{(1-\gamma)^{4}}$ {\color{blue}captures} the approximation error due to spatiotemporal truncation parameters $\kappa$ and $H$, $\frac{LR\sqrt{Nd_{\mathrm{tot}}}}{(1-\gamma)\sqrt{K}}$ accounts for the variance arising from approximating the truncated objective via empirical trajectory rewards over $K$ trials, and $LNd_{\mathrm{tot}}L_{\sigma}\sqrt{\frac{\log{M}}{M}}$ {\color{blue}quantifies the error} incurred in estimating the population-level preference probability from $M$ observed human feedback signals.


\par
Based on Theorem~\ref{thetheoremofconvergence}, we have the following corollary that characterizes the sample complexity {\color{blue}of Algorithm~\ref{Distributedzeroth-orderpolicygradientalgorithm}}.
\begin{corollary}\label{thecorollaryofconvergence}
Suppose Assumptions~\ref{theassumptionofreward}-\ref{theassumptionofgradientofobjectivefunction} hold.
For every $0<\epsilon<1$,
when the hyper-parameters in Algorithm~\ref{Distributedzeroth-orderpolicygradientalgorithm} satisfies $T=\mathcal{O}(\frac{1}{\epsilon})$, $\kappa=\Omega(\log{\frac{1}{\epsilon}})$, $H=\Omega(\log{\frac{1}{\epsilon}})$, $K=\mathcal{O}(\frac{1}{\epsilon^{2}})$, and $M=\mathcal{O}(\frac{\log{\frac{1}{\epsilon}}}{\epsilon^{2}})$,  Algorithm~\ref{Distributedzeroth-orderpolicygradientalgorithm} can learn an $\epsilon$-stationary policy.
\end{corollary}
\par
Corollary~\ref{thecorollaryofconvergence} reveals the sample complexity for Algorithm~\ref{Distributedzeroth-orderpolicygradientalgorithm} to converge to an $\epsilon$-stationary point.
It is a direct extension of Theorem~\ref{thetheoremofconvergence}, so its proof is omitted.

\section{Simulation}\label{SectionSimulations}
In this section, we study the empirical performance of the proposed Algorithm~\ref{Distributedzeroth-orderpolicygradientalgorithm} in a stochastic GridWorld environment~\citep{Zhangarxiv2024,Li2025ICLR,Kim2024arxiv}
{\color{blue}as well as a predator-prey environment, both of which corroborate the theoretical results.}
{\color{blue}The source code is publicly available at \texttt{https://github.com/Pengcheng-Dai/DZOPG}.}

\subsection{GridWorld environment}
We consider a cooperative task involving $N$ agents within a $5\times5$ GridWorld environment, where agents are randomly initialized {\color{blue}and collaboratively navigate toward a common goal.
The environment exhibits the following key characteristics.}
(i)~\textbf{Local observability}: each agent observes only its own position and makes decisions based solely on {\color{blue}local} information;
(ii)~\textbf{Stochastic dynamics}: agents select actions from a discrete set, {\color{blue}with movement subject to stochastic noise that diminishes as neighboring agents reach the goal};
(iii)~\textbf{Distance-aware cost}: {\color{blue}agents incur a per-step penalty proportional to their distance from the goal plus a fixed time cost, with zero reward upon and after reaching the target;}
(iv)~\textbf{Neighborhood communication}: a $\kappa$-hop communication {\color{blue}range enables} agents to exchange trajectory information for preference feedback.

\subsubsection{NMARL {\color{blue}setting}}\label{thesettingofNMARLinsimulation}
We now formally {\color{blue}formulate the NMARL model for this environment.}
\par
\textbf{State}. The local state $s_{i,t}$ of each agent $i\in\mathcal{N}=\{1,\cdots,N\}$ at time $t$ is defined as the 2-dimensional vector of its position in the grid.
The target is a specific location within the grid, defined as $s_{*}$.
\par
\textbf{Action}. At each time step, agent $i$ selects an action $a_{i,t} \in \mathcal{A}_i=\{\text{Up}, \text{Down}, \text{Left}, \text{Right}, \text{Stay}\}$, {\color{blue}corresponding to} the directional vectors $\{(0,1), (0,-1), (-1,0),(1,0),$
$ (0,0)\}$, respectively.
\par
\textbf{{\color{blue}Underlying network}}. Agents communicate via an {\color{blue}underlying} network $\mathcal{G}(\mathcal{N}, \mathcal{E})$.
Each agent $i$ has access to trajectory information from its $\kappa$-hop neighborhood $\mathcal{N}_i^\kappa$.
\par
\textbf{State transition function}. The state $s_{i,t+1}$ of agent $i$ evolves according to the following stochastic transition function:
\begin{align}
s_{i,t+1} =
\begin{cases}
s_{i,t}, & \text{if}~s_{i,t}=s_{*}\\
s_{i,t}+a_{i,t}+\varepsilon_{i,t}, & \text{otherwise},
\end{cases}
\end{align}
where $\varepsilon_{i,t}$ is a random perturbation drawn from $\{(0,0), (\pm1,0), (0,\pm1)\}$.
{\color{blue}The perturbation distribution governing agent
$i$'s movement depends on the number of its neighboring agents that have already reached the goal, formally defined as follows:}
\begin{align}
\mathbb{P}(\varepsilon_{i,t})=
\begin{cases}
1 - \chi_{i,t}, & \text{if}~\varepsilon_{i,t}=(0,0)\\
\frac{\chi_{i,t}}{4}, & \text{otherwise},
\end{cases}
\end{align}
where
$\chi_{i,t} = \chi_{\max} - (\chi_{\max} - \chi_{\min})\frac{|\{j|s_{j,t}=s_{*},\forall j\in\mathcal{N}_{i}\}|}{|\mathcal{N}_i|}$, {\color{blue}and $\chi_{\max}$ and $\chi_{\min}$ denote the maximum and minimum perturbation factors, respectively.}
\par
\textbf{Reward function}. The instantaneous reward for agent $i$ at time $t$ is defined as
\begin{align}
r_{i,t} =\left\{
\begin{array}{ll}
10, {\color{blue}\text{if}~s_{i,t}=s_{*}~\text{for the first time}} \\
0, {\color{blue}\text{if}~s_{i,t}=s_{*}~\text{after the first visit}} \\
-1 - \|s_{i,t}-s_{*}\|,\text{otherwise},
\end{array}
\right.
\end{align}
which penalizes both the time step cost $-1$ and the current distance of agent $i$ to the goal {\color{blue}before reaching it}.
{\color{blue}Specifically, the reward for an agent is $10$ when it reaches the goal for the first time, and becomes $0$ in all subsequent time steps while waiting for other agents to reach this goal.
This design encourages other agents to reach the goal as quickly as possible.}
\par
\textbf{Optimization objective}. {\color{blue}The local policy of agent $i$ is parameterized as a softmax policy, defined as}
\begin{align}
\pi_{i}(a_{i}|s_{i},\theta_{i})=&\frac{\exp{(\theta_{i,s_{i},a_{i}})}}{\sum_{a'_{i}}\exp{(\theta_{i,s_{i},a'_{i}})}},\label{policyexpress}
\end{align}
where $\theta_{i,s_{i},a_{i}}\in\mathbb{R}$ {\color{blue}denotes the policy parameter of agent $i$ associated with} the state-action pair $(s_{i},a_{i})$, and $\theta_{i}$ {\color{blue}is the concatenation of} $\theta_{i,s_{i},a_{i}}$ over all $(s_{i},a_{i})\in\mathcal{S}_{i}\times\mathcal{A}_{i}$.

The global objective {\color{blue}of agents} is to maximize {\color{blue}the objective function} $J(\bm{\theta})$ in (\ref{theobjectivefunction}).
\subsubsection{Performance of Algorithm~\ref{Distributedzeroth-orderpolicygradientalgorithm} {\color{blue}in GridWorld environment}}
In the simulation, we set $\mathcal{N}=\{1,2,3,4\}$ {\color{blue}with a chain communication network} $1 \leftrightarrow 2 \leftrightarrow 3 \leftrightarrow 4$.
The target position is $(4, 0)$, {\color{blue}and} the initial positions of agents are $(2, 1)$, $(3, 1)$, $(2, 2)$, and $(1, 0)$, respectively.
{\color{blue}The action perturbation bounds} are set to $\chi_{\text{max}} = 0.1$ and $\chi_{\text{min}} = 0.02$, and the discount factor is $\gamma = 0.9$.
{\color{blue}For Algorithm~\ref{Distributedzeroth-orderpolicygradientalgorithm}, we configure} $M = 1000$, $K = 500$, $H = 20$, $\kappa = 1$, $\alpha = 0.1$, $\mu = 0.1$, and $d_i = 125$ for all $i \in \mathcal{N}$, {\color{blue}with the local policy parameter of each agent initialized at} $\theta_{i,0} = \mathbf{0}_{125}$.
%
\begin{figure}[!ht]
    \centering
    \includegraphics[width=0.8\linewidth]{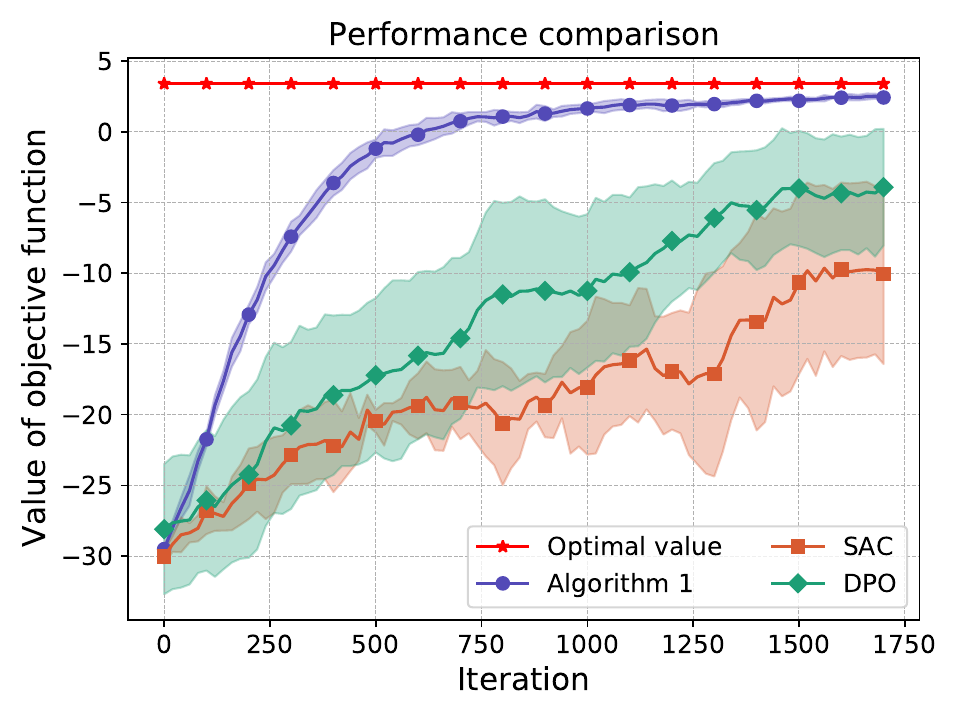}
    \caption{The evaluation of the performance of Algorithm~\ref{Distributedzeroth-orderpolicygradientalgorithm} with $M=1000$, $K=500$, $H=20$, {\color{blue}and $\kappa=1$}.}
    \label{fig:Optimal_performance}
\end{figure}
\par
{\color{blue}Since no existing distributed {\color{blue}preference-based policy optimization} algorithms are specifically tailored to the considered NMARL problem, we adopt the DPO algorithm~\citep{Rafailov2023} as a centralized preference-based baseline, and the scalable actor-critic (SAC) algorithm~\citep{QuCLDC2020} as a distributed reward-driven baseline.
In addition, to provide a reference for the best achievable performance, we include the optimal solution as a benchmark.
Fig.~\ref{fig:Optimal_performance} illustrates the performance of the proposed Algorithm~\ref{Distributedzeroth-orderpolicygradientalgorithm}, along with the DPO and SAC baselines, evaluated over 5 different random seeds.}
The solid {\color{blue}curve} represents the average performance, while the shaded area indicates the standard deviation.

{\color{blue}It can be observed that the proposed Algorithm~\ref{Distributedzeroth-orderpolicygradientalgorithm} closely approaches the optimal value, demonstrating} its effectiveness in achieving collaborative optimization solely based on human preference feedback.
{\color{blue}Moreover, compared with the DPO and SAC baselines, the proposed algorithm achieves superior performance and exhibits stable convergence behavior.}
{\color{blue}We note that a small performance gap between the proposed Algorithm~\ref{Distributedzeroth-orderpolicygradientalgorithm} and the optimal solution still exists.
This gap is primarily attributed to the distributed mechanism in the proposed algorithm, particularly the use of the truncated distance parameter $\kappa$ and the finite horizon length $H$ in the definition of the spatiotemporally truncated trajectory, which introduce an inherent approximation error.}

\subsubsection{{\color{blue}Ablation studies on sampling parameters}}
To investigate the influence of key algorithmic parameters on performance, we conduct comprehensive ablation studies on the length of trajectory $H$, {\color{blue}the} number of human preference queries $M$, and the number of trial $K$.
{\color{blue}All experiments are conducted with communication distance $\kappa=1$, and all results are averaged over 5~independent random seeds.}
\begin{figure}[!ht]
    \centering
    \includegraphics[width=0.8\linewidth]{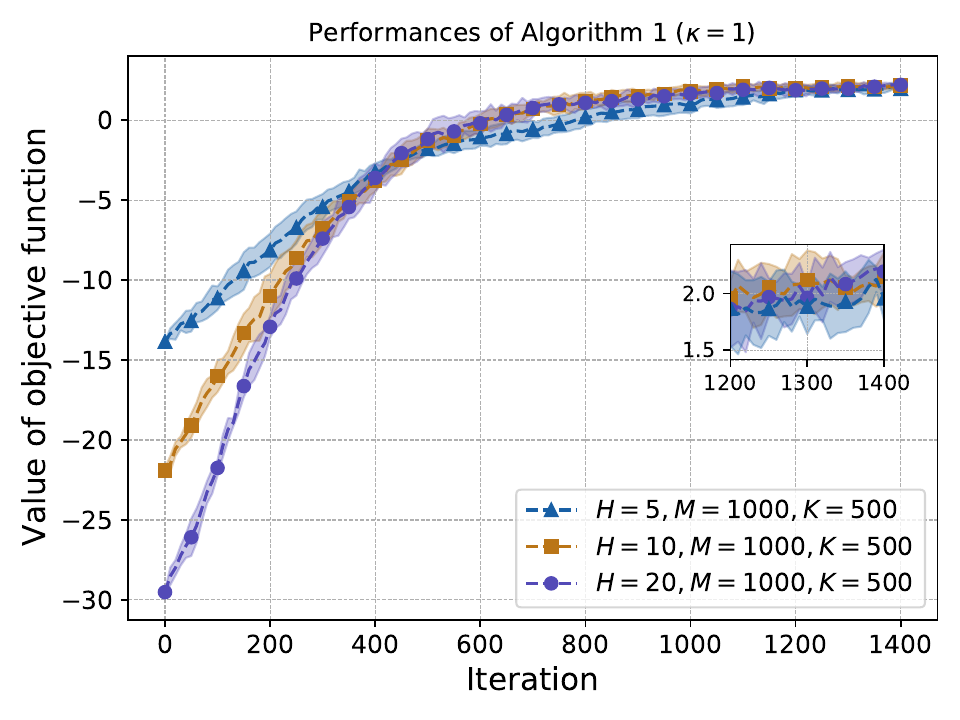}
    \caption{Ablation study on the number of trajectory length $H$ under fixed $M=1000$, and $K=500$.}
    \label{fig:ablation_H}
\end{figure}
\par
Fig.~\ref{fig:ablation_H} presents the performance of Algorithm~\ref{Distributedzeroth-orderpolicygradientalgorithm} with fixed $M=1000$, $K=500$, and varying horizon lengths $H\in\{5,10,20\}$.
{\color{blue}The results show that performance improves as $H$ increases, yet the differences are marginal.
This is primarily because agents are capable of completing the task within $H=4$ steps, meaning that even $H=5$ is sufficient for agents to collect enough informative trajectory data for effective policy optimization.}

%

\par
\begin{figure}[!ht]
    \centering
    \includegraphics[width=0.8\linewidth]{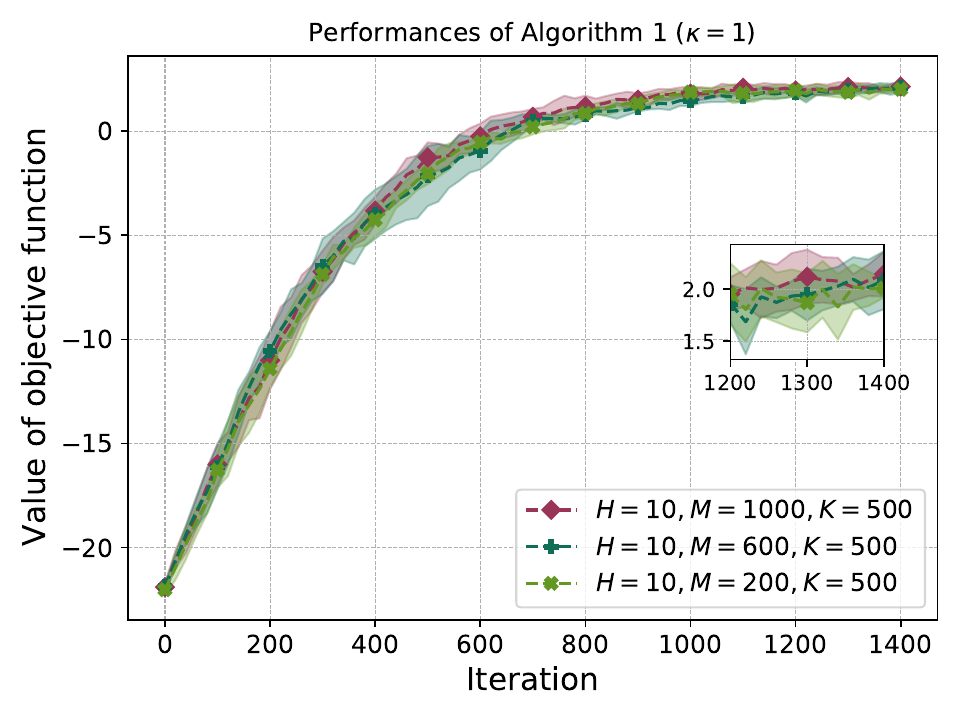}
    \caption{Ablation study on the number of preference queries $M$ under fixed $H=10$ and $K=500$.}
    \label{fig:ablation_M}
\end{figure}
\par
The performance of Algorithm~\ref{Distributedzeroth-orderpolicygradientalgorithm} with fixed $H=10$, $K=500$, and varying numbers of human evaluators $M\in\{200,600,1000\}$ is presented in Fig.~\ref{fig:ablation_M}.
{\color{blue}The results show that performance consistently improves with larger $M$, demonstrating the benefit of aggregating more human preference feedback in reducing estimation variance and stabilizing the learning process.}
\begin{figure}[!ht]
    \centering
    \includegraphics[width=0.8\linewidth]{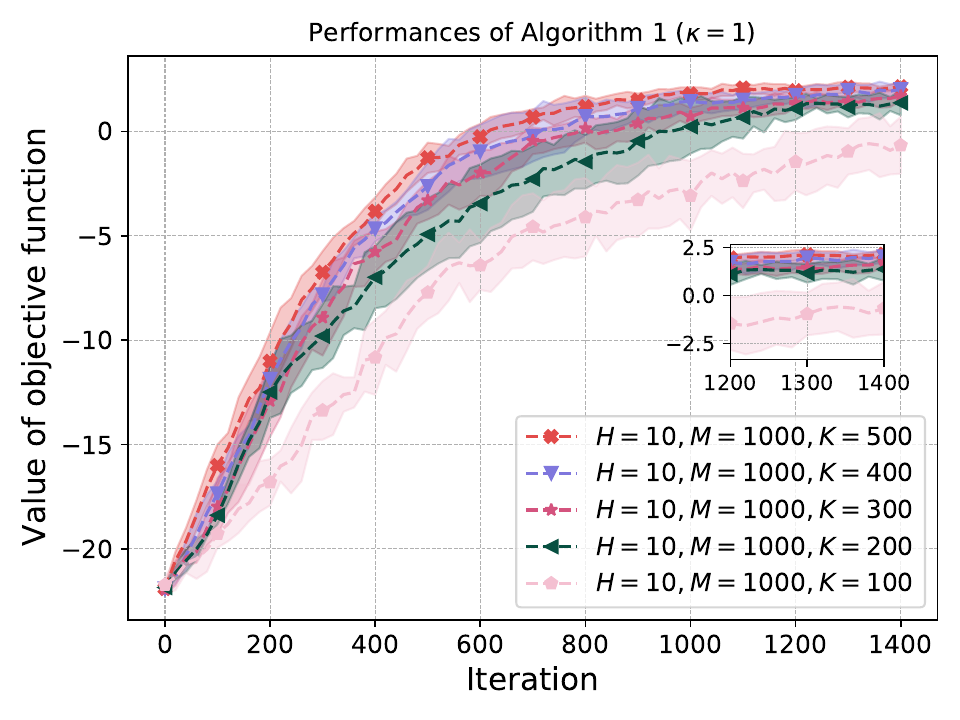}
    \caption{Ablation study on the number of trial $K$ under fixed $H=10$ and $M=1000$.}
    \label{fig:ablation_K}
\end{figure}
\par
Fig.~\ref{fig:ablation_K} presents the {\color{blue}performance} of Algorithm~\ref{Distributedzeroth-orderpolicygradientalgorithm} with fixed $H=10$ and $M=1000$ as well as varying values of $K\in\{100,200,300,400,500\}$.
As illustrated in Fig.~\ref{fig:ablation_K}, an increase in $K$ generally results in better performance, indicating that a larger number of trajectory samples contributes to the improved robustness of policy gradient estimation.
\par
Additionally, as illustrated in Figs.~\ref{fig:ablation_M}-\ref{fig:ablation_K}, Algorithm~\ref{Distributedzeroth-orderpolicygradientalgorithm} exhibits greater sensitivity to $K$ than to $M$.
{\color{blue}Specifically, with fixed $K$, varying $M$ yields consistent but modest performance gains, whereas with fixed $M$, increasing $K$ leads to more pronounced improvements.}
This empirical observation aligns with Theorem~\ref{thetheoremofconvergence}, which suggests that the error term $\frac{LR\sqrt{Nd_{\mathrm{tot}}}}{(1-\gamma)\sqrt{K}}$ associated with the number of trials $K$ {\color{blue}has a more dominant influence on convergence} than the preference estimation error $LNd_{\mathrm{tot}}L_{\sigma}\sqrt{\frac{\log{M}}{M}}$ {\color{blue}governed by} $M$.

%
%

\subsubsection{{\color{blue}Safety-aware NMARL setting}}
{\color{blue}To evaluate the effectiveness of the proposed Algorithm~\ref{Distributedzeroth-orderpolicygradientalgorithm} under safety considerations, we extend the GridWorld environment to include collision-avoidance penalties.
Specifically, whenever agents occupy the same grid location simultaneously, a negative reward is imposed, except at the target position, to discourage unsafe collisions.
This mechanism enforces a soft safety constraint that penalizes unsafe behaviors while preserving the tractability of the learning process, which is consistent with the principles of safety-aware multi-agent control.
The extended environment retains the original NMARL structure described in Section~\ref{thesettingofNMARLinsimulation}, with the sole modification in the reward function: for each agent $i$ at time step $t$, the reward is defined as
\begin{align}
r_{i,t} =\left\{
\begin{array}{ll}
10,~\text{if}~s_{i,t}=s_{*}~\text{for the first time} \\
0,~\text{if}~s_{i,t}=s_{*}~\text{after the first visit} \\
-1 - \|s_{i,t}-s_{*}\|,~\text{if}~s_{i,t}\neq s_{j,t},~\mathrm{for~all}~j\neq i\\
-1 - \|s_{i,t}-s_{*}\|-r_{\mathrm{collision}},~\text{otherwise},
\end{array}
\right.
\end{align}
where $r_{\mathrm{collision}}>0$ is a collision penalty coefficient.}
\par
\begin{figure}[!ht]
    \centering
    \includegraphics[width=0.8\linewidth]{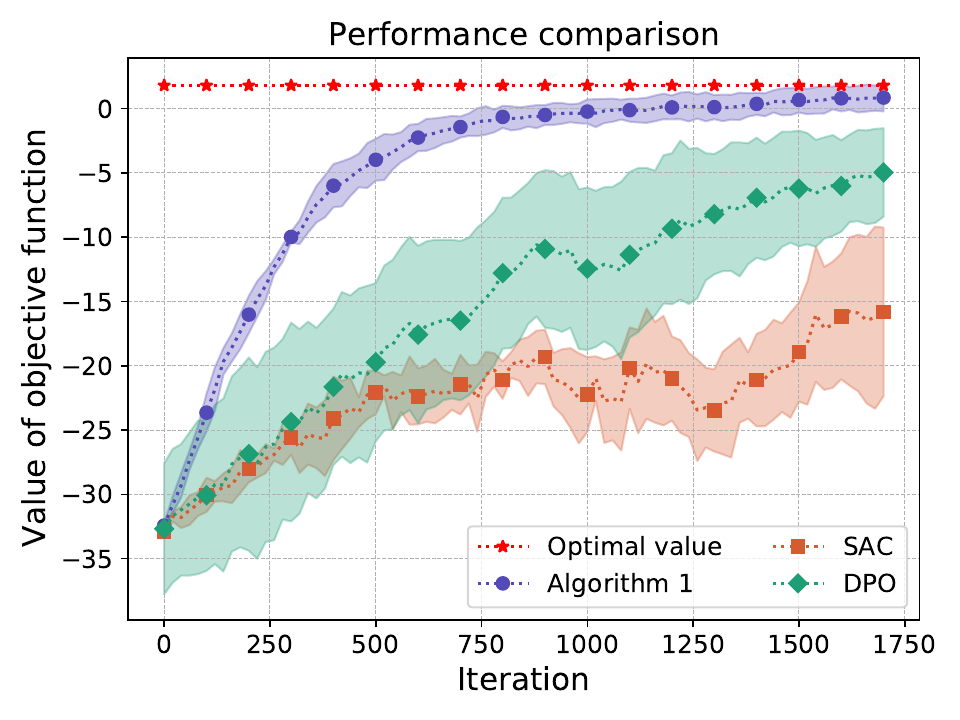}
    \caption{{\color{blue}The evaluation of the performance of Algorithm~\ref{Distributedzeroth-orderpolicygradientalgorithm} in the safety-aware GridWorld environment with $M=1000$, $K=500$, $H=20$, and $\kappa=1$.}}
    \label{fig:collision_performance}
\end{figure}
{\color{blue}We compare the proposed Algorithm~\ref{Distributedzeroth-orderpolicygradientalgorithm} against two representative baselines, DPO and SAC, with results averaged over $5$ independent random seeds presented in Fig.~\ref{fig:collision_performance}.
Algorithm~\ref{Distributedzeroth-orderpolicygradientalgorithm} consistently outperforms both baselines and converges to the corresponding optimal solution, which demonstrates that its effectiveness and performance advantages are retained even under safety considerations.}


\subsection{{\color{blue}Predator-prey environment}}
{\color{blue}In order to evaluate the scalability and effectiveness of the proposed Algorithm~\ref{Distributedzeroth-orderpolicygradientalgorithm}, we consider a more complex predator-prey environment.
The environment consists of a discrete $8\times 8$ grid populated by $N=20$ predators, $N_{p}=10$ prey, and $10$ static obstacles, where each predator is modeled as an agent.
\par
At each time step $t$, each agent $i$ observes a local state $s_{i,t}$ consisting solely of its own position and selects an action selects a movement action $a_{i,t}\in\{\text{Up}, \text{Down}, \text{Left}, \text{Right}, \text{Stay}\}$, subject to grid boundaries and obstacle constraints.
The reward $r_{i,t}$ of agent $i$ at time $t$ comprises three components: (i) a fixed time penalty to encourage efficiency, (ii) a shared capture reward distributed among agents that simultaneously occupy the location of a prey, and (iii) a distance-based shaping term that incentivizes progress toward the nearest prey.
Formally,
\begin{align}
r_{i,t} =& -r^{\mathrm{time}}+\sum_{p\in\mathcal{N}_{p,t}}\frac{r^{\mathrm{capture}}\mathds{1}\{s_{i,t}=s_{p}\}}{|\mathcal{C}_{p,t}|}\notag\\
&+\alpha_{p}(d_{i,t-1}-d_{i,t}),\label{therewarddefinitioninPredator-prey}
\end{align}
where $r^{\mathrm{time}}>0$ is the per-step time penalty, $\mathcal{N}_{p,t}$ is the set of remaining prey at time,
$s_p$ denotes the position of prey $p$,
$\mathcal{C}_{p,t}$ is the set of agents simultaneously capturing prey $p$ at time $t$,
$r^{\mathrm{capture}}>0$ is the capture reward,
$d_{i,t}$ denotes the distance from agent $i$ to its nearest prey,
and $\alpha_p>0$ is the scaling coefficient for the distance-based shaping term.
The objective of the agents is to cooperatively capture all prey as efficiently as possible, corresponding to maximizing the cumulative reward $J(\bm{\theta})$ defined in~(\ref{theobjectivefunction}).
}

\subsubsection{{\color{blue}Performance of Algorithm~\ref{Distributedzeroth-orderpolicygradientalgorithm} in predator-prey environment}}
{\color{blue}
In the predator-prey environment, we consider $N=20$ agents interconnected via a chain communication network $1\leftrightarrow 2\leftrightarrow\cdots\leftrightarrow 20$.
The initial positions of the agents are
$\{(4, 4), (5, 7), (4, 3), (1, 5), (4, 4), (5, 3),$
$(6, 4), (4, 0), (7, 3), (5, 2), (0, 5), (1, 0), (5, 2), (1, 4), (1, 0),$
$(2, 1),(5, 0), (4, 4), (3, 1), (6, 0)\}$, the prey are located at $\{(3, 2), (1, 3), (7, 0), (3, 0), (4, 5), (2, 0), (7, 5),(4, 2),$
$(3, 4), (3, 5)\}$, and the obstacles are placed at $\{(0, 1),$
$(5, 4), (6, 6), (7, 6), (6, 3), (7, 4), (3, 6), (2, 5), (2, 4), (6, 5)\}$.
The remaining environment parameters are set as follows: the per-step time penalty $r^{\mathrm{time}}=1$, the capture reward $r^{\mathrm{capture}}=1$, and the distance-based shaping coefficient $\alpha_p=0.5$.
}
\par
\begin{figure}[!ht]
    \centering
    \includegraphics[width=0.8\linewidth]{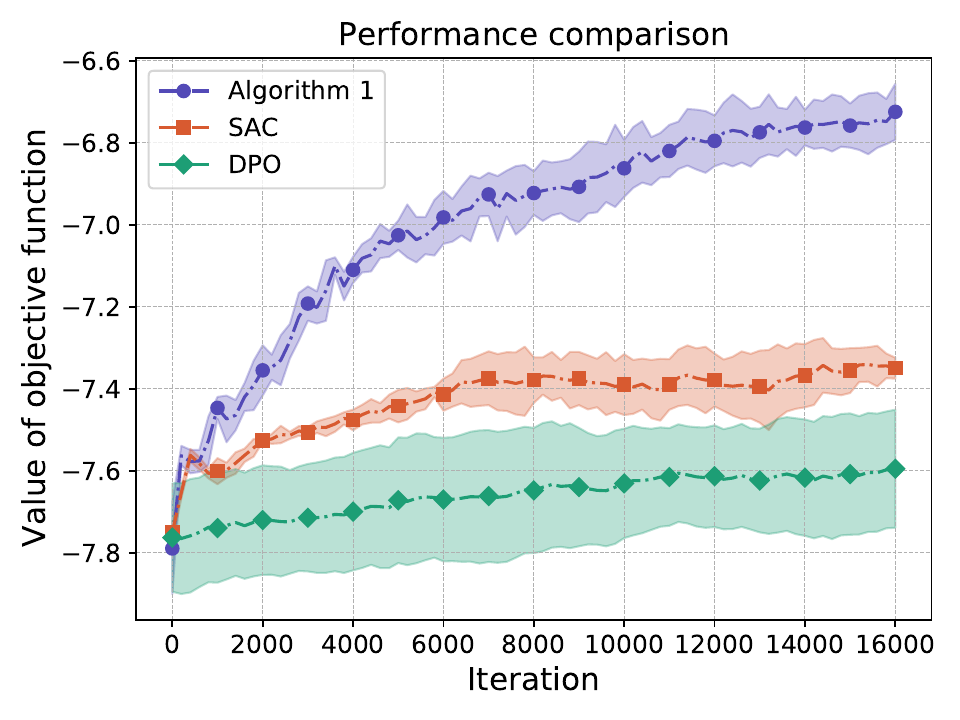}
    \caption{{\color{blue}The evaluation of the performance of Algorithm~\ref{Distributedzeroth-orderpolicygradientalgorithm} in the predator-prey environment with $M=200$, $K=100$, $H=50$, and $\kappa=1$.}}
    \label{fig:predator-prey_performance}
\end{figure}
{\color{blue}We evaluate Algorithm~\ref{Distributedzeroth-orderpolicygradientalgorithm} against DPO and SAC baselines, with results shown in Fig.~\ref{fig:predator-prey_performance}. Algorithm~\ref{Distributedzeroth-orderpolicygradientalgorithm} consistently achieves superior objective values, demonstrating its scalability and effectiveness in large-scale multi-agent scenarios.
Further implementation details are available at \texttt{https://github.com/Pengcheng-Dai/DZOPG}.
}


\subsubsection{{\color{blue}Ablation studies on sampling parameters}}
\begin{figure}[!ht]
    \centering
    \includegraphics[width=0.8\linewidth]{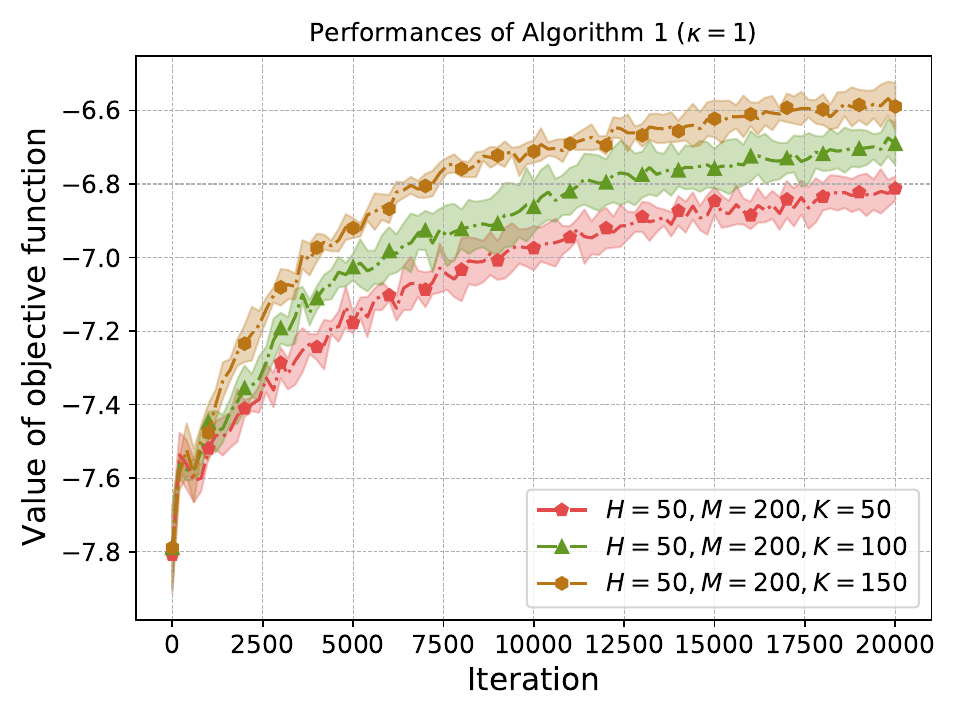}
    \caption{{\color{blue}Ablation study on the number of trial $K$ under fixed $H=50$ and $M=200$ in the predator-prey environment}.}
    \label{fig:predator-prey_ablation_K}
\end{figure}
{\color{blue}As established in Theorem~\ref{thetheoremofconvergence} and corroborated by the GridWorld results in Fig.~\ref{fig:ablation_K}, the sampling parameter $K$ plays a critical role in determining the performance of Algorithm~\ref{Distributedzeroth-orderpolicygradientalgorithm}.
To further investigate this sensitivity in a more complex setting, we evaluate the proposed algorithm under varying values of $K$ in the predator-prey environment.
As shown in Fig.~\ref{fig:predator-prey_ablation_K}, performance improves consistently with larger $K$, confirming that the number of trajectory samples for gradient estimation has a significant impact on the overall learning efficiency.}
\par
\begin{figure}[!ht]
    \centering
    \includegraphics[width=0.8\linewidth]{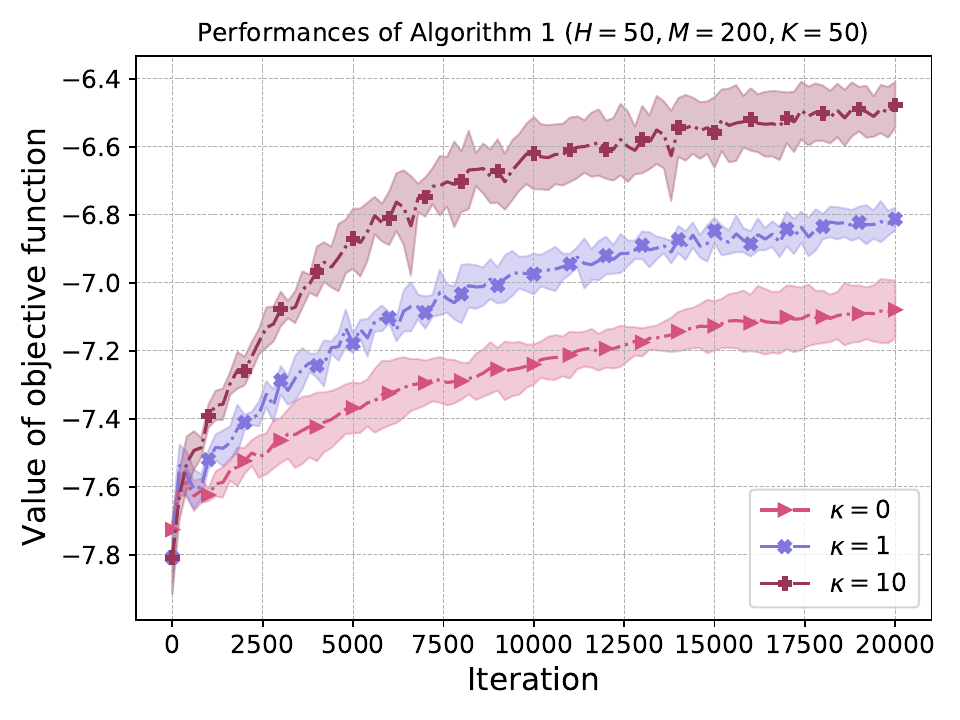}
    \caption{{\color{blue}Ablation study on the number of spatial truncated distance $\kappa$ in the predator-prey environment}.}
    \label{fig:predator-prey_ablation_Kappa}
\end{figure}
{\color{blue}An additional ablation study is conducted to investigate the influence of the parameter $\kappa$.
As illustrated in Fig.~\ref{fig:predator-prey_ablation_Kappa}, the algorithm's performance improves with increasing $\kappa$, which aligns with the theoretical result in Theorem~\ref{thetheoremofconvergence}.
This effect can be attributed to the fact that a larger $\kappa$ enables spatiotemporally truncated trajectories to incorporate information from a greater number of agents, thereby promoting coordinated optimization among the agents.}
{\color{blue}\begin{remark}
The parameters $H$ and $\kappa$ correspond to the temporal and spatial truncation lengths, respectively, in the proposed Algorithm~\ref{Distributedzeroth-orderpolicygradientalgorithm}. Intuitively, larger values of $H$ and $\kappa$ may improve algorithmic performance.
However, the appropriate choice of these parameters is closely related to the characteristics of the environment.
In general, $H$ should be set greater than the minimal horizon required for agents to reach the goal, so that the collected trajectories contain sufficiently informative samples.
The choice of $\kappa$ is determined by
the physical communication range of agents
and the degree of coupling among agents in the environment.
In strongly coupled multi-agent scenarios, a larger $\kappa$ is preferable, as it enables spatiotemporally truncated trajectories to capture agent interactions more effectively, thus facilitating coordinated behavior among agents.
\end{remark}}
%


\section{Conclusions}\label{SectionVConclusions}
This paper proposes a distributed zeroth-order policy gradient algorithm for the NMARL problem without relying on explicit reward signals.
In this approach, we introduce a novel human feedback mechanism based on spatiotemporally truncated trajectories and develop a Gaussian-perturbed gradient estimator to update the policy of agents.
The proposed algorithm demonstrates a provable polynomial sample complexity in achieving $\epsilon$-stationary convergence under mild assumptions, and exhibits strong empirical performance in stochastic environments.
Future work will explore several important directions.
{\color{blue}On the theoretical side, characterizing the robustness of the proposed framework to preference model misspecification and deriving misspecification-robust convergence guarantees remains an important open problem.
On the application side,} we will {\color{blue}investigate the deployment of the proposed algorithm} in various real-world scenarios, such as human-robot teaming in warehouse automation and collaborative industrial assembly.

\clearpage

\section{Appendix}

\subsection{Preliminary Lemma}\label{PreliminaryLemmas}
For a sequence of independent random variables $\{\zeta_k\}_{k=1}^{K}$, the following Rosenthal-type inequality for the fourth moment holds.
\begin{lemma}\label{lemmaofFourthmomentRosenthal-type}
Let $\{\zeta_k\}_{k=1}^{K}$ be a sequence of independent random variables and $\mathrm{Var}(\zeta_k)$ denote the variance of $\zeta_{k}$. If $\{\zeta_k\}_{k=1}^{K}$ satisfies:
(i) $\mathbb{E}[\zeta_k]=0$;
(ii) $|\zeta_k| \leq C_{\zeta}$ almost surely,
we can have the following result:
\begin{align}
\mathbb{E}\Big[\Big(\sum_{k=1}^{K} \zeta_k\Big)^4\Big]
\leq 3 \Big(\sum_{k=1}^{K} \mathrm{Var}(\zeta_k)\Big)^2 +
C_\zeta^2 \sum_{k=1}^{K}\mathrm{Var}(\zeta_k).\label{eq:rosenthal}
\end{align}
\end{lemma}
\begin{proof}
For the sequence of independent random variables $\{\zeta_k\}_{k=1}^{K}$, we have
\begin{align}
\mathbb{E}\Big[\Big(\sum_{k=1}^{K} \zeta_k\Big)^4\Big]
&= \sum_{k=1}^{K} \mathbb{E}[\zeta_k^4] + 6\!\!\sum_{1\leq i<j\leq K}\!\!\mathbb{E}[\zeta_i^2] \mathbb{E}[\zeta_j^2],\label{theFourthmomentRosenthal-type1}
\end{align}
where the equality can be obtained by the fact that the terms with odd powers vanish due to the zero-mean property, i.e., $\mathbb{E}[\zeta_k]=0$  for all $1\leq k\leq K$.
Next, we note that
\begin{align}
&\sum_{1\leq i<j\leq K} \mathbb{E}[\zeta_i^2]\mathbb{E}[\zeta_j^2]\notag\\
=&\frac{1}{2} \Big[ \Big( \sum_{k=1}^{K} \mathrm{Var}(\zeta_k)\Big)^2-\sum_{k=1}^{K}\mathrm{Var}(\zeta_k)^2\Big].\label{theFourthmomentRosenthal-type2}
\end{align}
\par
Substituting (\ref{theFourthmomentRosenthal-type2}) into (\ref{theFourthmomentRosenthal-type1}), we obtain
\begin{align}
&\mathbb{E}\Big[\Big(\sum_{k=1}^{K} \zeta_k\Big)^4\Big]\notag\\
=&\sum_{k=1}^{K} \mathbb{E}[\zeta_k^4] + 3\Big(\sum_{k=1}^{K}\mathrm{Var}(\zeta_k)\Big)^2-3\sum_{k=1}^{K}\mathrm{Var}(\zeta_k)^2\notag\\
\leq&3\Big( \sum_{k=1}^{K} \mathrm{Var}(\zeta_k) \Big)^2 + C_\zeta^2\sum_{k=1}^{K}\mathrm{Var}(\zeta_k),\label{theFourthmomentRosenthal-type3}
\end{align}
where (\ref{theFourthmomentRosenthal-type3}) uses the fact that $\mathbb{E}[\zeta_k^4] \le C_\zeta^2 \mathbb{E}[\zeta_k^2]=C_\zeta^2 \mathrm{Var}(\zeta_k)$.
\end{proof}
\par
This Rosenthal-type inequality for the fourth moment of independent, bounded, zero-mean random variables can be seen as a special case of the classical result in~\citep{Rosenthal1970}.


\subsection{Proof of Lemma~\ref{thelemmaofpreferenceestimation}}\label{ProofofLemmathelemmaofpreferenceestimation}
\begin{proof}
Similar to the analysis of Lemma~2 in~\citep{Zhangarxiv2024}, for any $\delta\in(0,1)$, we can easily obtain
\begin{align}\label{ConcentrationofPreferenceProbability}
\big|\hat{p}^{t,k}_{i}-\mathbb{P}\big(\tau^{t,k}_{\mathcal{N}^{\kappa}_{i},1}\succ\tau^{t,k}_{\mathcal{N}^{\kappa}_{i},0}\big)\big|
\leq\sqrt{\frac{\log{\frac{1}{\delta}}}{M}}
\end{align}
holds with probability at least $1-\delta$.
\par
Let $\mathcal{E}_{i,k}$ denote the event in (\ref{ConcentrationofPreferenceProbability}) holds and $\mathcal{E}^{\complement}_{i,k}$ represent the complementary event.
By Assumption~\ref{theassumptionofreward}, when the
concentration event $\mathcal{E}_{i,k}$ does not hold, we can have that:
\begin{align}\label{ConcentrationofPreferenceProbability1-1}
\big|\sigma^{-1}(\hat{p}^{t,k}_{i})-\big(\hat{r}_{i}(\tau^{t,k}_{\mathcal{N}^{\kappa}_{i},1})-\hat{r}_{i}(\tau^{t,k}_{\mathcal{N}^{\kappa}_{i},0})\big)\big|\leq\frac{4R(1-\gamma^{H})}{1-\gamma}.
\end{align}
When the concentration event $\mathcal{E}_{i,k}$ holds, we have
\begin{align}
&\big|\sigma^{-1}(\hat{p}^{t,k}_{i})-\big(\hat{r}_{i}(\tau^{t,k}_{\mathcal{N}^{\kappa}_{i},1})-\hat{r}_{i}(\tau^{t,k}_{\mathcal{N}^{\kappa}_{i},0})\big)\big|\notag\\
=&\big|\sigma^{-1}(\hat{p}^{t,k}_{i})-\sigma^{-1}\big(\mathbb{P}(\tau^{t,k}_{\mathcal{N}^{\kappa}_{i},1}\succ\tau^{t,k}_{\mathcal{N}^{\kappa}_{i},0})\big)\big|\label{ConcentrationofPreferenceProbability1-2}\\
\leq&L_{\sigma}\big|\hat{p}^{t,k}_{i}-\mathbb{P}(\tau^{t,k}_{\mathcal{N}^{\kappa}_{i},1}\succ\tau^{t,k}_{\mathcal{N}^{\kappa}_{i},0})\big|\label{ConcentrationofPreferenceProbability1-3}\\
\leq&L_{\sigma}\sqrt{\frac{\log{\frac{1}{\delta}}}{M}},\label{ConcentrationofPreferenceProbability1-4}
\end{align}
where (\ref{ConcentrationofPreferenceProbability1-2}) uses the fact that $\hat{r}_{i}(\tau^{t,k}_{\mathcal{N}^{\kappa}_{i},1})-\hat{r}_{i}(\tau^{t,k}_{\mathcal{N}^{\kappa}_{i},0})=\sigma^{-1}\big(\mathbb{P}(\tau^{t,k}_{\mathcal{N}^{\kappa}_{i},1}\succ\tau^{t,k}_{\mathcal{N}^{\kappa}_{i},0})\big)$, (\ref{ConcentrationofPreferenceProbability1-3}) can be achieved by {\color{blue}Assumption~\ref{theassumptionofsigma}}, and (\ref{ConcentrationofPreferenceProbability1-4}) uses the result in (\ref{ConcentrationofPreferenceProbability}).
Combining both cases and let $\delta=\frac{1}{M^{2}}$, we have
\begin{align}
&\mathbb{E}\big[\big|\sigma^{-1}(\hat{p}^{t,k}_{i})-\big(\hat{r}_{i}(\tau^{t,k}_{\mathcal{N}^{\kappa}_{i},1})-\hat{r}_{i}(\tau^{t,k}_{\mathcal{N}^{\kappa}_{i},0})\big)\big|\big]\notag\\
\leq&\mathbb{P}(\mathcal{E}_{i,k})L_{\sigma}\sqrt{\frac{\log{\frac{1}{\delta}}}{M}}+\mathbb{P}(\mathcal{E}^{\complement}_{i,k})\frac{4R(1-\gamma^{H})}{1-\gamma}\notag\\
\leq&L_{\sigma}\sqrt{\frac{2\log{M}}{M}}+\frac{4R(1-\gamma^{H})}{(1-\gamma)M^{2}},
\end{align}
which completes the proof of (\ref{thelemmaofpreferenceestimation(i)}) in Lemma~\ref{thelemmaofpreferenceestimation}.
\par
To prove (\ref{thelemmaofpreferenceestimation(iii)}), we apply a similar procedure.
When the event $\mathcal{E}^{\complement}_{i,k}$ holds, we have
\begin{align}\label{ConcentrationofPreferenceProbability2-1}
&\big|\sigma^{-1}(\hat{p}^{t,k}_{i})-\big(\hat{r}_{i}(\tau^{t,k}_{\mathcal{N}^{\kappa}_{i},1})-\hat{r}_{i}(\tau^{t,k}_{\mathcal{N}^{\kappa}_{i},0})\big)\big|^{4}\notag\\
\leq&\frac{256R^{4}(1-\gamma^{H})^{4}}{(1-\gamma)^{4}}.
\end{align}
When the event $\mathcal{E}_{i,k}$ holds, we can employ (\ref{ConcentrationofPreferenceProbability1-4}) and get
\begin{align}
\big|\sigma^{-1}(\hat{p}^{t,k}_{i})-\big(\hat{r}_{i}(\tau^{t,k}_{\mathcal{N}^{\kappa}_{i},1})-\hat{r}_{i}(\tau^{t,k}_{\mathcal{N}^{\kappa}_{i},0})\big)\big|^{4}\leq&L^{4}_{\sigma}\frac{(\log{\frac{1}{\delta}})^{2}}{M^{2}}.\label{ConcentrationofPreferenceProbability2-2}
\end{align}
Combining (\ref{ConcentrationofPreferenceProbability2-1}) and (\ref{ConcentrationofPreferenceProbability2-2}) and let $\delta=\frac{1}{M^{2}}$ again, we have
\begin{align}
&\mathbb{E}\big[\big|\sigma^{-1}(\hat{p}^{t,k}_{i})-\big(\hat{r}_{i}(\tau^{t,k}_{\mathcal{N}^{\kappa}_{i},1})-\hat{r}_{i}(\tau^{t,k}_{\mathcal{N}^{\kappa}_{i},0})\big)\big|^{4}\big]\notag\\
\leq&\frac{4L^{4}_{\sigma}(\log{M})^{2}}{M^{2}}+\frac{256R^{4}(1-\gamma^{H})^{4}}{(1-\gamma)^{4}M^{2}}.\notag
\end{align}
\end{proof}

 discussion of smooth of function
\subsection{Discussion of the rationality
Assumption~\ref{theassumptionofgradientofobjectivefunction}}\label{DiscussionoftherationalityAssumptiontheassumptionofgradientofobjectivefunction}
Recalling the definition of the softmax policy $\pi_{i}(a_{i}|s_{i},\theta_{i})$ in (\ref{policyexpress}), there exist positive constants $L^{s}_{b}$, $L^{s}_{g}$, and $L^{s}_{p}$ such that the joint policy $\bm{\pi}_{\bm{\theta}} = \prod_{i=1}^{N} \pi_{i}(a_{i}|s_{i},\theta_{i})$ satisfies the following properties:
\par
(i) The gradient of the joint log-policy is bounded, i.e., $\|\nabla_{\bm{\theta}} \log\bm{\pi}_{\bm{\theta}}(\bm{a}|\bm{s})\|\leq L^{s}_{b}$;
\par
(ii) The gradient of the joint log-policy is Lipschitz, i.e., $\|\nabla_{\bm{\theta}}\log\bm{\pi}_{\bm{\theta}}(\bm{a}|\bm{s})- \nabla_{\bm{\theta}} \log \bm{\pi}_{\bm{\theta}'}(\bm{a}|\bm{s})\| \leq L^{s}_{g}\|\bm{\theta}-\bm{\theta}'\|$;
\par
(iii) The joint policy is Lipschitz in total variation, i.e., $\|\bm{\pi}_{\bm{\theta}}(\cdot|\bm{s}) - \bm{\pi}_{\bm{\theta}'}(\cdot|\bm{s})\|_{\rm TV} \leq L^{s}_{p}\|\bm{\theta}-\bm{\theta}'\|$ for all $\bm{s}\in\bm{\mathcal{S}}$.
\par
These results follow directly from the properties of individual softmax policies.
Due to space limitations, the detailed derivations are omitted.
\par
Based on the above properties, $\widehat{J}_{i}(\bm{\theta})$ in (\ref{thetruncatedobjectivefunction}) has the following property.
\begin{proposition}
Suppose Assumption~\ref{theassumptionofreward} holds. For joint softmax policy $\bm{\pi_{\theta}}$,   $\nabla_{\bm{\theta}}\widehat{J}_{i}(\bm{\theta})$ is Lipschitz continuity.
\end{proposition}
\begin{proof}
Following the definition of $\widehat{J}_{i}(\bm{\theta})$ in (\ref{thetruncatedobjectivefunction}),
we define the associated $Q$-function accordingly
\begin{align}\label{thetruncatedobjectivefunctionQfunction}
\widehat{Q}^{\bm{\pi_{\theta}}}_{i,t}(\bm{s}_{t},\bm{a}_{t}) =&\mathbb{E}_{\bm{\pi_{\theta}}}\Big[\frac{1}{N}\sum_{h=t}^{H-1}\sum_{j\in\mathcal{N}^{\kappa}_{i}}\!\!\gamma^{h-t}r_{j,h}\Big|\bm{s}_{h}=\bm{s}_{t},\notag\\
&\bm{a}_{h}=\bm{a}_{t}\Big].
\end{align}
For a trajectory $\bm{\tau} = (\bm{s}_0, \bm{a}_0, \dots, \bm{s}_{H-1}, \bm{a}_{H-1})$, we define
\begin{align}\label{thefunctionofG}
G^{\bm{\pi_{\theta}}}_{i}(\bm{\tau})=\sum_{t=0}^{H-1}\gamma^t \widehat{Q}_{i,t}^{\bm{\pi_\theta}}(\bm{s}_t,\bm{a}_t)\nabla_{\bm{\theta}}\log\bm{\pi_\theta}(\bm{a}_t|\bm{s}_t).
\end{align}
The policy gradient $\nabla_{\bm{\theta}}\widehat{J}_{i}(\bm{\theta})$ is then expressed as
\begin{align}\label{thetruncatedobjectivefunctiongradient}
\nabla_{\bm{\theta}}\widehat{J}_{i}(\bm{\theta})=\mathbb{E}_{\bm{\tau}\sim\bm{\pi_{\theta}}}\Big[G^{\bm{\pi_{\theta}}}_{i}(\bm{\tau})\Big].
\end{align}
For any policy parameters $\bm{\theta},\bm{\theta}'$, we have
\begin{align}
&\nabla_{\bm{\theta}}\widehat{J}_{i}(\bm{\theta})-\nabla_{\bm{\theta}}\widehat{J}_{i}(\bm{\theta}')\notag\\
=&\underbrace{\mathbb{E}_{\bm{\tau}\sim\bm{\pi_\theta}}\big[G^{\bm{\pi_{\theta}}}_{i}(\bm{\tau})-G^{\bm{\pi}_{\bm{\theta}'}}_{i}(\bm{\tau})\big]}_{\mathrm{(i)}}\notag\\
&+\underbrace{\mathbb{E}_{\bm{\tau}\sim\bm{\pi}_{\bm{\theta}}}\big[G^{\bm{\pi}_{\bm{\theta}'}}_{i}(\bm{\tau})\big]-\mathbb{E}_{\bm{\tau}\sim\bm{\pi}_{\bm{\theta}'}}\big[G^{\bm{\pi}_{\bm{\theta}'}}_{i}(\bm{\tau})\big]}_{\mathrm{(ii)}}.\label{thetruncatedobjectivefunction1}
\end{align}
\par
In the $\mathrm{(i)}$-term of (\ref{thetruncatedobjectivefunction1}), we use the definition of $G^{\bm{\pi_{\theta}}}_{i}(\bm{\tau})$ in (\ref{thefunctionofG}) and have
\begin{align}
&G^{\bm{\pi_{\theta}}}_{i}(\bm{\tau})-G^{\bm{\pi}_{\bm{\theta}'}}_{i}(\bm{\tau})\notag\\
=&\sum_{t=0}^{H-1}\gamma^t\Big[
\big(\widehat{Q}^{\bm{\pi_{\theta}}}_{i}(\bm{s}_{t},\bm{a}_{t})-\widehat{Q}^{\bm{\pi}_{\bm{\theta}'}}_{i}(\bm{s}_{t},\bm{a}_{t})\big)\nabla_{\bm{\theta}}\log\bm{\pi}_{\bm{\theta}'}(\bm{a}_t|\bm{s}_t)\notag\\
&+\!\widehat{Q}^{\bm{\pi_{\theta}}}_{i}(\bm{s}_{t},\bm{a}_{t})\big(\nabla_{\bm{\theta}}\log\bm{\pi}_{\bm{\theta}}(\bm{a}_{t}|\bm{s}_{t})\!-\!\nabla_{\bm{\theta}}\log\bm{\pi}_{\bm{\theta}'}(\bm{a}_{t}|\bm{s}_{t})\big)\Big].\label{thetruncatedobjectivefunction1-1-1}
\end{align}
For any joint policy $\bm{\pi_{\theta}}$, let $P^{\bm{\pi}_{\bm{\theta}}}(\bm{s}',\bm{a}'|\bm{s},\bm{a})=\bm{\mathcal{P}}(\bm{s}'|\bm{s},\bm{a})$ $\bm{\pi}_{\bm{\theta}}(\bm{a}'|\bm{s}')$.
By the definition of $\widehat{Q}^{\bm{\pi_{\theta}}}_{i,t}(\bm{s}_{t},\bm{a}_{t})$ in (\ref{thetruncatedobjectivefunctionQfunction}), we have
\begin{align}
&\big|\widehat{Q}^{\bm{\pi_{\theta}}}_{i}(\bm{s}_{t},\bm{a}_{t})-\widehat{Q}^{\bm{\pi}_{\bm{\theta}'}}_{i}(\bm{s}_{t},\bm{a}_{t})\big|\notag\\
\leq&\gamma\sum_{\bm{s}_{t+1},\bm{a}_{t+1}}\Big|P^{\bm{\pi}_{\bm{\theta}}}(\bm{s}_{t+1},\bm{a}_{t+1}|\bm{s}_{t},\bm{a}_{t}) \widehat{Q}^{\bm{\pi_{\theta}}}_{i}(\bm{s}_{t+1},\bm{a}_{t+1}) \notag\\
&-P^{\bm{\pi}_{\bm{\theta}'}}(\bm{s}_{t+1},\bm{a}_{t+1}|\bm{s}_{t},\bm{a}_{t}) \widehat{Q}^{\bm{\pi}_{\bm{\theta}'}}_{i}(\bm{s}_{t+1},\bm{a}_{t+1})\Big|\notag\\
\leq&\gamma\sum_{\bm{s}_{t+1},\bm{a}_{t+1}}\Big|P^{\bm{\pi}_{\bm{\theta}}}(\bm{s}_{t+1},\bm{a}_{t+1}|\bm{s}_{t},\bm{a}_{t}) \Big(\widehat{Q}^{\bm{\pi_{\theta}}}_{i}(\bm{s}_{t+1},\bm{a}_{t+1})\notag\\
&-\widehat{Q}^{\bm{\pi}_{\bm{\theta}'}}_{i}(\bm{s}_{t+1},\bm{a}_{t+1})\Big) \notag\\
&+\Big(P^{\bm{\pi}_{\bm{\theta}}}(\bm{s}_{t+1},\bm{a}_{t+1}|\bm{s}_{t},\bm{a}_{t})-P^{\bm{\pi}_{\bm{\theta}'}}(\bm{s}_{t+1},\bm{a}_{t+1}|\bm{s}_{t},\bm{a}_{t}) \Big)\notag\\
&\times\widehat{Q}^{\bm{\pi}_{\bm{\theta}'}}_{i}(\bm{s}_{t+1},\bm{a}_{t+1})\Big|\notag\\
\leq&\gamma\max_{\bm{s}_{t+1},\bm{a}_{t+1}}\big|\widehat{Q}^{\bm{\pi_{\theta}}}_{i}(\bm{s}_{t+1},\bm{a}_{t+1})-\widehat{Q}^{\bm{\pi}_{\bm{\theta}'}}_{i}(\bm{s}_{t+1},\bm{a}_{t+1})\big|\notag\\
&+\gamma\sum_{\bm{s}_{t+1},\bm{a}_{t+1}}\bm{\mathcal{P}}(\bm{s}_{t+1}|\bm{s}_{t},\bm{a}_{t})\Big|\bm{\pi}_{\bm{\theta}}(\bm{a}_{t+1}|\bm{s}_{t+1})\notag\\
&-\bm{\pi}_{\bm{\theta}'}(\bm{a}_{t+1}|\bm{s}_{t+1})\Big|\Big|\widehat{Q}^{\bm{\pi}_{\bm{\theta}'}}_{i}(\bm{s}_{t+1},\bm{a}_{t+1})\Big|\notag\\
\leq&\gamma\max_{\bm{s}_{t+1},\bm{a}_{t+1}}\big|\widehat{Q}^{\bm{\pi_{\theta}}}_{i}(\bm{s}_{t+1},\bm{a}_{t+1})-\widehat{Q}^{\bm{\pi}_{\bm{\theta}'}}_{i}(\bm{s}_{t+1},\bm{a}_{t+1})\big|\notag\\
&+\sum_{\bm{s}_{t+1}}\bm{\mathcal{P}}(\bm{s}_{t+1}|\bm{s}_{t},\bm{a}_{t})\|\bm{\pi}_{\bm{\theta}}(\cdot|\bm{s}_{t+1})-\bm{\pi}_{\bm{\theta}'}(\cdot|\bm{s}_{t+1})\|_{1}\notag\\
&\times\sum_{h=t+1}^{H-1}\gamma^{h-t}R\label{thetruncatedobjectivefunction1-1-2}\\
\leq&\gamma\max_{\bm{s}_{t+1},\bm{a}_{t+1}}\big|\widehat{Q}^{\bm{\pi_{\theta}}}_{i}(\bm{s}_{t+1},\bm{a}_{t+1})-\widehat{Q}^{\bm{\pi}_{\bm{\theta}'}}_{i}(\bm{s}_{t+1},\bm{a}_{t+1})\big|\notag\\
&+\frac{2\gamma L^{s}_{p}R}{1-\gamma}\|\bm{\theta}-\bm{\theta}'\|,\label{thetruncatedobjectivefunction1-1-3}
\end{align}
where (\ref{thetruncatedobjectivefunction1-1-2}) uses the fact that $|\widehat{Q}^{\bm{\pi}_{\bm{\theta}'}}_{i}(\bm{s}_{t+1},\bm{a}_{t+1})|\leq\sum_{h=t+1}^{H-1}\gamma^{h-t-1}R$ and (\ref{thetruncatedobjectivefunction1-1-3}) use $\|\bm{\pi}_{\bm{\theta}}(\cdot|\bm{s}_{t+1})-\bm{\pi}_{\bm{\theta}'}(\cdot|\bm{s}_{t+1})\|_{1}$ $=2\|\bm{\pi}_{\bm{\theta}}(\cdot|\bm{s}_{t+1})-\bm{\pi}_{\bm{\theta}'}(\cdot|\bm{s}_{t+1})\|_{\mathrm{TV}}\leq2L^{s}_{p}\|\bm{\theta}-\bm{\theta}'\|$.
Unfolding the recursion (\ref{thetruncatedobjectivefunction1-1-3}) until horizon $H-1$, we obtain
\begin{align}\label{thetruncatedobjectivefunction1-1-4}
\big|\widehat{Q}^{\bm{\pi_{\theta}}}_{i}(\bm{s}_{t},\bm{a}_{t})-\widehat{Q}^{\bm{\pi}_{\bm{\theta}'}}_{i}(\bm{s}_{t},\bm{a}_{t})\big|\leq\frac{2L^{s}_{p}R}{(1-\gamma)^{2}}\|\bm{\theta}-\bm{\theta}'\|.
\end{align}
Substituting (\ref{thetruncatedobjectivefunction1-1-4}) into (\ref{thetruncatedobjectivefunction1-1-1}) and using properties (i)-(iii) of $\bm{\pi}_{\bm{\theta}}$, we have
\begin{align}\label{thetruncatedobjectivefunction1-1-5}
G^{\bm{\pi_{\theta}}}_{i}(\bm{\tau})-G^{\bm{\pi}_{\bm{\theta}'}}_{i}(\bm{\tau})\leq&\Big(\frac{L^{s}_{g}R}{(1-\gamma)^{2}}+\frac{2L^{s}_{b}L^{s}_{p}R}{(1-\gamma)^{3}}\Big)\Big\|\bm{\theta}-\bm{\theta}'\Big\|.
\end{align}
\par
For the $\mathrm{(ii)}$-term of (\ref{thetruncatedobjectivefunction1}), we define
$G^{\bm{\pi_{\theta}}}_{i,t}(\bm{s}_{t},\bm{a}_{t})=\gamma^{t}\widehat{Q}_{i,t}^{\bm{\pi_\theta}}(\bm{s}_t,\bm{a}_t)\nabla_{\bm{\theta}}\log\bm{\pi_\theta}(\bm{a}_t|\bm{s}_t)$ and $P^{\bm{\pi}_{\bm{\theta}}}_{t}(\bm{s}_{t},\bm{a}_{t})$ as the probability of occurrence of $(\bm{s}_{t},\bm{a}_{t})$ at time $t$ under joint policy $\bm{\pi_{\theta}}$ and initial state distribution $\bm{\rho}$.
Then, the $\mathrm{(ii)}$-term of (\ref{thetruncatedobjectivefunction1}) can be rewritten as
\begin{align}
&\mathbb{E}_{\bm{\tau}\sim\bm{\pi}_{\bm{\theta}}}\big[G^{\bm{\pi}_{\bm{\theta}'}}_{i}(\bm{\tau})\big]-\mathbb{E}_{\bm{\tau}\sim\bm{\pi}_{\bm{\theta}'}}\big[G^{\bm{\pi}_{\bm{\theta}'}}_{i}(\bm{\tau})\big]\notag\\
=&\sum_{t=0}^{H-1}\big(\mathbb{E}_{\bm{\pi}_{\bm{\theta}}}[G^{\bm{\pi}_{\bm{\theta}'}}_{i,t}(\bm{s}_{t},\bm{a}_{t})]-\mathbb{E}_{\bm{\pi}_{\bm{\theta}'}}[G^{\bm{\pi}_{\bm{\theta}'}}_{i,t}(\bm{s}_{t},\bm{a}_{t})]\big)\notag\\
=&\sum_{t=0}^{H-1}\sum_{\bm{s}_{t},\bm{a}_{t}}\big(P^{\bm{\pi}_{\bm{\theta}}}_{t}(\bm{s}_{t},\bm{a}_{t})-P^{\bm{\pi}_{\bm{\theta}'}}_{t}(\bm{s}_{t},\bm{a}_{t})\big)G^{\bm{\pi}_{\bm{\theta}'}}_{i,t}(\bm{s}_{t},\bm{a}_{t})\Big)\notag\\
\leq&\sum_{t=0}^{H-1}\frac{2L^{s}_{b}R}{1-\gamma}\|P^{\bm{\pi}_{\bm{\theta}}}_{t}-P^{\bm{\pi}_{\bm{\theta}'}}_{t}\|_{\mathrm{TV}},\label{thetruncatedobjectivefunction1-2-1}
\end{align}
where $P^{\bm{\pi}_{\bm{\theta}}}_{t}$ and $P^{\bm{\pi}_{\bm{\theta}'}}_{t}$ are vectors composed of elements $\{P^{\bm{\pi}_{\bm{\theta}}}_{t}(\bm{s}_{t},\bm{a}_{t})\}_{\bm{s}_{t},\bm{a}_{t}}$ and $\{P^{\bm{\pi}_{\bm{\theta}'}}_{t}(\bm{s}_{t},\bm{a}_{t})\}_{\bm{s}_{t},\bm{a}_{t}}$, respectively, and the last inequality follows from the fact that $|G^{\bm{\pi}_{\bm{\theta}'}}_{i,t}(\bm{s}_{t},\bm{a}_{t})|\leq\frac{L^{s}_{b}R}{1-\gamma}$. Since $P^{\bm{\pi}_{\bm{\theta}}}_{t}(\bm{s}_{t},\bm{a}_{t})=\sum_{\bm{s}_{t-1},\bm{a}_{t-1}}P^{\bm{\pi}_{\bm{\theta}}}_{t-1}(\bm{s}_{t-1},\bm{a}_{t-1})P^{\bm{\pi}_{\bm{\theta}}}(\bm{s}_{t},\bm{a}_{t}|\bm{s}_{t-1},\bm{a}_{t-1})$, we can derive that
\begin{align}
&\|P^{\bm{\pi}_{\bm{\theta}}}_{t}-P^{\bm{\pi}_{\bm{\theta}'}}_{t}\|_{\mathrm{TV}}\notag\\
\leq&\|P^{\bm{\pi}_{\bm{\theta}}}_{t-1}-P^{\bm{\pi}_{\bm{\theta}'}}_{t-1}\|_{\mathrm{TV}}+\max_{\bm{s}_{t-1},\bm{a}_{t-1}}\big\|P^{\bm{\pi}_{\bm{\theta}}}(\cdot,\cdot|\bm{s}_{t-1},\bm{a}_{t-1})\notag\\
&-P^{\bm{\pi}_{\bm{\theta}'}}(\cdot,\cdot|\bm{s}_{t-1},\bm{a}_{t-1})\big\|_{\mathrm{TV}}\notag\\
\leq&\|P^{\bm{\pi}_{\bm{\theta}}}_{t-1}-P^{\bm{\pi}_{\bm{\theta}'}}_{t-1}\|_{\mathrm{TV}}+\max_{\bm{s}_{t}}\|\bm{\pi}_{\bm{\theta}}(\cdot|\bm{s}_{t})-\bm{\pi}_{\bm{\theta}'}(\cdot|\bm{s}_{t})\|_{\mathrm{TV}}\label{thetruncatedobjectivefunction1-2-3}\\
\leq&\|P^{\bm{\pi}_{\bm{\theta}}}_{t-1}-P^{\bm{\pi}_{\bm{\theta}'}}_{t-1}\|_{\mathrm{TV}}+L^{s}_{p}\|\bm{\theta}-\bm{\theta}'\|\label{thetruncatedobjectivefunction1-2-4}\\
\vdots\notag\\
\leq& L^{s}_{p}t\|\bm{\theta}-\bm{\theta}'\|,\label{thetruncatedobjectivefunction1-2-5}
\end{align}
where (\ref{thetruncatedobjectivefunction1-2-3}) uses the fact that $P^{\bm{\pi}_{\bm{\theta}}}(\bm{s}_{t},\bm{a}_{t}|\bm{s}_{t-1},\bm{a}_{t-1})=\bm{\mathcal{P}}(\bm{s}_{t}|\bm{s}_{t-1},\bm{a}_{t-1})\bm{\pi}_{\bm{\theta}}(\bm{a}_{t}|\bm{s}_{t})$ and (\ref{thetruncatedobjectivefunction1-2-4}) uses $\|\bm{\pi}_{\bm{\theta}}(\cdot|\bm{s}) - \bm{\pi}_{\bm{\theta}'}(\cdot|\bm{s})\|_{\rm TV} \leq L^{s}_{p}\|\bm{\theta}-\bm{\theta}'\|$ for all $\bm{s}\in\bm{\mathcal{S}}$.\\
Substituting (\ref{thetruncatedobjectivefunction1-2-5}) into (\ref{thetruncatedobjectivefunction1-2-1}), we further have
\begin{align}\label{thetruncatedobjectivefunction1-2-6}
&\mathbb{E}_{\bm{\tau}\sim\bm{\pi}_{\bm{\theta}}}\big[G^{\bm{\pi}_{\bm{\theta}'}}_{i}(\bm{\tau})\big]-\mathbb{E}_{\bm{\tau}\sim\bm{\pi}_{\bm{\theta}'}}\big[G^{\bm{\pi}_{\bm{\theta}'}}_{i}(\bm{\tau})\big]\notag\\
\leq&\frac{L^{s}_{b}L^{s}_{p}(H-1)HR}{1-\gamma}\Big\|\bm{\theta}-\bm{\theta}'\Big\|.
\end{align}
\par
Combining (\ref{thetruncatedobjectivefunction1}), (\ref{thetruncatedobjectivefunction1-1-5}), and (\ref{thetruncatedobjectivefunction1-2-6}), we can obtain that $\nabla_{\bm{\theta}}\widehat{J}_{i}(\bm{\theta})$ is Lipschitz continuity.
\end{proof}

\subsection{Proof of Theorem~\ref{lemmaoftruncatederror}}\label{ProofofLemmalemmaoftruncatederror}
\begin{proof}
Before formally proving this theorem, let's first introduce some auxiliary definitions.
\par
Different to the objective function $J(\bm{\theta})$ in (\ref{theobjectivefunction}), we define $\widetilde{J}_{i}(\bm{\theta})$ as the spatially truncated objective function associated with agent $i$, which {\color{blue}is} expressed as
\begin{align}\label{thespatiallytruncatedobjectivefunction}
\widetilde{J}_{i}(\bm{\theta})=\mathbb{E}_{\bm{s}\sim\bm{\rho}}\Big[\frac{1}{N}\sum_{t=0}^{\infty}\sum_{j\in\mathcal{N}^{\kappa}_{i}}\gamma^{t}r_{j,t}\Big|\bm{s}_{0}=\bm{s},\bm{a}_{t}\sim\bm{\pi_{\theta}}(\cdot|\bm{s}_{t})\Big].
\end{align}
By using $\nabla_{\theta_{i}}\widetilde{J}_{i}(\bm{\theta})$ as the intermediate term for analyzing the errors {\color{blue}between} $\nabla_{\theta_{i}}\widehat{J}_{i}(\bm{\theta})$ and $\nabla_{\theta_{i}}J(\bm{\theta})$, we have
\begin{align}
&\|\nabla_{\theta_{i}}\widehat{J}_{i}(\bm{\theta})-\nabla_{\theta_{i}}J(\bm{\theta})\|\notag\\
\leq&\underbrace{\|\nabla_{\theta_{i}}\widehat{J}_{i}(\bm{\theta})-\nabla_{\theta_{i}}\widetilde{J}_{i}(\bm{\theta})\|}_{\mathrm{(i)}}+\underbrace{\|\nabla_{\theta_{i}}\widetilde{J}_{i}(\bm{\theta})-\nabla_{\theta_{i}}J(\bm{\theta})\|}_{\mathrm{(ii)}}.\label{theorem1bu1}
\end{align}
\par
For $\mathrm{(i)}$-term in (\ref{theorem1bu1}), we revisit the Monte Carlo policy gradient formulation presented in~\citep{Sutton1998} and obtain $\nabla_{\theta_{i}}\widetilde{J}_{i}(\bm{\theta})=\mathbb{E}_{\bm{s}_t,\bm{a}_t\sim\bm{\pi}_{\bm{\theta}}}\big[\frac{1}{N}\sum_{t=0}^{\infty}$ $\nabla_{\theta_i}\log\pi_{i}(a_{i,t}|s_{i,t},\theta_{i})\sum_{j\in\mathcal{N}^{\kappa}_{i}}\sum_{l=0}^{\infty} \gamma^{t+l}$
$r_{j,t+l}\big]$ and $\nabla_{\theta_{i}}\widehat{J}_{i}(\bm{\theta})=\mathbb{E}_{\bm{s}_t,\bm{a}_t\sim\bm{\pi}_{\bm{\theta}}}\big[\frac{1}{N}\sum_{t=0}^{H-1}$
$\nabla_{\theta_i}\log\pi_{i}(a_{i,t}|s_{i,t},\theta_{i})$
$\sum_{j\in\mathcal{N}^{\kappa}_{i}}\sum_{l=0}^{H-1-t} \gamma^{t+l}r_{j,t+l}\big]$.
Building upon these expressions, we have
\begin{align}
&\|\nabla_{\theta_{i}}\widehat{J}_{i}(\bm{\theta})-\nabla_{\theta_{i}}\widetilde{J}_{i}(\bm{\theta})\|\notag\\
=&\Big\|\mathbb{E}_{\bm{s}_t,\bm{a}_t\sim\bm{\pi}_{\bm{\theta}}}\Big[\frac{1}{N}\sum_{t=0}^{H-1} \nabla_{\theta_i}\log\pi_{i}(a_{i,t}|s_{i,t},\theta_{i})\sum_{j\in\mathcal{N}^{\kappa}_{i}}\notag\\
&\sum_{l=H-t}^{\infty}\gamma^{t+l}r_{j,t+l}\Big]+\mathbb{E}_{\bm{s}_t,\bm{a}_t\sim\bm{\pi}_{\bm{\theta}}}\Big[\frac{1}{N}\sum_{t=H}^{\infty}\notag\\
&\nabla_{\theta_i}\log\pi_{i}(a_{i,t}|s_{i,t},\theta_{i})\sum_{j\in\mathcal{N}^{\kappa}_{i}}\sum_{l=0}^{\infty}\gamma^{t+l}r_{j,t+l}\Big]\Big\|\notag\\
\leq&\sum_{t=0}^{H-1}\frac{BR\gamma^{H}}{1-\gamma} +\sum_{t=H}^{\infty}\frac{BR\gamma^{t}}{1-\gamma}\label{theorem1bu1-1}\\
\leq&\frac{BR(H+1)\gamma^{H}}{(1-\gamma)^{2}},\label{theorem1bu1-2}
\end{align}
where (\ref{theorem1bu1-1}) can be achieved by Assumption~\ref{theassumptionofreward} and Assumption~\ref{theassumptionofpolicy}.
\par
For $\mathrm{(ii)}$-term in (\ref{theorem1bu1}),
we introduce several definitions to facilitate the analysis.
By using $\xi^{\bm{\pi_{\theta}}}_{\bm{\rho}}(\bm{s},\bm{a})$ in {\color{blue}(\ref{thestationarydistributionofsa})}, we further design a class of  truncated $Q$-function as
\begin{align}\label{thetruncatedQfunctions}
&Q^{\bm{\pi_{\theta}}}_{tru,i}(s_{\mathcal{N}^{\kappa}_{i}},a_{\mathcal{N}^{\kappa}_{i}})\notag\\
=&\sum_{s_{\mathcal{N}^{\kappa}_{-i}},a_{\mathcal{N}^{\kappa}_{-i}}}\xi^{\bm{\pi_{\theta}}}_{\bm{\rho}}(s_{\mathcal{N}^{\kappa}_{-i}},a_{\mathcal{N}^{\kappa}_{-i}}|s_{\mathcal{N}^{\kappa}_{i}},a_{\mathcal{N}^{\kappa}_{i}})\notag\\
&{\color{blue}\times} Q^{\bm{\pi_{\theta}}}_{i}(s_{\mathcal{N}^{\kappa}_{i}},s_{\mathcal{N}^{\kappa}_{-i}},a_{\mathcal{N}^{\kappa}_{i}},a_{\mathcal{N}^{\kappa}_{-i}}),
\end{align}
where  $\xi^{\bm{\pi_{\theta}}}_{\bm{\rho}}(s_{\mathcal{N}^{\kappa}_{-i}},a_{\mathcal{N}^{\kappa}_{-i}}|s_{\mathcal{N}^{\kappa}_{i}},a_{\mathcal{N}^{\kappa}_{i}})$ is the weight coefficient and satisfies
\begin{align}\label{thestationarydistributionof-s-a}
&\xi^{\bm{\pi_{\theta}}}_{\bm{\rho}}(s_{\mathcal{N}^{\kappa}_{-i}},a_{\mathcal{N}^{\kappa}_{-i}}|s_{\mathcal{N}^{\kappa}_{i}},a_{\mathcal{N}^{\kappa}_{i}})\notag\\
=&\frac{\xi^{\bm{\pi_{\theta}}}_{\bm{\rho}}(s_{\mathcal{N}^{\kappa}_{i}},s_{\mathcal{N}^{\kappa}_{-i}},a_{\mathcal{N}^{\kappa}_{i}},a_{\mathcal{N}^{\kappa}_{-i}})}{\sum_{s'_{\mathcal{N}^{\kappa}_{-i}},a'_{\mathcal{N}^{\kappa}_{-i}}}\xi^{\bm{\pi_{\theta}}}_{\bm{\rho}}(s_{\mathcal{N}^{\kappa}_{i}},s'_{\mathcal{N}^{\kappa}_{-i}},a_{\mathcal{N}^{\kappa}_{i}},a'_{\mathcal{N}^{\kappa}_{-i}})}.
\end{align}
{\color{blue}By} the definition of $\xi^{\bm{\pi_{\theta}}}_{\bm{\rho}}(s_{\mathcal{N}^{\kappa}_{-i}},a_{\mathcal{N}^{\kappa}_{-i}}|s_{\mathcal{N}^{\kappa}_{i}},a_{\mathcal{N}^{\kappa}_{i}})$ in (\ref{thestationarydistributionof-s-a}) and Assumption~\ref{theassumptionofdistributionofs}, it is obvious that the weight coefficients in (\ref{thetruncatedQfunctions}) are non-negative and satisfy
\begin{align}\label{thesumofweight}
\sum_{s'_{\mathcal{N}^{\kappa}_{-i}},a'_{\mathcal{N}^{\kappa}_{-i}}}\!\!\!\!\xi^{\bm{\pi_{\theta}}}_{\bm{\rho}}(s'_{\mathcal{N}^{\kappa}_{-i}},a'_{\mathcal{N}^{\kappa}_{-i}}|s_{\mathcal{N}^{\kappa}_{i}},a_{\mathcal{N}^{\kappa}_{i}})=1.
\end{align}
Inspired by~\citep{QuCLDC2020}, we define {\color{blue}a} truncated policy gradient of agent $i$ as
\begin{align}
g^{\bm{\pi_{\theta}}}_{tru,i}=&\frac{1}{1-\gamma}\mathbb{E}_{\bm{s}\sim d^{\bm{\pi_{\theta}}}_{\bm{\rho}},\bm{a}\sim\bm{\pi_{\theta}}}\Big[\frac{1}{N}\sum_{j\in\mathcal{N}^{\kappa}_{i}}Q^{\bm{\pi_{\theta}}}_{tru,j}(s_{\mathcal{N}^{\kappa}_{j}},a_{\mathcal{N}^{\kappa}_{j}})\notag\\
&{\color{blue}\times}\nabla_{\theta_{i}}\log\pi_{i}(a_{i}|s_{i},\theta_{i})\Big],\label{thetruncatedpolicygradient}
\end{align}
and have that
\begin{align}\label{thetruncatedpolicygradienterror}
\|g^{\bm{\pi_{\theta}}}_{tru,i}-\nabla_{\theta_{i}}J(\bm{\theta})\|\leq\frac{2BR}{(1-\gamma)^{2}}\gamma^{\kappa+1}.
\end{align}
The result in (\ref{thetruncatedpolicygradienterror}) characterizes the approximation error between the truncated policy gradient {\color{blue}$g^{\bm{\pi_{\theta}}}_{tru,i}$} and the exact policy gradient $\nabla_{\theta_{i}}J(\bm{\theta})$, which can be directly obtained from Lemma~3 and Lemma~4 in~\citep{QuCLDC2020}.
Note that unlike~\citep{QuCLDC2020}, where the reward is assumed to lie in $[0, R]$, {\color{blue}the reward range in Assumption~\ref{theassumptionofreward} of this work} is $[-R, R]$.
Consequently, the upper bound of (\ref{thetruncatedpolicygradienterror}) contains an additional factor of $2$.
\par
For any agent $i\in\mathcal{N}$, by the definition of $g^{\bm{\pi_{\theta}}}_{tru,i}$ in (\ref{thetruncatedpolicygradient}), we have
\begin{align}
g^{\bm{\pi}_{\bm{\theta}}}_{tru,i}
=&\frac{1}{1-\gamma}\mathbb{E}_{\bm{s}\sim d^{\bm{\pi}_{\bm{\theta}}}_{\bm{\rho}},\bm{a}\sim\bm{\pi}_{\bm{\theta}}}\Big[\frac{1}{N}\sum_{j\in\mathcal{N}^{\kappa}_{i}}Q^{\bm{\pi}_{\bm{\theta}}}_{tru,j}(s_{\mathcal{N}^{\kappa}_{j}},a_{\mathcal{N}^{\kappa}_{j}})\notag\\
&{\color{blue}\times}\nabla_{\theta_{i}}\log\pi_{i}(a_{i}|s_{i},\theta_{i})\Big]\notag\\
=&\frac{1}{1-\gamma}\mathbb{E}_{\bm{s}\sim d^{\bm{\pi}_{\bm{\theta}}}_{\bm{\rho}},\bm{a}\sim\bm{\pi}_{\bm{\theta}}}\Big[\frac{1}{N}\sum_{j\in\mathcal{N}^{\kappa}_{i}}\sum_{\tilde{s}_{\mathcal{N}^{\kappa}_{-j}},\tilde{a}_{\mathcal{N}^{\kappa}_{-j}}}\notag\\
&\xi^{\bm{\pi}_{\bm{\theta}}}_{\bm{\rho}}(\tilde{s}_{\mathcal{N}^{\kappa}_{-j}},\tilde{a}_{\mathcal{N}^{\kappa}_{-j}}|s_{\mathcal{N}^{\kappa}_{j}},a_{\mathcal{N}^{\kappa}_{j}})\notag\\
&{\color{blue}\times}Q^{\bm{\pi}_{\bm{\theta}}}_{j}(s_{\mathcal{N}^{\kappa}_{j}},\tilde{s}_{\mathcal{N}^{\kappa}_{-j}},a_{\mathcal{N}^{\kappa}_{j}},\tilde{a}_{\mathcal{N}^{\kappa}_{-j}})\notag\\
&{\color{blue}\times}\nabla_{\theta_{i}}\log\pi_{i}(a_{i}|s_{i},\theta_{i,t})\Big]\notag\\
=&\frac{1}{1-\gamma}\mathbb{E}_{\bm{s}\sim d^{\bm{\pi}_{\bm{\theta}}}_{\bm{\rho}},\bm{a}\sim\bm{\pi}_{\bm{\theta}}}\Big[\frac{1}{N}\notag\\
&{\color{blue}\times}\sum_{j\in\mathcal{N}^{\kappa}_{i}}Q^{\bm{\pi}_{\bm{\theta}}}_{j}(s_{\mathcal{N}^{\kappa}_{j}},s_{\mathcal{N}^{\kappa}_{-j}},a_{\mathcal{N}^{\kappa}_{j}},a_{\mathcal{N}^{\kappa}_{-j}})\notag\\
&{\color{blue}\times}\nabla_{\theta_{i}}\log\pi_{i}(a_{i}|s_{i},\theta_{i})\Big]\label{theequationoftherelationshipbetweentwofunctionskey}\\
=&\nabla_{\theta_{i}}\widetilde{J}_{i}(\bm{\theta}),\label{theequationoftherelationshipbetweentwofunctions2}
\end{align}
where the equality (\ref{theequationoftherelationshipbetweentwofunctionskey}) is obtained from the definition of $\xi^{\bm{\pi_{\theta}}}_{\bm{\rho}}(s_{\mathcal{N}^{\kappa}_{-i}},a_{\mathcal{N}^{\kappa}_{-i}}|s_{\mathcal{N}^{\kappa}_{i}},a_{\mathcal{N}^{\kappa}_{i}})$ in (\ref{thestationarydistributionof-s-a})
and the equality (\ref{theequationoftherelationshipbetweentwofunctions2}) uses another expression of $\nabla_{\theta_{i}}\widetilde{J}_{i}(\bm{\theta})$, similar to the form of $\nabla_{\theta_{i}}J(\bm{\theta})$ in (\ref{thepolicygradienttheorem}).
By combining (\ref{thetruncatedpolicygradienterror}) and (\ref{theequationoftherelationshipbetweentwofunctions2}), we obtain
\begin{align}\label{midelresult}
\|\nabla_{\theta_{i}}\widetilde{J}_{i}(\bm{\theta})-\nabla_{\theta_{i}}J(\bm{\theta})\|\leq\frac{2BR}{(1-\gamma)^{2}}\gamma^{\kappa+1}
\end{align}
for all agent $i\in\mathcal{N}$.
\par
Substituting (\ref{theorem1bu1-2}) and (\ref{midelresult}) into (\ref{theorem1bu1}), we can complete the proof.
\end{proof}

\subsection{Proof of Lemma~\ref{thelemmaofperturbedobjectivefunction}}\label{ProofofLemmathelemmaofperturbedobjectivefunction}
\begin{proof}
(i) By the definition of $\widehat{J}^{\mu}_{i}(\bm{\theta})$ in (\ref{perturbedvaluefunction}), for any $\bm{\theta},\bm{\theta}'\in\mathbb{R}^{d_{\mathrm{tot}}}$, we can have
\begin{align}
&\|\nabla_{\bm{\theta}}\widehat{J}^{\mu}_{i}(\bm{\theta})-\nabla_{\bm{\theta}}\widehat{J}^{\mu}_{i}(\bm{\theta}')\|\notag\\
=&\big\|\nabla_{\bm{\theta}}\mathbb{E}_{\bm{v}}[\widehat{J}_{i}((\bm{\theta}+\mu\bm{v})]-\nabla_{\bm{\theta}}\mathbb{E}_{\bm{v}}[\widehat{J}_{i}(\bm{\theta}'+\mu\bm{v})]\big\|\notag\\
=&\big\|\mathbb{E}_{\bm{v}}\big[\nabla_{\bm{\theta}}\widehat{J}_{i}(\bm{\theta}+\mu\bm{v})-\nabla_{\bm{\theta}}\widehat{J}_{i}(\bm{\theta}'+\mu\bm{v})\big]\big\|\notag\\
\leq&\mathbb{E}_{\bm{v}}\big[\big\|\nabla_{\bm{\theta}}\widehat{J}_{i}(\bm{\theta}+\mu\bm{v})-\nabla_{\bm{\theta}}\widehat{J}_{i}(\bm{\theta}'+\mu\bm{v})\big\|\big]\label{thegradientofperturbedobjectivefunction1-1}\\
\leq&L\|\bm{\theta}-\bm{\theta}'\|,\label{thegradientofperturbedobjectivefunction1-2}
\end{align}
where (\ref{thegradientofperturbedobjectivefunction1-1}) uses the triangle inequality of norm and (\ref{thegradientofperturbedobjectivefunction1-2}) can be obtained by Assumption~\ref{theassumptionofgradientofobjectivefunction}.
\par
Recalling the expression for $\widehat{J}^{\mu}_{i}(\bm{\theta})$ in (\ref{perturbedvaluefunction}), we can have
\begin{align}
\nabla_{\bm{\theta}}\widehat{J}^{\mu}_{i}(\bm{\theta})
=&\nabla_{\bm{\theta}}\mathbb{E}_{\bm{v}}[\widehat{J}_{i}(\bm{\theta}+\mu\bm{v})]\notag\\
=&\mathbb{E}_{\bm{v}}[\nabla_{\bm{\theta}}\widehat{J}_{i}(\bm{\theta}+\mu\bm{v})]\notag\\
=&\mathbb{E}_{\bm{v}}\Big[\frac{1}{\mu}\widehat{J}_{i}(\bm{\theta}+\mu\bm{v})\bm{v}\Big]\label{thegradientofperturbedobjectivefunction2-1}\\
=&\mathbb{E}_{\bm{v}}\Big[\frac{1}{\mu}\big(\widehat{J}_{i}(\bm{\theta}+\mu\bm{v})-\widehat{J}_{i}(\bm{\theta})\big)\bm{v}\Big]\label{thegradientofperturbedobjectivefunction2-2},
\end{align}
where (\ref{thegradientofperturbedobjectivefunction2-1}) uses the Stein's identity and (\ref{thegradientofperturbedobjectivefunction2-2}) comes follows the fact that $\mathbb{E}_{\bm{v}}[\widehat{J}_{i}(\bm{\theta})\bm{v}]=\mathbf{0}_{d_{\mathrm{tot}}}$.
\par
(ii) By applying the $L$-smooth of $\widehat{J}_{i}(\bm{\theta})$, we can derive
\begin{align}\label{thegradientofperturbedobjectivefunction(ii)1-1}
\big|\widehat{J}_{i}(\bm{\theta}+\mu\bm{v})-\widehat{J}_{i}(\bm{\theta})-\mu\nabla_{\bm{\theta}}\widehat{J}_{i}(\bm{\theta})^{\top}\bm{v}\big|\leq\frac{L\mu^{2}}{2}\bm{v}^{\top}\bm{v}.
\end{align}
{\color{blue}By using} the definition on $\widehat{J}^{\mu}_{i}(\bm{\theta})$ in (\ref{perturbedvaluefunction}),
we further have
\begin{align}
&|\widehat{J}^{\mu}_{i}(\bm{\theta})-\widehat{J}_{i}(\bm{\theta})|\notag\\
=&\big|\mathbb{E}_{\bm{v}}[\widehat{J}_{i}(\bm{\theta}+\mu\bm{v})]-\widehat{J}_{i}(\bm{\theta})\big|\notag\\
=&\big|\mathbb{E}_{\bm{v}}[\widehat{J}_{i}(\bm{\theta}+\mu\bm{v})]-\mathbb{E}_{\bm{v}}[\widehat{J}_{i}(\bm{\theta})]-\mathbb{E}_{\bm{v}}[\mu\nabla_{\bm{\theta}}\widehat{J}_{i}(\bm{\theta})^{\top}\bm{v}]\big|\label{thegradientofperturbedobjectivefunction(ii)3-1}\\
\leq&\mathbb{E}_{\bm{v}}\big[\big|\widehat{J}_{i}(\bm{\theta}+\mu\bm{v})-\widehat{J}_{i}(\bm{\theta})-\mu\nabla_{\bm{\theta}}\widehat{J}_{i}(\bm{\theta})^{\top}\bm{v}\big|\big]\notag\\
\leq&\frac{L\mu^{2}}{2}\mathbb{E}_{\bm{v}}[\bm{v}^{\top}\bm{v}]\label{thegradientofperturbedobjectivefunction(ii)3-3}\\
=&\frac{L\mu^{2}d_{\mathrm{tot}}}{2},\label{thegradientofperturbedobjectivefunction(ii)3-4}
\end{align}
where (\ref{thegradientofperturbedobjectivefunction(ii)3-1}) uses the fact that $\mathbb{E}_{\bm{v}}[\nabla_{\bm{\theta}}\widehat{J}_{i}(\bm{\theta})^{\top}\bm{v}]=0$, (\ref{thegradientofperturbedobjectivefunction(ii)3-3}) can be obtained by (\ref{thegradientofperturbedobjectivefunction(ii)1-1}), and (\ref{thegradientofperturbedobjectivefunction(ii)3-4}) can be achieved by the fact that $\mathbb{E}_{\bm{v}}[\bm{v}^{\top}\bm{v}]=d_{\mathrm{tot}}$.
\par
(iii) Recalling the definition of $\widehat{J}^{\mu}_{i}(\bm{\theta})$ in (\ref{perturbedvaluefunction}) again, we have
\begin{align}
&\|\nabla_{\bm{\theta}}\widehat{J}^{\mu}_{i}(\bm{\theta})-\nabla_{\bm{\theta}}\widehat{J}_{i}(\bm{\theta})\|\notag\\
=&\big\|\mathbb{E}_{\bm{v}}[\nabla_{\bm{\theta}}\widehat{J}_{i}(\bm{\theta}+\mu\bm{v})-\nabla_{\bm{\theta}}\widehat{J}_{i}(\bm{\theta})]\big\|\notag\\
\leq&\mathbb{E}_{\bm{v}}\big[\big\|\nabla_{\bm{\theta}}\widehat{J}_{i}(\bm{\theta}+\mu\bm{v})-\nabla_{\bm{\theta}}\widehat{J}_{i}(\bm{\theta})\big\|\big]\notag\\
\leq&L\mu\mathbb{E}_{\bm{v}}[\|\bm{v}\|]\label{thegradientofperturbedobjectivefunction(iii)1-1}\\
\leq&L\mu\sqrt{d_{\mathrm{tot}}},\label{thegradientofperturbedobjectivefunction(iii)1-2}
\end{align}
where (\ref{thegradientofperturbedobjectivefunction(iii)1-1}) uses the $L$-smoothness of $\widehat{J}_{i}(\bm{\theta})$ and (\ref{thegradientofperturbedobjectivefunction(iii)1-2}) can be achieved by the fact that $\mathbb{E}_{\bm{v}}[\|\bm{v}\|]\leq\sqrt{d_{\mathrm{tot}}}$.
\par
(iv) For $\Big\|\frac{1}{\mu}\big(\widehat{J}_{i}(\bm{\theta}+\mu\bm{v})-\widehat{J}_{i}(\bm{\theta})\big)\bm{v}\Big\|^{2}$, we have
\begin{align}
&\Big\|\frac{1}{\mu}\big(\widehat{J}_{i}(\bm{\theta}+\mu\bm{v})-\widehat{J}_{i}(\bm{\theta})\big)\bm{v}\Big\|^{2}\notag\\
=&\frac{1}{\mu^{2}}\big(\widehat{J}_{i}(\bm{\theta}+\mu\bm{v})-\widehat{J}_{i}(\bm{\theta})\big)^{2}\|\bm{v}\|^{2}\notag\\
\leq&\frac{1}{\mu^{2}}\big(\mu\|\nabla_{\bm{\theta}}\widehat{J}_{i}(\bm{\theta})\|\|\bm{v}\|+\frac{\mu^{2}L}{2}\|\bm{v}\|^{2}\big)^{2}\|\bm{v}\|^{2}\label{thegradientofperturbedobjectivefunction(iv)1-1}\\
\leq&\frac{1}{\mu^{2}}\big(2\mu^{2}\|\nabla_{\bm{\theta}}\widehat{J}_{i}(\bm{\theta})\|^{2}\|\bm{v}\|^{2}+\frac{\mu^{4}L^{2}}{2}\|\bm{v}\|^{4}\big)\|\bm{v}\|^{2}\label{thegradientofperturbedobjectivefunction(iv)1-2}\\
=&2\|\nabla_{\bm{\theta}}\widehat{J}_{i}(\bm{\theta})\|^{2}\|\bm{v}\|^{4}+\frac{\mu^{2}L^{2}}{2}\|\bm{v}\|^{6},\label{thegradientofperturbedobjectivefunction(iv)1-3}
\end{align}
where (\ref{thegradientofperturbedobjectivefunction(iv)1-1}) uses the property of $L$-smooth of $\widehat{J}_{i}(\bm{\theta})$ and (\ref{thegradientofperturbedobjectivefunction(iv)1-2}) comes from the fact that $(x+y)^{2}\leq2x^{2}+2y^{2}$ for all $x,y\in\mathbb{R}$.
Taking the expectation of both sides of (\ref{thegradientofperturbedobjectivefunction(iv)1-3}), we have
\begin{align}
&\mathbb{E}_{\bm{v}}\Big[\Big\|\frac{1}{\mu}\big(\widehat{J}_{i}(\bm{\theta}+\mu\bm{v})-\widehat{J}_{i}(\bm{\theta})\big)\bm{v}\Big\|^{2}\Big]\notag\\
\leq&\mathbb{E}_{\bm{v}}\Big[2\|\nabla_{\bm{\theta}}\widehat{J}_{i}(\bm{\theta})\|^{2}\|\bm{v}\|^{4}+\frac{\mu^{2}L^{2}}{2}\|\bm{v}\|^{6}\Big]\notag\\
=&2d_{\mathrm{tot}}(d_{\mathrm{tot}}\!+\!2)\|\nabla_{\bm{\theta}}\widehat{J}_{i}(\bm{\theta})\|^{2}\!+\!\frac{\mu^{2}L^{2}d_{\mathrm{tot}}(d_{\mathrm{tot}}\!+\!2)(d_{\mathrm{tot}}\!+\!4)}{2},\label{thegradientofperturbedobjectivefunction(iv)2-1}
\end{align}
where the last inequality can be achieved by the fact that $\mathbb{E}_{\bm{v}}[\|\bm{v}\|^{4}]=d_{\mathrm{tot}}(d_{\mathrm{tot}}+2)$ and $\mathbb{E}_{\bm{v}}[\|\bm{v}\|^{6}]=d_{\mathrm{tot}}(d_{\mathrm{tot}}+2)(d_{\mathrm{tot}}+4)$.
\end{proof}

\subsection{Proof of Theorem~\ref{thetheoremofpolicygradient}}\label{ProofofTheoremthetheoremofpolicygradient}
\begin{proof}
(i) Recalling the definition of $\mathrm{Bias}_{t}$ in (\ref{theequationofthefirst-ordergradientbias}), we directly have
\begin{align}\label{equationoftheorem1-2-1}
\mathrm{Bias}_{t}\leq\big\|\nabla_{\bm{\theta}}J(\bm{\theta}_{t})\big\|\big\|\mathbb{E}\big[\nabla_{\bm{\theta}}J(\bm{\theta}_{t})-\bm{\hat{g}}_{t}\big|\mathcal{F}_{t}\big]\big\|.
\end{align}
Since $\nabla_{\bm{\theta}}J(\bm{\theta}_{t})=\big(\nabla_{\theta_{1}}J(\bm{\theta}_{t})^{\top},\cdots,\nabla_{\theta_{N}}J(\bm{\theta}_{t})^{\top}\big)^{\top}$ and $\bm{\hat{g}}_{t}=\big(\hat{g}^{\top}_{1,t},\cdots,\hat{g}^{\top}_{N,t}\big)^{\top}$,
for any agent $i\in\mathcal{N}$,
we have
\begin{align}
&\big\|\mathbb{E}\big[\nabla_{\theta_{i}}J(\bm{\theta}_{t})-\hat{g}_{i,t}\big|\mathcal{F}_{t}\big]\big\|\notag\\
=&\big\|\mathbb{E}\big[\nabla_{\theta_{i}}J(\bm{\theta}_{t})-\nabla_{\theta_{i}}\widehat{J}_{i}(\bm{\theta}_{t})+\nabla_{\theta_{i}}\widehat{J}_{i}(\bm{\theta}_{t})-\nabla_{\theta_{i}}\widehat{J}^{\mu}_{i}(\bm{\theta}_{t})\notag\\
&+\nabla_{\theta_{i}}\widehat{J}^{\mu}_{i}(\bm{\theta}_{t})-\hat{g}_{i,t}\big|\mathcal{F}_{t}\big]\big\|\notag\\
\leq&\big\|\nabla_{\theta_{i}}J(\bm{\theta}_{t})-\nabla_{\theta_{i}}\widehat{J}_{i}(\bm{\theta}_{t})\big\|+\big\|\nabla_{\theta_{i}}\widehat{J}_{i}(\bm{\theta}_{t})-\nabla_{\theta_{i}}\widehat{J}^{\mu}_{i}(\bm{\theta}_{t})\big\|\notag\\
&+\big\|\mathbb{E}\big[\nabla_{\theta_{i}}\widehat{J}^{\mu}_{i}(\bm{\theta}_{t})-\hat{g}_{i,t}\big|\mathcal{F}_{t}\big]\big\|\notag\\
\leq&\frac{BR\big((H+1)\gamma^{H}+2\gamma^{\kappa+1}\big)}{(1-\gamma)^{2}}+L\mu\sqrt{d_{\mathrm{tot}}}\notag\\
&+\big\|\mathbb{E}\big[\nabla_{\theta_{i}}\widehat{J}^{\mu}_{i}(\bm{\theta}_{t})-\hat{g}_{i,t}\big|\mathcal{F}_{t}\big]\big\|,\label{theequationofcorollary1}
\end{align}
where the last inequality follows from
Theorem~\ref{lemmaoftruncatederror} and (iii) of Lemma~\ref{thelemmaofperturbedobjectivefunction}.
\par
For the last term $\big\|\mathbb{E}\big[\nabla_{\theta_{i}}\widehat{J}^{\mu}_{i}(\bm{\theta}_{t})-\hat{g}_{i,t}\big|\mathcal{F}_{t}\big]\big\|$ in the right side of (\ref{theequationofcorollary1}), we have
\begin{align}
&\big\|\mathbb{E}\big[\nabla_{\theta_{i}}\widehat{J}^{\mu}_{i}(\bm{\theta}_{t})-\hat{g}_{i,t}\big|\mathcal{F}_{t}\big]\big\|\notag\\
=&\Big\|\mathbb{E}\Big[\frac{1}{\mu}\big(\widehat{J}_{i}(\bm{\theta}_{t}+\mu\bm{v}_{t})-\widehat{J}_{i}(\bm{\theta}_{t})\big)v_{i,t}\Big|\mathcal{F}_{t}\Big]\notag\\
&-\mathbb{E}\Big[\frac{1}{K\mu}\sum_{k=1}^{K}\sigma^{-1}(\hat{p}^{t,k}_{i})v_{i,t}\Big|\mathcal{F}_{t}\Big]\Big\|,\label{theequationofcorollary2}
\end{align}
which use the definition of $\hat{g}_{i,t}$ in (\ref{policygradientestimation}) and (i) of Lemma~\ref{thelemmaofperturbedobjectivefunction}.
Since each trajectory pair $(\tau^{t,k}_{\mathcal{N}^{\kappa}_{i},0},\tau^{t,k}_{\mathcal{N}^{\kappa}_{i},1})$ is generated independently from other {\color{blue}trajectory pairs}, we have
\begin{align}
&\mathbb{E}\Big[\frac{1}{K\mu}\sum_{k=1}^{K}\sigma^{-1}(\hat{p}^{t,k}_{i})v_{i,t}\Big|\mathcal{F}_{t}\Big]\notag\\
=&\mathbb{E}\Bigg[\mathbb{E}\Big[\frac{1}{K\mu}\sum_{k=1}^{K}\sigma^{-1}(\hat{p}^{t,k}_{i})\Big|v_{i,t}\Big]v_{i,t}\Bigg|\mathcal{F}_{t}\Bigg]\label{theequationofcorollary3-1}\\
=&\frac{1}{\mu}\mathbb{E}\Big[\mathbb{E}\big[\sigma^{-1}(\hat{p}^{t,k}_{i})\big|v_{i,t}\big]v_{i,t}\Big|\mathcal{F}_{t}\Big]\label{theequationofcorollary3-2}\\
=&\frac{1}{\mu}\mathbb{E}\Big[\mathbb{E}\big[\hat{r}_{i}(\tau^{t,k}_{\mathcal{N}^{\kappa}_{i},1})-\hat{r}_{i}(\tau^{t,k}_{\mathcal{N}^{\kappa}_{i},0})\big|v_{i,t}\big]v_{i,t}\Big|\mathcal{F}_{t}\Big]\notag\\
&+\frac{1}{\mu}\mathbb{E}\Big[\mathbb{E}\big[\big(\sigma^{-1}(\hat{p}^{t,k}_{i})-(\hat{r}_{i}(\tau^{t,k}_{\mathcal{N}^{\kappa}_{i},1})-\hat{r}_{i}(\tau^{t,k}_{\mathcal{N}^{\kappa}_{i},0}))\big)\big|v_{i,t}\big]\notag\\
&\times v_{i,t}\Big|\mathcal{F}_{t}\Big],\label{theequationofcorollary3-3}
\end{align}
where (\ref{theequationofcorollary3-1}) uses the law of total expectation and (\ref{theequationofcorollary3-2}) uses the fact that the {\color{blue}trajectory pairs are generated by the same policy pair and mutually independent}.
Substituting (\ref{theequationofcorollary3-3}) into (\ref{theequationofcorollary2}), we have
\begin{align}
&\big\|\mathbb{E}\big[\nabla_{\theta_{i}}\widehat{J}^{\mu}_{i}(\bm{\theta}_{t})-\hat{g}_{i,t}\big|\mathcal{F}_{t}\big]\big\|\notag\\
\leq&\Big\|\frac{1}{\mu}\mathbb{E}\Big[\mathbb{E}\big[\big(\widehat{J}_{i}(\bm{\theta}_{t}+\mu\bm{v}_{t})-\widehat{J}_{i}(\bm{\theta}_{t})-\hat{r}_{i}(\tau^{t,k}_{\mathcal{N}^{\kappa}_{i},1})\notag\\
&+\hat{r}_{i}(\tau^{t,k}_{\mathcal{N}^{\kappa}_{i},0})\big)\big|v_{i,t}\big]v_{i,t}\Big|\mathcal{F}_{t}\Big]-\frac{1}{\mu}\mathbb{E}\Big[\mathbb{E}\big[\big(\sigma^{-1}(\hat{p}^{t,k}_{i})\notag\\
&-(\hat{r}_{i}(\tau^{t,k}_{\mathcal{N}^{\kappa}_{i},1})-\hat{r}_{i}(\tau^{t,k}_{\mathcal{N}^{\kappa}_{i},0}))\big)\big|v_{i,t}\big]v_{i,t}\Big|\mathcal{F}_{t}\Big]\Big\|\label{theequationofcorollary4-1}\\
\leq&\frac{\sqrt{d_{\mathrm{tot}}}}{\mu}\Big(L_{\sigma}\sqrt{\frac{2\log{M}}{M}}+\frac{4R(1-\gamma^{H})}{(1-\gamma)M^{2}}\Big)\notag\\
\leq&\frac{2L_{\sigma}\sqrt{d_{\mathrm{tot}}}}{\mu}\sqrt{\frac{2\log{M}}{M}},\label{theequationofcorollary4-3}
\end{align}
where the second inequality follows from the facts that  $\mathbb{E}\big[\widehat{J}_{i}(\bm{\theta}_{t}+\mu\bm{v}_{t})-\hat{r}_{i}(\tau^{t,k}_{\mathcal{N}^{\kappa}_{i},1})\big]=\mathbb{E}\big[\widehat{J}_{i}(\bm{\theta}_{t})-\hat{r}_{i}(\tau^{t,k}_{\mathcal{N}^{\kappa}_{i},0})\big]=0$, the norm bound $\mathbb{E}[\|v_{i,t}\|]\leq\sqrt{d_{\mathrm{tot}}}$, and the estimation error result in (\ref{thelemmaofpreferenceestimation(i)}), and the last inequality follows from the selected $M\geq\max\{e, 2(\frac{R^{2}(1-\gamma^{H})^{2}}{(1-\gamma)^{2}L^{2}_{\sigma}})^{\frac{1}{3}}\}$.
\par
By substituting (\ref{theequationofcorollary4-3}) into (\ref{theequationofcorollary1}) and combining it with (\ref{equationoftheorem1-2-1}), we can finish the proof of part (i) of Theorem~\ref{thetheoremofpolicygradient}.
\par
(ii) By the definition of $\hat{g}_{i,t}$ in (\ref{policygradientestimation}), we define
\begin{align}\label{thetheorem2-1}
\mathrm{Var}_{i,t}=\mathbb{E}[\|\hat{g}_{i,t}\|^{2}|\mathcal{F}_{t}]
\end{align}
and have $\mathrm{Var}_{t}=\sum_{i=1}^{N}\mathrm{Var}_{i,t}$.
Define $\mathrm{Var}^{(1)}_{i,t}$, $\mathrm{Var}^{(2)}_{i,t}$, and $\mathrm{Var}^{(3)}_{i,t}$ as follows:
\begin{align}\label{thetheorem2-2}
\left\{
\begin{array}{ll}
\!\mathrm{Var}^{(1)}_{i,t}\!=\!\mathbb{E}\Big[\Big\|\frac{1}{K\mu}\sum_{k=1}^{K}\Big(\sigma^{-1}(\hat{p}^{t,k}_{i})-\big(\hat{r}_{i}(\tau^{t,k}_{\mathcal{N}^{\kappa}_{i},1})\\
\;\;\;\;\;\;\;\;\;\;\;\;-\hat{r}_{i}(\tau^{t,k}_{\mathcal{N}^{\kappa}_{i},0})\big)\Big)v_{i,t}\Big\|^{2}\Big|\mathcal{F}_{t}\Big]\\
\!\mathrm{Var}^{(2)}_{i,t}\!=\!\mathbb{E}\Big[\Big\|\frac{1}{\mu}\Big(\frac{1}{K}\sum_{k=1}^{K}\big(\hat{r}_{i}(\tau^{t,k}_{\mathcal{N}^{\kappa}_{i},1})-\hat{r}_{i}(\tau^{t,k}_{\mathcal{N}^{\kappa}_{i},0})\big)\\
\;\;\;\;\;\;\;\;\;\;\;\;-\big(\widehat{J}_{i}(\bm{\theta}_{t}+\mu\bm{v}_{t})-\widehat{J}_{i}(\bm{\theta}_{t})\big)\Big)v_{i,t}\Big\|^{2}\Big|\mathcal{F}_{t}\Big]\\
\!\mathrm{Var}^{(3)}_{i,t}\!=\!\mathbb{E}\Big[\Big\|\frac{1}{\mu}\big(\widehat{J}_{i}(\bm{\theta}_{t}+\mu\bm{v}_{t})-\widehat{J}_{i}(\bm{\theta}_{t})\big)v_{i,t}\Big\|^{2}\Big|\mathcal{F}_{t}\Big].
\end{array}
\right.
\end{align}
For $\mathrm{Var}_{i,t}$ in (\ref{thetheorem2-1}), we have
\begin{align}\label{thetheorem2-3}
\mathrm{Var}_{i,t}=&\mathbb{E}\Big[\Big\|\frac{1}{K\mu}\sum_{k=1}^{K}\sigma^{-1}(\hat{p}^{t,k}_{i})v_{i,t}\Big\|^{2}\Big|\mathcal{F}_{t}\Big]\notag\\
\leq&3\big(\mathrm{Var}^{(1)}_{i,t}+\mathrm{Var}^{(2)}_{i,t}+\mathrm{Var}^{(3)}_{i,t}\big),
\end{align}
where the inequality can be obtained by the fact that $(a+b+c)^{2}\leq3(a^{2}+b^{2}+c^{2})$ for all $a,b,c\in\mathbb{R}$.
\par
For the first term, $\mathrm{Var}^{(1)}_{i,t}$, which arises from the estimation error in human preferences, we have
\begin{align}
\mathrm{Var}^{(1)}_{i,t}
=&\frac{1}{K^{2}\mu^{2}}\mathbb{E}\Big[\Big|\sum_{k=1}^{K}\Big(\sigma^{-1}(\hat{p}^{t,k}_{i})-\big(\hat{r}_{i}(\tau^{t,k}_{\mathcal{N}^{\kappa}_{i},1})\notag\\
&-\hat{r}_{i}(\tau^{t,k}_{\mathcal{N}^{\kappa}_{i},0})\big)\Big)\Big|^{2}\times\|v_{i,t}\|^{2}\Big|\mathcal{F}_{t}\Big]\notag\\
\leq&\frac{1}{K^{2}\mu^{2}}\mathbb{E}\Big[\Big|\sum_{k=1}^{K}\Big(\sigma^{-1}(\hat{p}^{t,k}_{i})-\big(\hat{r}_{i}(\tau^{t,k}_{\mathcal{N}^{\kappa}_{i},1})\notag\\
&-\hat{r}_{i}(\tau^{t,k}_{\mathcal{N}^{\kappa}_{i},0})\big)\Big)\Big|^{4}\Big|\mathcal{F}_{t}\Big]^{\frac{1}{2}}\times\mathbb{E}\Big[\|v_{i,t}\|^{4}\Big|\mathcal{F}_{t}\Big]^{\frac{1}{2}}\label{thetheorem3-1-1}\\
\leq&\frac{\sqrt{d_{i}(d_{i}+2)}}{K^{2}\mu^{2}}\mathbb{E}\Big[K^{3}\sum_{k=1}^{K}\Big|\sigma^{-1}(\hat{p}^{t,k}_{i})-\big(\hat{r}_{i}(\tau^{t,k}_{\mathcal{N}^{\kappa}_{i},1})\notag\\
&-\hat{r}_{i}(\tau^{t,k}_{\mathcal{N}^{\kappa}_{i},0})\big)\Big|^{4}\Big|\mathcal{F}_{t}\Big]^{\frac{1}{2}}\label{thetheorem3-1-2}\\
\leq&\frac{\sqrt{d_{i}(d_{i}+2)}}{\mu^{2}}\Big(\frac{4L^{4}_{\sigma}(\log{M})^{2}}{M^{2}}\notag\\
&+\frac{256R^{4}(1-\gamma^{H})^{4}}{(1-\gamma)^{4}M^{2}}\Big)^{\frac{1}{2}}\label{thetheorem3-1-3}\\
\leq&\frac{2\sqrt{2d_{i}(d_{i}+2)}L^{2}_{\sigma}\log{M}}{\mu^{2}M},\label{thetheorem3-1-4}
\end{align}
where (\ref{thetheorem3-1-1}) uses the Cauchy-Schwarz inequality, (\ref{thetheorem3-1-2}) uses the facts that $\mathbb{E}[\|v_{i,t}\|^{4}]=d_{i}(d_{i}+2)$ and $\big(\sum_{k=1}^{K}x_{k}\big)^{4}\leq K^{3}\sum_{k=1}^{K}x_{k}^{4}$ for all $x_{k}\in\mathbb{R}$,
(\ref{thetheorem3-1-3})
follows from (\ref{thelemmaofpreferenceestimation(iii)}), and the last inequality achieves by the choosing $M\geq\exp\big(\frac{8R^{2}(1-\gamma^{H})^{2}}{(1-\gamma)^{2}L^{2}_{\sigma}}\big)$.
\par
For the second term, $\mathrm{Var}^{(2)}_{i,t}$, which comes from using the empirical rewards of agent $i$'s $\kappa$-hop neighbors  to the truncated objective $\widehat{J}_{i}(\bm{\theta})$ in (\ref{thetruncatedobjectivefunction}), we have
\begin{align}
\mathrm{Var}^{(2)}_{i,t}
=&\frac{1}{\mu^{2}K^{2}}\mathbb{E}\Big[\Big|\sum_{k=1}^{K}\Big(\big(\hat{r}_{i}(\tau^{t,k}_{\mathcal{N}^{\kappa}_{i},1})-\hat{r}_{i}(\tau^{t,k}_{\mathcal{N}^{\kappa}_{i},0})\big)\notag\\
&-\big(\widehat{J}_{i}(\bm{\theta}_{t}+\mu\bm{v}_{t})-\widehat{J}_{i}(\bm{\theta}_{t})\big)\Big)\Big|^{2}\|v_{i,t}\|^{2}\Big|\mathcal{F}_{t}\Big]\notag\\
\leq&\frac{\sqrt{d_{i}(d_{i}+2)}}{\mu^{2}K^{2}}\mathbb{E}\Big[\Big|\sum_{k=1}^{K}\Big(\big(\hat{r}_{i}(\tau^{t,k}_{\mathcal{N}^{\kappa}_{i},1})-\hat{r}_{i}(\tau^{t,k}_{\mathcal{N}^{\kappa}_{i},0})\big)\notag\\
&-\big(\widehat{J}_{i}(\bm{\theta}_{t}+\mu\bm{v}_{t})-\widehat{J}_{i}(\bm{\theta}_{t})\big)\Big)\Big|^{4}\Big|\mathcal{F}_{t}\Big]^{\frac{1}{2}},\label{thetheorem3-2-1}
\end{align}
where the inequality follows from the Cauchy-Schwarz inequality and the fact that $\mathbb{E}[\|v_{i,t}\|^{4}]=d_{i}(d_{i}+2)$.
Let $\xi_{i}^{t,k}=\big(\hat{r}_{i}(\tau^{t,k}_{\mathcal{N}^{\kappa}_{i},1})-\hat{r}_{i}(\tau^{t,k}_{\mathcal{N}^{\kappa}_{i},0})\big)-\big(\widehat{J}_{i}(\bm{\theta}_{t}+\mu\bm{v}_{t})-\widehat{J}_{i}(\bm{\theta}_{t})\big)$, (\ref{thetheorem3-2-1}) can be further written as
\begin{align}
\mathrm{Var}^{(2)}_{i,t}\leq&\frac{\sqrt{d_{i}(d_{i}+2)}}{\mu^{2}K^{2}}\mathbb{E}\Big[\Big(\sum_{k=1}^{K}\xi_{i}^{t,k}\Big)^{4}\Big|\mathcal{F}_{t}\Big]^{\frac{1}{2}}.\label{thetheorem3-2-2}
\end{align}
Recalling the definitions of (\ref{therewardoftrajectory}) and (\ref{thetruncatedobjectivefunction}), we have that
$\mathbb{E}[\xi_{i}^{t,k}]=0$, $|\xi_{i}^{t,k}|\leq\frac{4R(1-\gamma^{H})}{1-\gamma}$, and
$\mathrm{var}(\xi_{i}^{t,k})\leq\frac{16R^{2}(1-\gamma^{H})^{2}}{(1-\gamma)^{2}}$.
By using Lemma~\ref{lemmaofFourthmomentRosenthal-type}, we can get
\begin{align}\label{thetheorem3-2-3}
\mathbb{E}\Big[\Big(\sum_{k=1}^{K}\xi_{i}^{t,k}\Big)^{4}\Big|\mathcal{F}_{t}\Big]^{\frac{1}{2}}\leq&\frac{16\sqrt{3}KR^{2}(1-\gamma^{H})^{2}}{(1-\gamma)^{2}}\notag\\
&+\frac{16\sqrt{K}R^{2}(1-\gamma^{H})^{2}}{(1-\gamma)^{2}}\notag\\
\leq&\frac{44KR^{2}}{(1-\gamma)^{2}},
\end{align}
where the last inequality follows from $K\geq1$ and $0<\gamma<1$.
Substituting (\ref{thetheorem3-2-3}) into (\ref{thetheorem3-2-2}), we can have
\begin{align}\label{thetheorem3-2-4}
\mathrm{Var}^{(2)}_{i,t}\leq&\frac{44\sqrt{d_{i}(d_{i}+2)}R^{2}}{\mu^{2}(1-\gamma)^{2}K}.
\end{align}
\par
For the last term, $\mathrm{Var}^{(3)}_{i,t}$, we use (iv) of Lemma~\ref{thelemmaofperturbedobjectivefunction} and can get
\begin{align}
\mathrm{Var}^{(3)}_{i,t}
\leq&2d_{\mathrm{tot}}(d_{\mathrm{tot}}+2)\|\nabla_{\bm{\theta}}\widehat{J}_{i}(\bm{\theta}_{t})-\nabla_{\theta_{i}}J(\bm{\theta}_{t})\notag\\
&+\nabla_{\theta_{i}}J(\bm{\theta}_{t})\|^{2}+\frac{\mu^{2}L^{2}d_{\mathrm{tot}}(d_{\mathrm{tot}}+2)(d_{\mathrm{tot}}+4)}{2}\notag\\
\leq&4d_{\mathrm{tot}}(d_{\mathrm{tot}}+2)\|\nabla_{\theta_{i}}J(\bm{\theta}_{t})\|^{2}\notag\\
&+\frac{4B^{2}R^{2}\big((H+1)\gamma^{H}+2\gamma^{\kappa+1}\big)^{2}d_{\mathrm{tot}}(d_{\mathrm{tot}}+2)}{(1-\gamma)^{4}}\notag\\
&+\frac{\mu^{2}L^{2}d_{\mathrm{tot}}(d_{\mathrm{tot}}+2)(d_{\mathrm{tot}}+4)}{2},\label{thetheorem3-3-1}
\end{align}
where the inequality can be achieved by Theorem~\ref{lemmaoftruncatederror}.
\par
By substituting (\ref{thetheorem3-1-4}), (\ref{thetheorem3-2-4}), and (\ref{thetheorem3-3-1}) into (\ref{thetheorem2-3}), and utilizing the relation $\mathrm{Var}_{t} = \sum_{i=1}^{N} \mathrm{Var}_{i,t}$, the proof of Part (ii) is thus completed.
\end{proof}

\subsection{Proof of Theorem~\ref{thetheoremofconvergence}}\label{ProofofTheoremthetheoremofconvergence}
\begin{proof}
By the property of $L$-Lipschitz continuous of $\nabla_{\bm{\theta}}J(\bm{\theta})$ in  Assumption~\ref{theassumptionofgradientofobjectivefunction}, we have
\begin{align}
&J(\bm{\theta}_{t})-J(\bm{\theta}_{t+1})\notag\\
\leq&\langle-\nabla_{\bm{\theta}}J(\bm{\theta}_{t}),\bm{\theta}_{t+1}-\bm{\theta}_{t}\rangle+\frac{L}{2}\|\bm{\theta}_{t+1}-\bm{\theta}_{t}\|^{2}\notag\\
=&-\alpha\langle\nabla_{\bm{\theta}}J(\bm{\theta}_{t}),\bm{\hat{g}}_{t}\rangle+\frac{L\alpha^{2}}{2}\|\bm{\hat{g}}_{t}\|^{2}\label{finalthetheorem2-1}\\
=&-\alpha\|\nabla_{\bm{\theta}}J(\bm{\theta}_{t})\|^{2}+\alpha\langle\nabla_{\bm{\theta}}J(\bm{\theta}_{t}),\nabla_{\bm{\theta}}J(\bm{\theta}_{t})-\bm{\hat{g}}_{t}\rangle\notag\\
&+\frac{L\alpha^{2}}{2}\|\bm{\hat{g}}_{t}\|^{2},\label{finalthetheorem2-2}
\end{align}
where (\ref{finalthetheorem2-1}) follows from (\ref{updateofpolicyparameter}).
We take a conditional expectation of both sides of (\ref{finalthetheorem2-2}) over the $\sigma$-algebra $\mathcal{F}_{t}$ of iteration $t$ and obtain
\begin{align}
&\mathbb{E}[J(\bm{\theta}_{t})-J(\bm{\theta}_{t+1})|\mathcal{F}_{t}]\notag\\
\leq&-\alpha\|\nabla_{\bm{\theta}}J(\bm{\theta}_{t})\|^{2}+\alpha\mathrm{Bias}_{t}+\frac{L\alpha^{2}}{2}\mathrm{Var}_{t}\notag\\
\leq&-\alpha\|\nabla_{\bm{\theta}}J(\bm{\theta}_{t})\|^{2}+\alpha\big\|\nabla_{\bm{\theta}}J(\bm{\theta}_{t})\big\|\big(\mathrm{Bias}^{\mathrm{(i)}}_{t}+\mathrm{Bias}^{\mathrm{(ii)}}_{t}\notag\\
&+\mathrm{Bias}^{\mathrm{(iii)}}_{t}\big)+\frac{L\alpha^{2}}{2}\Big(12Nd_{\mathrm{tot}}(d_{\mathrm{tot}}+2)\|\nabla_{\bm{\theta}}J(\bm{\theta}_{t})\|^{2}\notag\\
&+\mathrm{Var}^{\mathrm{(ii)}}_{t}+\mathrm{Var}^{\mathrm{(iii)}}_{t}+\mathrm{Var}^{\mathrm{(iv)}}_{t}+\mathrm{Var}^{\mathrm{(v)}}_{t}\Big),\label{finalthetheorem3-1}
\end{align}
where the last inequality can be obtained by Theorem~\ref{thetheoremofpolicygradient}.
Since $\alpha=\frac{1}{12LNd_{\mathrm{tot}}(d_{\mathrm{tot}}+2)}$, we can further simplify (\ref{finalthetheorem3-1}) as
\begin{align}
&\mathbb{E}[J(\bm{\theta}_{t})-J(\bm{\theta}_{t+1})|\mathcal{F}_{t}]\notag\\
\leq&-\frac{\alpha}{2}\|\nabla_{\bm{\theta}}J(\bm{\theta}_{t})\|^{2}+\alpha\big\|\nabla_{\bm{\theta}}J(\bm{\theta}_{t})\big\|\big(\mathrm{Bias}^{\mathrm{(i)}}_{t}+\mathrm{Bias}^{\mathrm{(ii)}}_{t}\notag\\
&+\mathrm{Bias}^{\mathrm{(iii)}}_{t}\big)+\frac{L\alpha^{2}}{2}\big(\mathrm{Var}^{\mathrm{(ii)}}_{t}+\mathrm{Var}^{\mathrm{(iii)}}_{t}+\mathrm{Var}^{\mathrm{(iv)}}_{t}\notag\\
&+\mathrm{Var}^{\mathrm{(v)}}_{t}\big).\label{finalthetheorem3-2}
\end{align}
\par
Next, we will consider two different cases.
\par
1) When $\|\nabla_{\bm{\theta}}J(\bm{\theta}_{t})\|$ is large and satisfies $\|\nabla_{\bm{\theta}}J(\bm{\theta}_{t})\|\geq4\big(\mathrm{Bias}^{\mathrm{(i)}}_{t}+\mathrm{Bias}^{\mathrm{(ii)}}_{t}+\mathrm{Bias}^{\mathrm{(iii)}}_{t}\big)$, (\ref{finalthetheorem3-2}) can be rewritten as
\begin{align}
&\mathbb{E}[J(\bm{\theta}_{t})-J(\bm{\theta}_{t+1})|\mathcal{F}_{t}]\notag\\
\leq&-\frac{\alpha}{4}\|\nabla_{\bm{\theta}}J(\bm{\theta}_{t})\|^{2}+\frac{L\alpha^{2}}{2}\big(\mathrm{Var}^{\mathrm{(ii)}}_{t}+\mathrm{Var}^{\mathrm{(iii)}}_{t}+\mathrm{Var}^{\mathrm{(iv)}}_{t}\notag\\
&+\mathrm{Var}^{\mathrm{(v)}}_{t}\big).\label{finalthetheorem3-3}
\end{align}
\par
2) On the other hand, the policy gradient $\|\nabla_{\bm{\theta}}J(\bm{\theta}_{t})\|\leq4\big(\mathrm{Bias}^{\mathrm{(i)}}_{t}+\mathrm{Bias}^{\mathrm{(ii)}}_{t}+\mathrm{Bias}^{\mathrm{(iii)}}_{t}\big)$, (\ref{finalthetheorem3-2}) can further write as
\begin{align}
&\mathbb{E}[J(\bm{\theta}_{t})-J(\bm{\theta}_{t+1})|\mathcal{F}_{t}]\notag\\
\leq&-\frac{\alpha}{2}\|\nabla_{\bm{\theta}}J(\bm{\theta}_{t})\|^{2}+4\alpha\big(\mathrm{Bias}^{\mathrm{(i)}}_{t}+\mathrm{Bias}^{\mathrm{(ii)}}_{t}+\mathrm{Bias}^{\mathrm{(iii)}}_{t}\big)^{2}\notag\\
&+\frac{L\alpha^{2}}{2}\big(\mathrm{Var}^{\mathrm{(ii)}}_{t}+\mathrm{Var}^{\mathrm{(iii)}}_{t}+\mathrm{Var}^{\mathrm{(iv)}}_{t}+\mathrm{Var}^{\mathrm{(v)}}_{t}\big).\label{finalthetheorem3-4}
\end{align}
\par
By taking the maximum of both bounds given in (\ref{finalthetheorem3-3}) and (\ref{finalthetheorem3-4}) and subsequently computing the expectation, we obtain
\begin{align}
&\mathbb{E}[J(\bm{\theta}_{t})-J(\bm{\theta}_{t+1})]\notag\\
\leq&-\frac{\alpha}{4}\mathbb{E}[\|\nabla_{\bm{\theta}}J(\bm{\theta}_{t})\|^{2}]+4\alpha\big(\mathrm{Bias}^{\mathrm{(i)}}_{t}+\mathrm{Bias}^{\mathrm{(ii)}}_{t}+\mathrm{Bias}^{\mathrm{(iii)}}_{t}\big)^{2}\notag\\
&+\frac{L\alpha^{2}}{2}\big(\mathrm{Var}^{\mathrm{(ii)}}_{t}+\mathrm{Var}^{\mathrm{(iii)}}_{t}+\mathrm{Var}^{\mathrm{(iv)}}_{t}+\mathrm{Var}^{\mathrm{(v)}}_{t}\big)\notag\\
\leq&-\frac{\alpha}{4}\mathbb{E}[\|\nabla_{\bm{\theta}}J(\bm{\theta}_{t})\|^{2}]\notag\\
&+\frac{13\alpha B^{2}R^{2}N\big((H+1)\gamma^{H}+2\gamma^{\kappa+1}\big)^{2}}{(1-\gamma)^{4}}\notag\\
&+\Big(12\alpha L^{2}Nd_{\mathrm{tot}}+\frac{\alpha L^{2}(d_{\mathrm{tot}}+4)}{16}\Big)\mu^{2}\notag\\
&\!+\!\Big(\frac{6\alpha R^{2}}{(1-\gamma)^{2}K}\!+\!\frac{96\alpha L^{2}_{\sigma}Nd_{\mathrm{tot}}\log{M}}{M}\!+\!\frac{\alpha L^{2}_{\sigma}\log{M}}{4M}\Big)\frac{1}{\mu^{2}}\label{finalthetheorem3-5}\\
\leq&-\frac{\alpha}{4}\mathbb{E}[\|\nabla_{\bm{\theta}}J(\bm{\theta}_{t})\|^{2}]\notag\\
&\!+\!\frac{13\alpha B^{2}R^{2}N\big((H\!+\!1)\gamma^{H}\!+\!2\gamma^{\kappa\!+\!1}\big)^{2}}{(1\!-\!\gamma)^{4}}\!+\!13\alpha L^{2}Nd_{\mathrm{tot}}\mu^{2}\notag\\
&+\Big(\frac{6\alpha R^{2}}{(1-\gamma)^{2}K}+\frac{97\alpha L^{2}_{\sigma}Nd_{\mathrm{tot}}\log{M}}{M}\Big)\frac{1}{\mu^{2}},\label{finalthetheorem3-6}
\end{align}
where (\ref{finalthetheorem3-5}) follows from Theorem~\ref{thetheoremofpolicygradient} and (\ref{finalthetheorem3-6}) comes from $d_{\mathrm{tot}}\geq1$.
\par
Since  $\mu^{2}=\max\big\{
\frac{R}{(1-\gamma)L\sqrt{KNd_{\mathrm{tot}}}},\frac{3L_{\sigma}}{L}\sqrt{\frac{\log{M}}{M}}\big\}$, we have $13\alpha L^{2}Nd_{\mathrm{tot}}\mu^{2}\geq\frac{6\alpha R^{2}}{K(1-\gamma)^{2}\mu^{2}}$
and $13\alpha L^{2}Nd_{\mathrm{tot}}\mu^{2}\geq\frac{97\alpha L^{2}_{\sigma}Nd_{\mathrm{tot}}\log{M}}{M\mu^{2}}$.
By using these facts, (\ref{finalthetheorem3-6}) can further simplified to
\begin{align}
&\mathbb{E}[J(\bm{\theta}_{t})-J(\bm{\theta}_{t+1})]\notag\\
\leq&-\frac{\alpha}{4}\mathbb{E}[\|\nabla_{\bm{\theta}}J(\bm{\theta}_{t})\|^{2}]\notag\\
&+\!\frac{13\alpha B^{2}R^{2}N\big((H\!+\!1)\gamma^{H}\!+\!2\gamma^{\kappa\!+\!1}\big)^{2}}{(1\!-\!\gamma)^{4}}\!+\!39\alpha L^{2}Nd_{\mathrm{tot}}\mu^{2}.\label{finalthetheorem3-7}
\end{align}
By performing a telescoping summation for (\ref{finalthetheorem3-7}) and rearranging the terms, we obtain the following result:
\begin{align}
&\frac{1}{T}\sum_{t=0}^{T-1}\mathbb{E}[\|\nabla_{\bm{\theta}}J(\bm{\theta}_{t})\|^{2}]\notag\\
\leq&\frac{4\big(\mathbb{E}[J(\bm{\theta}_{0})]-\mathbb{E}[J(\bm{\theta}_{T})]\big)}{T\alpha}\notag\\
&+\frac{52B^{2}R^{2}N\big((H+1)\gamma^{H}+2\gamma^{\kappa+1}\big)^{2}}{(1-\gamma)^{4}}+156L^{2}Nd_{\mathrm{tot}}\mu^{2}\notag\\
=&\frac{96LNRd_{\mathrm{tot}}(d_{\mathrm{tot}}+2)}{(1-\gamma)T}\notag\\
&+\frac{52B^{2}R^{2}N\big((H+1)\gamma^{H}+2\gamma^{\kappa+1}\big)^{2}}{(1-\gamma)^{4}}+156L^{2}Nd_{\mathrm{tot}}\mu^{2}\notag\\
=&\mathcal{O}\Big(\frac{LNRd_{\mathrm{tot}}(d_{\mathrm{tot}}+2)}{(1-\gamma)T}\notag\\
&+\frac{B^{2}R^{2}N\big((H+1)\gamma^{H}+2\gamma^{\kappa+1}\big)^{2}}{(1-\gamma)^{4}}\notag\\
&+\max\Big\{\frac{LR\sqrt{Nd_{\mathrm{tot}}}}{(1-\gamma)\sqrt{K}},LNd_{\mathrm{tot}}L_{\sigma}\sqrt{\frac{\log{M}}{M}}\Big\}\Big),\notag
\end{align}
where the first inequality can be obtained by $\frac{-R}{1-\gamma}\leq J(\bm{\theta})\leq\frac{R}{1-\gamma}$ and $\alpha=\frac{1}{12LNd_{\mathrm{tot}}(d_{\mathrm{tot}}+2)}$, and the last equality uses $\mu^{2}=\max\big\{\frac{R}{(1-\gamma)L\sqrt{KNd_{\mathrm{tot}}}},\frac{3L_{\sigma}}{L}\sqrt{\frac{\log{M}}{M}}\big\}$.
\end{proof}

%

\bibliographystyle{plain}        

\begin{thebibliography}{99}     





%
%
%
%
%

\bibitem[Dai \emph{et al.}, 2020]{DaiTII2020}
Dai, P., Yu, W., Wen, G., \& Baldi, S. (2020).
Distributed reinforcement learning algorithm for dynamic economic dispatch with unknown generation cost functions.
\emph{IEEE Transactions on Industrial Informatics}, 16(4), 2258-2267.

\bibitem[Li \emph{et al.}, 2020]{LiTCYB2019}
Li, F., Qin, J., \& Zheng, W. (2020).
Distributed $Q$-learning-based online optimization algorithm for unit commitment and dispatch in smart grid.
\emph{IEEE Transactions on Cybernetics}, 50(9), 4146-4156.

\bibitem[Dai \emph{et al.}, 2022]{DaiTCYB2022}
Dai, P., Yu, W., \& Chen, D. (2022).
Distributed Q-learning algorithm for dynamic resource allocation with unknown objective functions and application to microgrid.
\emph{IEEE Transactions on Cybernetics}, 52(11), 12340-12350.

\bibitem[Chu \emph{et al.}, 2020]{ChuTITS2020}
Chu, T., Wang, J., Codec\`{a}, L., \& Li, Z. (2020).
Multi-agent deep reinforcement learning for large-scale traffic signal control.
\emph{IEEE Transactions on Intelligent Transportation Systems},
21(3), 1086-1095.

\bibitem[Wang \emph{et al.}, 2021]{WangTCYB2021}
Wang, X., Ke, L., Qiao, Z., \& Chai, X. (2021).
Large-scale traffic signal control using a novel multiagent reinforcement learning.
\emph{IEEE Transactions on Cybernetics},
51(1), 174-187.

\bibitem[Dai \emph{et al.}, 2024]{DaiTII2024}
Dai, P., Yu, W., Wang, H., \& Jiang, J. (2024).
Applications in traffic signal control: a distributed policy gradient decomposition algorithm.
\emph{IEEE Transactions on Industrial Informatics},
20(2), 2762-2775.



\bibitem[Tan~\emph{et al.}, 2021]{TanTWC2021}
Tan, J., Liang, Y.-C., Zhang, L., \& Feng, G. (2021).
Deep reinforcement learning for joint channel selection and power control in D2D networks.
\emph{IEEE Transactions on Wireless Communications}, 20(2), 1363-1378.

\bibitem[Meng~\emph{et al.}, 2020]{MengTWC2020}
Meng, F., Chen, P., Wu, L., \& Cheng, J. (2020).
Power allocation in multiuser cellular networks: Deep reinforcement learning approaches.
\emph{IEEE Transactions on Wireless Communications}, 19(10), 6255-6267.





\bibitem[Afsar~\emph{et al.}, 2022]{AfsarSurvey2022}
Afsar, M. M., Crump, T., \& Far, B. (2022)
Reinforcement learning based recommender systems: a survey. \emph{ACM Computing Surveys}, 55(7), 1-38.


\bibitem[Lin~\emph{et al.}, 2024]{LinSurvey2022}
Lin, Y., Liu, Y., Lin, F., Zou, L., Wu, P., \& Zeng, W. (2024).
A survey on reinforcement learning for recommender systems.
\emph{IEEE Transactions on Neural Networks and Learning Systems},
35(10), 13164-13184.

\bibitem[Sutton \& Barto, 1998]{Sutton1998}
Sutton, R. S., \& Barto, A. G. (1998).
\emph{Reinforcement Learning: An Introduction}.
Cambridge, MA, USA: MIT Press.

\bibitem[Mnih~\emph{et al.}, 2015]{Minhnature2015}
Mnih, V., Kavukcuoglu, K., Silver, D., Rusu, A. A., Veness, J.,
Bellemare, M. G., Graves, A., Riedmiller, M., Fidjeland, A. K.,
Ostrovski, G., Petersen, S., Beattie, C., Sadik, A., Antonoglou, I., King, H., Kumaran, D., Wierstra, D., Legg, S., \& Hassabis, D. (2015).
Human-level control through deep reinforcement learning. \emph{Nature}, 518(7540), 529-533.


\bibitem[Sha \emph{et al.}, 2022]{Sha2022}
Sha, X., Zhang, J., You, K., Zhang, K., \& Ba\c{s}ar, T. (2022).
Fully asynchronous policy evaluation in distributed reinforcement
learning over networks.
\emph{Automatica}, 136, 110092.


%


\bibitem[Christiano~\emph{et al.}, 2017]{Christiano2017}
Christiano, P. F., Leike, J., Brown, T. B., Martic, M., Legg, S., \& Amodei, D. (2017).
Deep reinforcement learning from human preferences.
In \emph{Advances in Neural Information Processing Systems}, pages 4302-4310.

\bibitem[Wirth~\emph{et al.}, 2017]{Wirth2017}
{\color{blue}Wirth, C., Akrour, R., Neumann, G.,
\& F\"{u}rnkranz, J. (2017).
A survey of preference-based reinforcement learning methods. \emph{Journal of Machine Learning Research}, 18(136), 1-46.}




%
%

\bibitem[Bengs~\emph{et al.}, 2021]{Bengs2021}
Bengs, V., Busa-Fekete, R., Mesaoudi-Paul, A. El, \& H\"{u}llermeier, E. (2021).
Preference-based online learning with dueling bandits: a survey. \emph{Journal of Machine Learning Research}, 22(7), 1-108.

\bibitem[Schulman~\emph{et al.}, 2017]{Schulman2017}
Schulman, J., Wolski, F., Dhariwal, P., Radford, A., \& Klimov, O. (2017).
Proximal policy optimization algorithms. \emph{arXiv preprint arXiv:1707.06347}.

\bibitem[Rafailov~\emph{et al.}, 2023]{Rafailov2023}
Rafailov, R., Sharma, A., Mitchell, E., Manning, C. D., Ermon, S., \& Finn, C. (2023).
Direct preference optimization: your language model is secretly a reward model.
In \emph{Advances in Neural Information Processing Systems}, pages~53728-53741.



\bibitem[Zhang \& Ying, 2024]{Zhangarxiv2024}
Zhang, Q. \& Ying L. (2024).
Zeroth-order policy gradient for reinforcement learning from human feedback without reward inference. \emph{arXiv preprint arXiv:2409.17401}.



\bibitem[Qu \emph{et al.}, 2020a]{QuCLDC2020}
Qu, G., Wierman, A., \& Li, N. (2020).
Scalable reinforcement learning of localized policies for multi-agent networked systems.
In \emph{Proceedings of the Conference on Learning for Dynamics and Control}, pages 256-266.

\bibitem[Qu \emph{et al.}, 2020b]{QuNIPS2020}
Qu, G., Lin, Y., Wierman, A., \& Li, N. (2020).
Scalable multi-agent reinforcement learning for networked systems with average reward.
In \emph{Advances in Neural Information Processing Systems},
pages 2074-2086.


\bibitem[Huang~\emph{et al.}, 2024]{Huang2024}
{\color{blue}Huang. C., Zang, W., Pinciroli, C., Li, Z. J., Banerjee, T., Su, L., \& Liu, R. (2024).
Reactive multi-robot navigation in outdsoor environments through uncertainty-aware active learning of human preference landscape
\emph{arXiv preprint arXiv:2409.16577}.}




\bibitem[Zhao~\emph{et al.}, 2026]{Zhao2026}
{\color{blue}Zhao, C., Cahill, V., \& Dusparic, I. (2026).
Balancing multiple objectives in urban traffic control with reinforcement learning from AI feedback.
\emph{arXiv preprint arXiv:2602.20728}.}

\bibitem[Zhang \emph{et al.}, 2018]{Zhangkaiqing2018}
{\color{blue}Zhang, K., Yang, Z., Liu, H., Zhang, T., \& Ba\c{s}ar, T. (2018).
Fully decentralized multi-agent reinforcement learning with networked agents.
In \emph{Proceedings of the International Conference on Machine Learning}, pages 5872-5881.}


\bibitem[Dai \emph{et al.}, 2025]{DaiTAC2025}
{\color{blue}Dai, P., Mo, Y., Yu, W., \& Ren, W. (2025).
Distributed neural policy gradient algorithm for global convergence of networked multi-agent reinforcement learning.
\emph{IEEE Transactions on Automatic Control}, 70(11), 7109-7124.}




%
%


\bibitem[Sutton \emph{et al.}, 2000]{Sutton2000}
Sutton, R. S., McAllester, D. A., Singh, S. P., \& Mansour, Y. (2000).
Policy gradient methods for reinforcement learning with function approximation.
In \emph{Advances in Neural Information Processing Systems},
pages 1057-1063.





\bibitem[Train, 2009]{Train2009}
Train, K. E. (2009).
\emph{Discrete Choice Methods with Simulation}.
Cambridge, UK: Cambridge University Press.

\bibitem[Greene, 2010]{Greene2010}
Greene, W. H. (2010).
\emph{Modeling Ordered Choices: A Primer}.
New York, NY, USA: Cambridge University Press, 2010.



\bibitem[Du \emph{et al.}, 2024]{Duyihan2024}
{\color{blue}Du, Y., Winnicki, A., Dalal, G., Mannor, S., \& Srikant, R. (2024).
Exploration-driven policy optimization in rlhf: Theoretical insights on efficient data utilization.
\emph{arXiv preprint arXiv:2402.10342}.}







\bibitem[Lin \emph{et al.}, 2021]{LinNIPS2021}
Lin, Y., Qu, G., Huang, L., \& Wierman, A. (2021).
Multi-agent reinforcement learning in stochastic networked systems.
In \emph{Advances in Neural Information Processing Systems},
pages~7825-7837.

\bibitem[Zhou \emph{et al.}, 2023]{Zhou2023UAI}
Zhou, Z., Chen, Z., Lin, Y., \& Wierman, A. (2023).
Convergence rates for localized actor-critic in networked markov potential games.
In \emph{Proceedings of the Conference on Uncertainty in Artificial Intelligence}, pages~2563-2573.

\bibitem[Zhang \emph{et al.}, 2022]{Zhangrunyu2022NIPS}
Zhang, R., Mei, J., Dai, B., Schuurmans, D., \& Li, N. (2022).
On the global convergence rates of decentralized softmax gradient play in markov potential games.
In \emph{Advances in Neural Information Processing Systems},
pages~1923-1935.


\bibitem[Ying \emph{et al.}, 2023]{Ying2023NIPS}
Ying, D., Zhang, Y., Ding, Y., Koppel, A., \& Lavaei, J. (2023).
Scalable primal-dual actor-critic method for safe multi-agent rl with general utilities.
In \emph{Advances in Neural Information Processing Systems},
pages 36524-36539.


\bibitem[Li \emph{et al.}, 2025]{Li2025ICLR}
Li, W., Liu, J., \& Wei, K. (2025).
$\phi$-update: a class of policy update methods with policy convergence guarantee.
In \emph{Proceedings of the International Conference on Learning Representations}.


\bibitem[Kim \emph{et al.}, 2024]{Kim2024arxiv}
Kim, D., Lee, J., Park, J., \& Seo, M. (2024).
How language models extrapolate outside the training data: a case study in textualized gridworld.
\emph{arXiv preprint arXiv:2406.15275}.

\bibitem[Rosenthal, 1970]{Rosenthal1970}
Rosenthal, H. P. (1970).
On the subspaces of $L_{p}$ ($p>2$) spanned by sequences of independent random variables. \emph{Israel Journal of Mathematics}, 8(3), 273-303.













%
%





%






\end{thebibliography}

\end{document}